\documentclass[12pt,reqno,a4paper,oneside]{amsart}    

\usepackage[totalwidth=450pt, totalheight=717pt]{geometry}

\usepackage{color}
\usepackage{amsmath, amssymb}
\usepackage{amsthm}

\usepackage{accents}     

\usepackage{bbm}         

\usepackage{ifthen}       

\usepackage{mathrsfs}    
\renewcommand\mathcal\mathscr  

\newcommand{\look}[1]{\rule{5ex}{1em}\textbf{*}
    \footnote{ #1 }}
\newcommand{\lookO}[1]{\rule{5ex}{1em}\textbf{*}
    \footnote{\textbf{Olaf:} #1 }}
\newcommand{\lookRi}[1]{\rule{5ex}{1em}\textbf{*}
    \footnote{\textbf{Ricardo:} #1 }}
\newcommand{\lookRa}[1]{\rule{5ex}{1em}\textbf{*}
    \footnote{\textbf{Rainer:} #1 }}
\newcommand{\hiddenfootnote}[1]{}

\ifthenelse{\isundefined \draft }
{\renewcommand{\look}[1]{}%
 \renewcommand{\lookO}[1]{}%
 \renewcommand{\lookRi}[1]{}%
 \renewcommand{\lookRa}[1]{}%
 \renewcommand{\hiddenfootnote}[1]{}%
}{\usepackage[notref,notcite]{showkeys}} 

\usepackage{comment, color}

\ifthenelse{\isundefined \colorsforreferees }
{\usepackage{marginnote}   
 \newcommand{\colour}[1][]{}
}
{\usepackage{marginnote}   
 \newcommand{\colour}[1][]{%
  \marginnote{\textbf{\color{blue} \small #1}}\color{red}}
}
\numberwithin{equation}{section}

\newcounter{myenumi}
\newenvironment{myenumerate}[1]{
\begin{list}{\indent(\themyenumi) }
  {\renewcommand{\themyenumi}{#1{myenumi}}
    \usecounter{myenumi}
    \setlength{\topsep}{0em}
    \setlength{\itemsep}{0em}
    \setlength{\leftmargin}{0em}
    \setlength{\labelwidth}{0em}
    \setlength{\labelsep}{0em}}
  }
  {
  \end{list}
  }
\newcommand{\itemref}[1]{\eqref{#1}}

\theoremstyle{plain}            
\newtheorem{theorem}{Theorem}[section]

\newtheorem{proposition}[theorem]{Proposition}
\newtheorem{lemma}[theorem]{Lemma}
\newtheorem{corollary}[theorem]{Corollary}

\theoremstyle{definition}       
\newtheorem{definition}[theorem]{Definition}
\newtheorem{assumption}[theorem]{Assumption}
\newtheorem{example}[theorem]{Example}

\newtheorem{remark}[theorem]{Remark}

\newcommand{\Sec}[1]{Section~\ref{sec:#1}}

\newcommand{\App}[1]{Appendix~\ref{app:#1}}

\newcommand{\Eq}[1]{Eq.~\eqref{eq:#1}}

\newcommand{\Thm}[1]{Theorem~\ref{thm:#1}}

\newcommand{\Ex}[1]{Example~\ref{ex:#1}}

\newcommand{\Lem}[1]{Lemma~\ref{lem:#1}}

\newcommand{\Cor}[1]{Corollary~\ref{cor:#1}}

\newcommand{\Prp}[1]{Proposition~\ref{prp:#1}}

\newcommand{\Rem}[1]{Remark~\ref{rem:#1}}

\newcommand{\Remenum}[2]{Remark~\ref{rem:#1}~(\ref{#2})}

\newcommand{\Def}[1]{Definition~\ref{def:#1}}

\newcommand{\Ass}[1]{Assumption~\ref{ass:#1}}
\newcommand{\Asss}[2]{Assumptions~\ref{ass:#1} and~\ref{ass:#2}}

\newcommand{\abs}[2][{}]{\lvert{#2}\rvert_{{#1}}}    
\newcommand{\abssqr}[2][{}]{\lvert{#2}\rvert^2_{#1}} 
\newcommand{\bigabs}[2][{}]{\bigl\lvert{#2}\bigr\rvert_{#1}}     
\newcommand{\bigabssqr}[2][{}]{\bigl\lvert{#2}\bigr\rvert^2_{#1}}
\newcommand{\Bigabs}[2][{}]{\Bigl\lvert{#2}\Bigr\rvert_{#1}}     

\newcommand{\normsymb}{|\!|}
\newcommand{\triplenormsymb}{|\!|\!|}
\newcommand{\bignormsymb}[1]{#1|\!#1|}

\newcommand{\norm}[2][{}]{\normsymb{#2}\normsymb_{{#1}}}    
\newcommand{\normsqr}[2][{}]{\normsymb{#2}\normsymb^2_{#1}} 
\newcommand{\bignorm}[2][{}]{\bignormsymb{\bigl}{#2}\bignormsymb{\bigr}_{#1}}
\newcommand{\bignormsqr}[2][{}]{\bignormsymb{\bigl}{#2}%
                                \bignormsymb{\bigr}^2_{#1}}
\newcommand{\Bignorm}[2][{}]{\bignormsymb{\Bigl}{#2}\bignormsymb{\Bigr}_{#1}}

\newcommand{\triplenorm}[2][{}]{\triplenormsymb{#2}\triplenormsymb_{{#1}}}
\newcommand{\triplenormsqr}[2][{}]{\triplenormsymb{#2}\triplenormsymb^2_{{#1}}}

\newcommand{\iprod}[3][{}]{\langle{#2},{#3}\rangle_{#1}}  
\newcommand{\bigiprod}[3][{}]{\bigl\langle{#2},{#3}\bigr\rangle_{#1}}

\newcommand{\set}[2]{\{ \, #1 \, | \, #2 \, \} }      
\newcommand{\bigset}[2]{\bigl\{ \, #1 \, \bigl|\bigr. \, #2 \, \bigr\} }

\newcommand{\map}[3]{ #1 \colon #2 \longrightarrow #3}    
\newcommand{\embmap}[3]{ #1 \colon #2 \hookrightarrow #3} 

\newcommand{\bd}  {\partial}             
\newcommand{\clo}[1]{\overline{{#1}}} 

\newcommand{\conj}[1]{\overline {{#1}}}  

\newcommand{\restr}[1]{{\restriction}_{#1}} 

\DeclareMathOperator{\const}  {const}
\newcommand{\dd}    {\, \mathrm d}    
\DeclareMathOperator{\dist}   {dist}
\DeclareMathOperator{\dom}    {Dom}
\DeclareMathOperator{\ran}    {Ran}
\DeclareMathOperator{\Ker}    {Ker}  
\DeclareMathOperator{\inj}    {inj}  
\DeclareMathOperator{\Ric}    {Ric}

\DeclareMathOperator{\sgn}    {sgn}
\DeclareMathOperator{\supp}   {supp}
\DeclareMathOperator{\vol}    {vol}
\DeclareMathOperator{\dvol}    {d\, vol}

\newcommand{\de} {\mathord{\mathrm d}} 

\newcommand{\strong}          {\mathrm s}    
\DeclareMathOperator{\slim}   {\strong\text{-}lim}  

\newcommand{\eps}{\varepsilon} 
\renewcommand{\phi}{\varphi}   
\renewcommand{\rho}{\varrho}   
\renewcommand{\theta}{\vartheta}

\newcommand{\R}{\mathbb{R}} 
\newcommand{\C}{\mathbb{C}} 
\newcommand{\N}{\mathbb{N}} 
\newcommand{\Sphere}{\mathbb{S}} 

\newcommand{\e}{\mathrm e}  
\newcommand{\im}{\mathrm i} 

\newcommand{\wt}{\widetilde}           
\newcommand {\qf}[1]{\mathfrak{#1}}    

\newcommand{\HS}{\mathcal H}           

\newcommand{\Sobsymb} {\mathsf H}      
\newcommand{\Sobnsymb} {\ring{\mathsf H}}   
\newcommand{\SobWsymb}{\mathsf W}      
\newcommand{\SobnWsymb}{\ring{\mathsf W}}   

\newcommand{\Contsymb} {\mathsf C}     
\newcommand{\Lsymb}    {\mathsf L}     

\newcommand{\TRspace}{\BdOpsymb_1}
\newcommand{\HSspace}{\BdOpsymb_2}

\newcommand{\Sobspace}[1][1]{\Sobsymb^{#1}} 
\newcommand{\Sobnspace}[1][1]{\Sobnsymb^{#1}} 

\newcommand{\SobWspace}[2][p]{\SobWsymb_{#1}^{#2}}  
\newcommand{\SobnWspace}[2][p]{\SobnWsymb_{#1}^{#2}}  
\newcommand{\Contspace}[1][{}]{\Contsymb^{#1}}     
\newcommand{\Lpspace}[1][p]    {\Lsymb_{#1}}     
\newcommand{\Lsqrspace}    {\Lpspace[2]}     

\newcommand{\BdOpsymb} {\mathcal B}       
\newcommand{\BdOp}[2][{}]{\BdOpsymb_{#1}({#2})}

\newcommand{\Ci} [2][{}]{\Contspace [\infty]_{#1} ({#2})}
\newcommand{\Cci}[1]{\Ci[\mathrm c]{#1}}
\newcommand{\Cont}[2][{}]{\Contspace[#1]({#2})}

\newcommand{\Schwartz}[1]{\mathcal S(#1)}
\newcommand{\Lp}[2][p]{\Lpspace [#1]({#2})} 
\newcommand{\Lploc}[2][p]{\Lpspace [#1,\mathrm{loc}]({#2})}

\newcommand{\Lsqr}[2][{}]{\Lsqrspace^{#1}({#2})} 
\newcommand{\Sob}[2][1]{\Sobspace [#1]({#2})}         
\newcommand{\Sobn}[2][1]{\Sobnspace [#1]({#2})}  

\newcommand{\Sobx}[3][1]{\Sobspace [#1]_{{#2}}({#3})} 
\newcommand{\SobW}[3][p]{\SobWspace[#1]{#2}(#3)} 

\newcommand{\SobnW}[3][p]{\SobnWspace[#1]{#2}(#3)} 

\newcommand{\Met}{\mathsf{Met}}  

\newcommand{\Sesq}{\mathsf{Sesq}} 

\newcommand{\rH}[1]{r_{#1}}  

\newcommand{\unifL}[1]{\inf\nolimits_{#1}} 

\newcommand {\ini}{\mathrm{ini}} 
\newcommand {\odd}{\mathrm{odd}}
\newcommand {\loc}{\mathrm{loc}}
\newcommand {\dec}{\mathrm{dec}}
\newcommand {\ac} {\mathrm {ac}}

\newcommand{\pot}{{w}} 

\newcommand{\mybigoplus}{\bigoplus} 
\newcommand{\Mscr} {{\mathcal M}}

\newcommand{\Uscr} {{\mathcal U}}

\newcommand{\slimpm}{\slim_{t \to \pm \infty}}

\newcommand{\slimm}{\slim_{t \to - \infty}}

\begin{document}

\title[Scattering on Manifolds with Ends]{On Open Scattering
  Channels\\ for Manifolds with Ends}

\author{Rainer Hempel}
\address{Institute for Computational Mathematics, TU Braunschweig,
  Pockelsstr. 14, 38106 Braunschweig, Germany}
\email{r.hempel@tu-bs.de}

\author{Olaf Post}
\address{Department of Mathematical Sciences,
  Durham University, South Road, Durham, DH1 3LE, United Kingdom}
\email{olaf.post@durham.ac.uk}

\author{Ricardo Weder} %

\address{Departamento de F\'{\i}sica Matem\'atica.  Instituto de
  Investigaciones en Matem\'aticas Aplicadas y en Sistemas,
  Universidad Nacional Aut\'onoma de M\'exico, Apartado Postal
  20--726, M\'exico, D.F.  01000.} %
\email{weder@unam.mx}


\keywords{Laplacian on Riemannian manifolds, scattering matrix,
  harmonic radius} \subjclass[2000]{58J50, 34P25, 37A40, 81U99}

\begin{abstract}
  In the framework of time-dependent geometric scattering theory, we
  study the existence and completeness of the wave operators for
  perturbations of the Riemannian metric for the Laplacian on a
  complete manifold of dimension $n$. The smallness condition for the
  perturbation is expressed {\colour (intrinsically and coordinate
    free)} in purely geometric terms using the harmonic radius;
  therefore, the size of the perturbation can be controlled in terms
  of local bounds on the {\colour[A3] injectivity radius} and the
  Ricci-curvature.  As an application of these ideas we obtain a
  stability result for the scattering matrix with respect to
  perturbations of the Riemannian metric. This stability result
  implies that a scattering channel which interacts with other
  channels preserves this property under small perturbations.
\end{abstract} 

\maketitle

\setcounter{tocdepth}{1}
\tableofcontents

%
\section{Introduction}
\label{sec:intro}
%

The first fundamental problem in multi-channel scattering theory is to
establish the existence and the (asymptotic) completeness of the wave
operators.  These questions are currently quite well understood in
various situations including the case of the $N$-body problem in
quantum mechanics, multi-channel scattering in perturbed acoustic and
electromagnetic wave guides, and scattering on manifolds with ends;
cf., e.g., \cite{yafaev:10}, \cite{dg:97}, \cite{weder:91}, and the
literature discussed at the end of this introduction.  Roughly
speaking, asymptotic completeness in multi-channel scattering means
the following: as time goes to $\pm \infty$, any scattering state
decays into a number of states living in subsystems (channels); these
subsystems then evolve according to a simpler reference dynamics, like
clusters of particles in the quantum mechanical case, radiation and
guided modes for perturbed wave guides, and components that travel
into the various ends of a manifold.  However, given an initial state
belonging to a particular channel (as time goes to $-\infty)$,
asymptotic completeness does not tell us into which channels our state
will decay as time goes to $+\infty$, or, put differently, which
subsystems will actually be non-zero. We are therefore led to ask
which channels are \emph{open} to an initial state belonging to a
particular channel in the past ($t \to - \infty$). Clearly, one would
expect that two scattering channels will be open to each other unless
a particular obstruction prevents the decay from one into the other;
put differently, two channels should be open to each other in some
generic sense.  It appears, though, that there are no general methods
in mathematical multi-channel scattering which would allow to prove
such a result. 

\subsection{Open scattering channels}
\label{sec:intro.open.channels}
As a first step in the analysis of this issue, the present paper
studies the interaction of the channels in geometric scattering theory
where the dynamics is given by the Laplacian on a complete
$n$-dimensional Riemannian manifold with a finite number of ends.  Any
geometric end gives rise to a scattering channel provided the
corresponding decoupled part of the Laplacian has a non-zero
absolutely continuous part.  One of our main results (cf.\ \Thm{main2}
and \Cor{main2}) roughly says the following: Suppose that the $i$-th
scattering channel is open to the $k$-th channel in the sense that the
channel scattering operator $S_{ik}$ for these channels satisfies
 \begin{equation}
  \label{eq:ch.open}
  S_{ik}\ne 0;  
\end{equation}
then the same property~\eqref{eq:ch.open} holds for small
perturbations of the metric.  In other words, we derive a
\emph{stability theorem} for property~\eqref{eq:ch.open}.  The
smallness of the perturbation is expressed in geometric terms that
involve the \emph{harmonic radius} $\rH \Mscr (x)$ at a point $x \in
M$ of a Riemannian manifold $\Mscr=(M,g)$, defined as in the work of
Anderson and Cheeger~\cite{anderson-cheeger:92}.  According
to~\cite{anderson-cheeger:92} and~\cite{hebey-herzlich:98}, $\rH \Mscr
(x)$ depends only on (local) lower bounds for the radius of
injectivity and the Ricci curvature.  Note that we do not need to
require any particular structure of the unperturbed manifolds or its
ends.  As explained at the end of \Sec{non-zero.transm}, the
property~\eqref{eq:ch.open} is symmetric in $k$ and $i$, i.e., $S_{ik}
\ne 0 \Leftrightarrow S_{ki} \ne 0$. %
{\colour[A2] 
  Some comments on the notion of openness of scattering channels can
  be found at the end of \Sec{non-zero.transm} in
  \Rem{when.channels.closed}.%
}

\subsection{An intrinsic trace class perturbation result}
\label{sec:intro.tr.class.pert}
In preparation for the above analysis we derive a rather general
theorem which establishes existence and completeness of the wave
operators for a pair of Laplacians on a manifold $M$ with two
different metrics. This result generalizes Theorem~0.1 of M\"uller and
Salomonsen~\cite{mueller-salomonsen:07} in several directions and may
be of independent interest.

{\colour[A1] 
Although the definitions are somewhat involved, let us try and
give a description of our basic construction:

For $x \in M$ and quasi-isometric metrics $g_1$ and $g_2$ (see
\Def{quasi}), let $A(x)$ be the endomorphism from $T^*_xM$ to itself
defined by
\begin{equation*}
  g_2(x)(\xi,\zeta) = g_1(x)(A(x) \xi, \zeta),
  \qquad  \xi, \zeta \in T_x^*M.
\end{equation*}
We denote the (positive) eigenvalues of $A(x)$ by $\alpha_k(x)$, $k=1,
\ldots, n$.  Then our distance function is defined as
\begin{equation*}
  \wt d(g_1,g_2)(x) 
  := 2 \sinh \Bigl(\frac n 4 \cdot \max_k \abs{\ln \alpha_k(x)} \Bigr)
    = \max_k \bigabs{\alpha_k(x)^{n/4}-\alpha_k(x)^{-n/4}}.
\end{equation*}
Note that $g_1$ and $g_2$ are quasi-isometric if and only if 
\begin{equation*}
  \wt d_\infty(g_1, g_2) := \sup_{x \in M} \wt d(g_1, g_2)(x)<\infty.
\end{equation*}
(The tilde $\wt \cdot$ here and below indicates that the distance
functions $\wt d$, $\wt d_\infty$ etc.\ only satisfy a weaker version
of the triangle inequality, see \App{dist.met}.)  Of equal importance
is a distance function in form of a weighted integral,
\begin{equation}
  \wt d_1(g_1, g_2)
   := \int_M  \wt d(g_1, g_2)\cdot r_0^{-(n+2)}
  \cdot (1+\rho_{g_2,g_1}) \dvol_{g_1},
\end{equation}
where $\wt d(g_1,g_2)$ is the pointwise distance introduced above,
$\map{r_0} M {(0,1]}$ is a continuous function (in practice $r_0(x)$
is a common lower bound for the harmonic radii of $g_1$ and $g_2$ at
the point $x \in M$), and $\rho_{g_2,g_1}$ is the density of
$\dvol_{g_2}$ with respect to $\dvol_{g_1}$.  It is a key element of
our analysis that the trace-class condition for relative scattering
theory with respect to the metrics $g_1$ and $g_2$ is satisfied
provided $\wt d_1(g_1,g_2)$ is finite.  By a basic result of Birman
and Belopol'skii (\cite[Thm.~XI.13]{reed-simon-3} or
\Thm{birman-belopolskii} below), the trace class condition then
implies existence and completeness of the wave operators. Passing from
our weighted integral condition to the existence and completeness of
the wave operators is almost immediate, and thus the finiteness of
$\wt d_1$ is a very natural, intrinsic condition.  Also note that we
express our perturbations in terms of quadratic forms to keep the
assumptions minimal.} %

In the above construction, we use harmonic coordinates in conjunction
with elliptic regularity theory in $\Lpspace$ as
in~\cite{anderson-cheeger:92} to obtain estimates for the Green's
function and its first order derivatives. These estimates are then
employed to verify the trace class condition that is required in the
Birman-Belopol'skii theorem.  A similar approach can be found
in~\cite{wed:84} in the case of higher-order operators in domains with
infinite boundary.

{\colour[A1] 
Let us note as an aside that the distance $\wt d_1^*(g_1,g_2)$ of
eqn.~\eqref{eq:eq.q-met.d1}, which is defined as $\wt d_1(g_1,g_2)$
above but without the factor $1+\rho_{g_2,g_1}$, can be computed more
or less explicitly in some particularly simple cases. In fact, in
\Rem{computation-of-metric} we consider the case of two
quasi-isometric Riemannian metrics $g_1$ and $g_2$ on $M = \R \times
\Sphere^{n-1}$ of the form $g_i = \dd s^2 + r_i(s)^2 \dd
g_{\Sphere^{n-1}}$, for $i = 1,2$, where the functions $r_i$ have to
satisfy some natural conditions.  Here one obtains
\begin{equation*}
  \wt d_1^*(g_1,g_2) =  \omega_{n-1} \int_{-\infty}^\infty  
  \Bigabs{\Bigl(\frac{r_2}{r_1}\Bigr)^{n/2}-
    \Bigl(\frac{r_1}{r_2}\Bigr)^{n/2}}
  \frac 1 {(\min\{1,r_1,r_2\})^{n+2}} \dd s, 
\end{equation*}
where $\omega_{n-1}$ denotes the volume of the $(n-1)$-sphere, and
$\wt d_1^*(g_1,g_2) \le \wt d_1(g_1,g_2) \le c \wt d_1^*(g_1,g_2)$ for
some constant $c \ge 1$ depending only on the quasi-isometric distance
$\wt d_\infty(g_1,g_2)$.%
}

\subsection{Structure of the paper}
\label{sec:intro.structure}
This paper is organized as follows. \Sec{prelim} introduces some basic
definitions concerning scattering in a two-Hilbert space setting,
Sobolev spaces on Riemannian manifolds, and the harmonic radius
according to~\cite{anderson-cheeger:92}.

\Sec{ex.wo} presents a first main result, \Thm{main1}, which
establishes the existence and completeness of the wave operators for
the Laplacian on a Riemannian manifold with respect to perturbations
of the metric tensor.  The trace class condition required in the
Birman-Belopol'skii theorem can be verified under fairly general and
simple conditions that depend on (local) lower bounds for the Ricci
curvature and the injectivity radius given in
eqn.~\eqref{eq:lower.bd}.  Note that we do not need any assumptions on
the derivatives of the curvature tensor nor do we need to control the
derivative of the relative perturbation.

In \Sec{mfds.ends} we introduce a class of Riemannian manifolds with
ends where we discuss the Laplacian $H$ and a decoupled version
$H_\dec$. It is well known that, under mild conditions, the wave
operators for the pair $(H,H_\dec)$ exist and are complete (cf., e.g.,
Carron~\cite{carron:02}). This allows us to define the scattering
operator and the scattering matrix for the pair $(H, H_\dec)$ in a
natural way.


In \Sec{non-zero.transm} we arrive at our second main result,
\Thm{main2}, which establishes strong continuity of the scattering
operator under perturbations of the metric that are small at infinity.
More precisely, we allow for perturbed metrics $g$ which are
quasi-isometric to the given metric $g_0$ and enjoy roughly the same
bounds for the injectivity radius and the Ricci curvature as $g_0$.
Furthermore, the perturbation has to satisfy a trace class condition
on each end, similar to the condition required in \Thm{main1}. As a
direct consequence, we then obtain a continuity result for the
scattering matrix which implies, in particular, that scattering
channels which are open for the metric $g_0$ will also be open for
metrics $g$ that are close to $g_0$ in the sense explained above.


Some simple examples are discussed in \Sec{example}. We restrict our
attention to manifolds $(M,g)$ with two ends and $M = \R \times
\Sphere^{n-1}$. We first give examples for \Thm{main1} where neither
the unperturbed nor the perturbed metric enjoy rotational symmetry
(however, a surface of revolution is used for the sake of comparison
to obtain a lower bound for the {\colour[A3] injectivity radius}).
Two examples illustrating \Cor{main2} have a surface of revolution as
unperturbed manifold while the perturbed manifolds may be more
general. For simplicity, we restrict our attention to the case where
one end is a horn while the other end is asymptotically Euclidean.  We
show by standard techniques that condition~\eqref{eq:ch.open} holds
for the rotationally symmetric case and determine a class of
admissible perturbations of the metric for which
property~\eqref{eq:ch.open} is preserved.  Obtaining suitable lower
bounds for the radius of injectivity is a cumbersome obstacle and we
have been happy to use a comparison theorem
of~\cite{mueller-salomonsen:07} which, however, requires a global
bound on the sectional curvature.  These issues are discussed in
\App{inj.rad}.


The paper comes with three more Appendices. \App{dist.met} contains a
coordinate free way of measuring the distance between two Riemannian
metrics $g$ and $g_0$ on a manifold.  \App{ptw.bd} gives some details
on how to apply elliptic regularity theory
(\cite{gilbarg-trudinger:83}) to the Laplacian in harmonic
coordinates; here we mostly follow~\cite{anderson-cheeger:92}.
\App{line} provides some (actually rather standard) material
concerning scattering on the line for $-\dd^2/\dd x^2$ and $-\dd^2/\dd
x^2 + \pot$ for short-range potentials $\pot$.

{\colour[A4]
\subsection{Notes on the literature}
\label{sec:intro.lit}

 There exists an extensive literature dealing with the
  major issues in Riemannian scattering, most notably the existence
  and completeness of the wave operators, absence of singular
  continuous spectrum, absence of embedded eigenvalues, counting of
  resonances, and the construction of a reference dynamics from the
  geometry. A part of this work was done from a starting point in
  mathematical physics, while other groups rather originate in
  differential geometry, like the school of R.~Melrose who introduced
  the important and fruitful concept of \emph{scattering manifolds}.
  Closest to our work is the recent paper by M\"uller and
  Salomonsen~\cite{mueller-salomonsen:07} mentioned above; Theorem~0.1
  of~\cite{mueller-salomonsen:07} yields existence and completeness of
  the wave operators for perturbations of the curvature tensor. This
  appears to be the first coordinate-free perturbative result. Our
  \Thm{main1} is stronger than their Theorem~0.1; cf.\
  \Rem{mueller-salomonsen} for a detailed comparison.
%

  More detailed information on the scattering operator can be obtained
  if one assumes that the ends possess some additional structure, often
  expressed in terms of a coordinate system which is globally defined
  on each end. Then advanced analytical tools from micro-local analysis, 
  pseudo-differential operators or Fourier integral operators can 
  be used to gain rather precise information on the behavior of 
  wave packets and the scattering matrix. Let us  
  highlight some of these developments.  A vast body of work has
  been devoted to scattering manifolds and the scattering on hyperbolic ends 
  (the survey by Perry~\cite{perry:07} lists some 170~references). 
  From its inception, the study of scattering manifolds in the sense of 
  Melrose has produced a constant stream of papers, devoted to various 
  aspects and issues. We have to restrict ourselves to a small selection 
  which nonetheless, as we hope, displays some of the variety and depth 
  of what has been achieved by various groups in two decades:   
  Melrose~\cite{melrose:95}, 
  Datchev~\cite{datchev:09}, 
  Guillop\'e and Zworski~\cite{guillope-zworski:97},  
  Hassell and Wunsch~\cite{hassell-wunsch:05}, \cite{hassell-wunsch:08},  
  Ito and   Nakamura~\cite{kenichi-nakamura:10}, \cite{kenichi-nakamura:12}, 
  Ito   and Skibsted~\cite{kenichi-skibsted:13b},
  Mazzeo and Vasy~\cite{mazzeo-vasy:05},  
  Melrose, S\'a Barreto and Vasy~\cite{melrose-barreto-vasy:13}, 
  S\'a Barreto~\cite{sabarreto:05},  
  The papers~\cite{guillope-zworski:97} and Wunsch and
  Zworski~\cite{wunsch-zworksi:00} deal with the counting of
  resonances on scattering manifolds. Early on, the Mourre-method has
  been applied by Froese and Hislop~\cite{froese-hislop:89} and others
  to exclude singular continuous spectrum, while
  Donnelly~\cite{donnelly:99}, Kumura~\cite{kumura:10} study
  asymptotically Euclidean manifolds where they exclude singular
  continuous spectrum and embedded eigenvalues as well; cf.~also the
  recent paper by Ito and Skibsted~\cite{kenichi-skibsted:13}.  The
  decay of solutions of the Schr\"odinger equation on asymptotically
  conical ends has recently been studied by Schlag, Soffer and
  Staubach~\cite{sss:10,sss:10b}.  In many instances progress is
  achieved via the (non-perturbative) construction of ``natural''
  dynamics on the ends, starting from the classical geodesic flow;
  cf., e.g., Herbst and Skibsted~\cite{herbst-skibsted} and the most
  recent work of Ito and Skibsted~\cite{kenichi-skibsted:13b} or Ito
  and Nakamura~\cite{kenichi-nakamura:10}, and the references therein.
  From our point of view, some of the non-perturbative methods provide
  natural comparison dynamics which can be used as a reference for
  perturbed systems. This point of view will be taken up again
    in \Rem{ref.op}.
  We emphasize that neither these remarks nor the list of references
  are in any way complete.  }

%
\section{Preliminaries}
\label{sec:prelim}
%
In this section, we introduce basic notation and definitions
concerning scattering in two Hilbert spaces and Laplacians on
manifolds; this material is fairly standard. We then discuss harmonic
coordinates and the harmonic radius in the sense of Anderson and
Cheeger~\cite{anderson-cheeger:92}.

\subsection{Some basic notation}
\label{sec:notation}

Let $U \subset \R^d$ be open. The vector space of infinitely
differentiable functions $\map \phi U \C$ with compact support in $U$
is denoted as $\Cci U$.  For $0 < \alpha \le 1$ and $k \in \N_0$, we
denote as $\Cont[k,\alpha] U$ the space of functions $\map f U \C$
that are $k$-times continuously differentiable with all partial
derivatives of order $k$ being locally H\"older-continuous, and
similarly, we denote as $\Cont[k,\alpha]{\clo U}$ the subspace of
$\Cont[k,\alpha] U$ with all derivatives of order $k$ being
\emph{uniformly} H\"older-continuous functions, as defined
in~\cite[Sec.~4.1]{gilbarg-trudinger:83}.  In the special case $U =
\R^d$, we will also need the Schwartz space $\Schwartz {\R^d}$ of
rapidly decreasing functions.

For $1 \le p < \infty$, the space of (equivalence classes of)
Borel-measurable functions $\map f U \C$ with $\int_U \abs{f(x)}^p \dd
x < \infty$, equipped with the usual norm, is denoted as $\Lp U$;
$\Lp[\infty] U$ is the Banach space of (equivalence classes of
essentially) bounded Borel-functions $\map f U \C$ with the usual
norm. For $1 \le p \le \infty$ and $k \in \N$, we let $\SobW k U$
denote the Sobolev space of all functions $f \in \Lp U$ with
distributional derivatives {\colour[B1] up to order $k$} in $\Lp U$,
equipped with the canonical norm as
in~\cite[Sec.~7.5]{gilbarg-trudinger:83}.  In the special case $p=2$
we write $\Sob[k] U := \SobW[2] k U$, a Hilbert space. The local
spaces $\SobW[p,\loc] k U$, $\Sobx [k] \loc U$ are defined
accordingly. We will also need the subspaces $\SobnW k U$ and
$\Sobn[k] U$ obtained as the closure of $\Cci U$ in the respective
norms.

For linear operators $T$ acting in a Hilbert space $\HS$ we denote by
$\dom T$, $\ran T$, and $\Ker T$ the domain, the range, and the kernel
of $T$, respectively.

\subsection{Scattering in two Hilbert spaces}
\label{sec:two.hs.scatt}
Let $H_1$ and $H_2$ be self-adjoint operators acting in separable
Hilbert spaces $\HS_1$ and $\HS_2$, respectively, and let $J$ be a
bounded operator from $\HS_1$ into $\HS_2$.  We define the wave
operators
\begin{equation}
  \label{eq:2.1}
  W_\pm (H_2, H_1, J)
  =  \slimpm \e^{\im t H_2} J \e^{-\im tH{_1}} P_\ac(H_1),
\end{equation}
provided that the strong limits exist, with $P_\ac (H_1)$ denoting the
projection onto the subspace of absolute continuity of $H_1$.  We say
that the wave operators $W_\pm (H_2, H_1, J)$ are \emph{complete} if
$(\Ker W_\pm (H_2, H_1, J))^\perp = \HS_\ac(H_1)$ and
\begin{equation}
  \label{eq:2.2}
  \clo{\ran W_\pm (H_2, H_1, J)}
  = \HS_\ac (H_2), 
\end{equation}
where $\HS_\ac (H_i)$ denotes the subspace of absolute continuity of
$H_i$; note that since we only assume that $J$ is bounded, $\ran W_\pm
(H_2, H_1, J)$ is not necessarily closed.

A bounded linear operator $\map {T}{\HS_1}{\HS_2}$ ($T \in
\BdOp{\HS_1,\HS_2}$) is said to be \emph{trace class} if $(T^*
T)^{1/2}$ is trace class in $\HS_1$; the corresponding space of trace
class operators is denoted as $\BdOp[1] {\HS_1, \HS_2}$.
Equivalently, $T$ is trace class if $T$ can be factorized as $T=T_2^*
T_1$ with $\map {T_i}{\HS_i} {\HS_0}$ being Hilbert-Schmidt operators
into a third Hilbert space $\HS_0$.  The space of such Hilbert-Schmidt
operators is denoted by $\BdOp[2]{\HS_i, \HS_0}$.

For further basic definitions and results in
two-Hilbert space scattering, we refer to~\cite{reed-simon-3,kato:67}.
Our main result will be based on the Birman-Belopol'skii theorem as
given in~\cite[Thm.~XI.13]{reed-simon-3}:
{\colour[A5] 
\begin{theorem}
  \label{thm:birman-belopolskii}
  For $j=1,2$, let $H_j$ be a self-adjoint and semi-bounded
  operator in a Hilbert space $\HS_j$ with associated quadratic form
  $\qf h_j$ and spectral projectors $E_{\mathbb I}(H_j)$.  Suppose
  that $I \in \BdOp{\HS_1,\HS_2}$ such that
  \begin{myenumerate}{\alph}
    \item
      $I$ has a two-sided bounded inverse;
    \item we have $E_{\mathbb I}(H_2)(H_2I-IH_1)E_{\mathbb I}(H_1) \in
      \BdOp[1]{\HS_1,\HS_2}$ for any bounded interval $\mathbb I$;
    \item the operator $(I^*I-1)E_{\mathbb I}(H_1)$ is compact
      for any bounded interval $\mathbb I \subset \R$;
    \item  $I(\dom \qf h_1) =\dom \qf h_2$.
  \end{myenumerate}
  Then the wave operators $W_\pm(H_2,H_1,I)$ exist, are complete, and
  partially isometric with initial space $\HS_\ac(H_1)$ and final
  space $\HS_\ac(H_2)$.
\end{theorem}
}

\subsection{Manifolds and their Laplacians}
\label{sec:mfds.sob.lap}
Let $M$ be a smooth, complete, oriented, connected mani\-fold of
dimension $n \ge 2$; typically, $M$ will not be compact.  

{\colour[B2] We will define various objects such as Hilbert spaces or
  norms intrinsically without referring to an atlas, in order to make
  the definitions as natural and general as possible.

  Let $T^*M$ denote the cotangent bundle on $M$.  A \emph{Riemannian
    metric} $g$ on $M$ is a smooth family of positive definite
  sesquilinear forms on $T^*M$. The corresponding \emph{Riemannian}
  manifold will be written as $\Mscr = (M,g)$.  As explained in
  \Rem{regularity}, we may weaken the regularity assumption on the
  metric.

  We assume---chiefly for simplicity---that $M$ has no boundary, and
  that $(M,g)$ is complete.  Our results easily extend to many cases
  where the manifold has a boundary.  In this case, we have to specify
  suitable boundary conditions, such as Dirichlet or Neumann, to
  obtain a self-adjoint realization of the Laplacian.}

The metric $g$ naturally induces a \emph{volume measure} on $M$,
denoted by $\dvol_g$.  The corresponding Hilbert space of square
integrable (equivalence classes of) functions on $M$ is denoted by
$\Lsqr \Mscr=\Lsqr{M,g}$ with inner product $\iprod u v := \int_M u
\conj v \dvol_g$ and corresponding norm
\begin{equation*}
  \normsqr[\Mscr] u = \normsqr[\Lsqr \Mscr] u
  = \int_M \abssqr u \dvol_g.
\end{equation*}
Similarly, we denote by $\Lsqr {T^*\Mscr}$ the square integrable
sections in the 
{\colour[B3] Riemannian} cotangent bundle $T^* \Mscr = (T^* M, g)$
with norm
\begin{equation*}
  \normsqr[T^*\Mscr] \omega = \normsqr[\Lsqr {T^*\Mscr}] \omega
  = \int_M \abssqr[g] \omega \dvol_g
\end{equation*}
where $\abssqr[g] \omega = g(\omega,\omega)$ depends on the metric
$g$.  {\colour[B3] We assume that the fibres $T_x^*M$ are
  \emph{complex} vector spaces, and that $g_x$ is a
  \emph{sesquilinear} form on $T_x^*M$, linear in the first and
  anti-linear in the second argument.}  In a coordinate chart $\map
\Phi B U$ with $B \subset M$ and $U \subset \R^n$ open and $(x^1,
\dots, x^n)=\Phi(x)$, $x \in B$, we have the local expression
\begin{equation*}
  \abssqr[g] \omega 
  = \sum_{i,j=1}^n g^{ij} \omega_i \conj \omega_j
\end{equation*}
where $\omega=\sum_i \omega_i \de x^i$ and $g^{ij} = g(\de x^i, \de
x^j)$.  On the \emph{tangent} bundle $T \Mscr$, the metric components
are given by the inverse matrix $(g_{ij})$ with $g_{ij}=g(\partial_i,
\partial_j)$.  Moreover, $\dvol_g = \sqrt {\det g} \dd x^1
{\colour[B4] \cdots} \dd x^n$ is the volume measure in the coordinate
chart, where $\det g$ is the determinant of the matrix $(g_{ij})$.

Let $\de u$ be the exterior derivative of $u$, a section in the
cotangent bundle.  Let $\Cci M$ be the space of smooth functions with
compact support.  We denote by $\Sob \Mscr$ the closure of $\Cci M$
with respect to the norm
\begin{equation*}
  \normsqr[\Sob \Mscr] u
  := \normsqr [\Lsqr \Mscr] u +
     \normsqr [\Lsqr {T^*\Mscr}] {\de u}.
\end{equation*}
We define the operator $\map D {\Sob \Mscr} \Lsqr{T^*\Mscr}$ by $Du :=
\de u$.  Note that by definition of $\Sob \Mscr$, $D$ is a closed
operator.  Similarly, the quadratic form $\qf d$ given by $\qf d(u) :=
\normsqr [\Lsqr {T^*\Mscr}] {\de u}$ and $\dom \qf d := \Sob \Mscr$ is
closed.  We denote the corresponding sesquilinear form obtained via 
the polarization identity by the same symbol $\qf d$.  By the first
representation theorem (cf.~\cite[Thm.~VI.2.1]{kato:66}), there exists
a self-adjoint, non-negative operator $\Delta=\Delta_\Mscr$, the
\emph{Laplacian} of the Riemannian manifold $\Mscr$, satisfying
\begin{equation*}
  \qf d(u,v) = \iprod {\Delta u} v
\end{equation*}
for all $u \in \dom \Delta$ and $v \in \dom \qf d$.  We also write
$\Delta =\Delta_\Mscr = \Delta_g$ in order to stress the dependence on
the metric $g$.  Note that we define the Laplacian as a non-negative
operator---instead of a non-positive operator---as is often the
case in the mathematical physics literature.

\subsection{Harmonic radius}
\label{sec:harm.rad}

We denote by $B(x,r)=B_\Mscr(x,r)$ the open geodesic ball of radius
$r$ around $x$ in the Riemannian manifold $\Mscr=(M,g)$.  A central
role in our analysis is played by the (local) \emph{harmonic radius}
of $\Mscr$.  Roughly speaking, for any point $x \in M$, the harmonic
radius $\rH \Mscr (x)$ at $x$ is the largest radius (less than the
{\colour injectivity radius} at $x$) with the property that there exists a
system of harmonic coordinates in $B_\Mscr(x,\rH \Mscr(x))$. This
means that there is an open set $U \subset \R^n$ and a diffeomorphism
$\map \Phi {B_{\Mscr}(x, \rH \Mscr(x))} {U}$ such that each component
function $\Phi^i$ satisfies $\Delta_g \Phi^i = 0$ in $B_{\Mscr}(x, \rH
\Mscr(x))$.  The actual definition, given below, is a bit more
technical and provides several rather precise estimates.  We mostly
follow the work of Anderson and Cheeger~\cite{anderson-cheeger:92};
cf.\ also~\cite{deturck-kazdan:81, hebey:96, hebey-herzlich:98,
  hebey:99}.  We first define the \emph{harmonic radius at $x \in M$}
as in~\cite{anderson-cheeger:92}:

\begin{definition}
  \label{def:harm.rad}
  For $p \in (n,\infty)$ and $Q \in (1,\infty)$, the
  \emph{$\SobWspace[p] 1$-harmonic radius} of $\Mscr$ at $x \in M$ is
  the largest number $\rH \Mscr(x) = \rH \Mscr(x;p,Q)$ such that there
  is a system of harmonic coordinates in $B_{\Mscr}(x, \rH \Mscr(x))$
  with the property that the metric tensor $(g_{ij})$ in these
  coordinates satisfies
  \begin{gather}
    \label{eq:2.3}
    Q^{-1} (\delta_{ij}) \le (g_{ij})
    \le Q (\delta_{ij}), \qquad
     \intertext{as bilinear forms, and}
    \label{eq:2.4}
    \rH \Mscr(x)^{1 - n/p} \norm[\Lp U]{\partial_k g_{ij}}
    \le Q - 1,
  \end{gather}
  where $U \subset \R^n$ is the domain of the coordinates.
\end{definition}

{\colour[A6] We are mainly interested in $Q \in (1,2]$ close to $1$;
  the actual choice of $Q$ will only become important in \App{ptw.bd}
  (see the text before eqn.~\eqref{eq:b.15}) when we transfer
  estimates in elliptic regularity theory from Euclidean space to the
  manifold). At the same time, we will fix $p := n+1$. Our choice of
  $Q$ in \App{ptw.bd} is not mandatory; other $Q$'s would lead to
  different constants, however.  }

Note that~\eqref{eq:2.4} gives control of the $g_{ij}$ in
$\alpha$-H\"older-norm with exponent $\alpha : = 1 - n/p$.  A fine
point in the way how~\cite{anderson-cheeger:92} use elliptic
regularity theory in conjunction with harmonic coordinates concerns
the modulus of continuity of the $g_{ij}$ which enters the
$\Lpspace$-estimate of~\cite[Thm.~9.11]{gilbarg-trudinger:83}.  By
\cite[Thm.~7.19]{gilbarg-trudinger:83}, the above eqn.~\eqref{eq:2.4}
implies the following H\"older estimate: there exists a constant $C =
C(n,\alpha)>0$ such that for any ball $B$ satisfying $\clo B \subset
U$ we have
\begin{equation}
  \label{eq:gij.hoelder}
      |g_{ij}(y) - g_{ij}(y')| 
      \le {\colour C (Q-1) \rH \Mscr(x)^{-\alpha}} 
      |y -y'|^\alpha, 
      \qquad y, y' \in \clo B, 
\end{equation}
with $\alpha = 1 - n/p$.  In particular, $|g_{ij}(y) - g_{ij}(y')| \le C
(2Q)^\alpha (Q - 1)$, for $y, y' \in \clo B \subset U$. {\colour 
 By eqn.~\eqref{eq:2.3} we also have the elementary estimate 
 $1/Q \le |g_{ij}(y)| \le Q$, for all $i,j$ and all $y \in U$.} 

We need a local lower bound on the harmonic radius in terms of local
lower bounds on the injectivity radius $\inj_\Mscr$ and the Ricci
curvature $\Ric_\Mscr$.  Denote by $\map{\Ric^-_\Mscr} M \R$ the
(pointwise) lowest eigenvalue of $\Ric_\Mscr$ viewed as an
endomorphism on $T^*M$.

For $\delta>0$ and a continuous function $\map f M \R$ denote by
\begin{equation*}
  (\unifL \delta f)(x) := \inf_{y \in B_{\Mscr}(x,\delta)} f(y) 
\end{equation*}
the \emph{$\delta$-homogenized lower bound} of $f$.  Note that if $f$
is Lipschitz-continuous with Lipschitz constant $L>0$ on
$B_{\Mscr}(x,\delta)$ then
\begin{equation}
  \label{eq:unif.inj.lip}
  f(x) - L \delta \le \unifL \delta f (x) \le f(x).
\end{equation}

As in~\cite{anderson-cheeger:92}, we denote by
\begin{equation}
  \label{eq:unif.inj.rad}
  \iota_\Mscr(x) := \sup_{\delta > 0}
     \min \bigl\{\delta, \unifL \delta \inj_\Mscr(x) \bigr\}
\end{equation}
the largest radius $\delta$ for which the injectivity radius at $y \in
B(x,\delta)$ is bounded from below by $\delta$; we call
$\iota_\Mscr(x)$ the \emph{homogenized injectivity radius}.
{\colour[B5] An important technical ingredient in the proof of
  Thm.~0.3 in Anderson and Cheeger~\cite[p.~271]{anderson-cheeger:92}
  consists in the observation that
\begin{equation*}
  \iota_{\mathcal B}(y) = \frac 12 \dist(y,\bd B(x,r)) 
\end{equation*}
for $y \in B(x,r)$ and $r = \inj_{\mathcal B}(x)$, where $\mathcal
B=(B(x,r),g)$.  This fact justifies the complicated definition of
$\iota_\Mscr(x)$. In particular, for $y=x$, we have $\iota_{\mathcal
  B}(x)=r/2$.  We use a slight generalisation of this fact in the
following form that $r \le \inj_\Mscr(x)$ implies $\iota_{\mathcal
  B}(x)\ge r/2$.

This inequality can be easily seen as follows: $r \le \inj_\Mscr(x)$
implies that the exponential map is well-defined in $B(y,r/2)$ for any
$y \in B(x,r/2)$, hence $\iota_{\mathcal B'}(x) \ge r/2$ with
$\mathcal B'=(B(y,r/2),g)$.  The domain monotonicity of the
homogenized injectivity radius implies
\begin{equation}
  \label{eq:hom.injrad2}
  r/2 \le \iota_{\mathcal B'}(x) \le \iota_{\mathcal B}(x).
\end{equation}
}

The harmonic radius is a purely geometric quantity with a lower bound
depending only on lower bounds for the injectivity radius and the Ricci
curvature (cf.~\cite[Thm.~0.3]{anderson-cheeger:92}):
\begin{proposition}
  \label{prp:lower.bd.rh}
  Let $\Mscr=(M,g)$ be a smooth Riemannian manifold and let $\map
  {r_0} M {(0,1]}$ be a continuous function such that the
  homogenized injectivity radius and Ricci curvature satisfy the lower
  bounds
  \begin{equation}
    \label{eq:lower.bd}
    \iota_\Mscr (x) \ge r_0(x)
       \quad\text{and}\quad
    \unifL {r_0(x)} \Ric_\Mscr^-(x) \ge - \frac 1 {r_0(x)^2}
  \end{equation}
  for all $x \in M$.  Then there exists a constant $c=c(n,p,Q) > 0$
  such that the $\SobWspace[p] 1$-harmonic radius of $\Mscr$ at $x \in
  M$ satisfies the lower bound
  \begin{equation*}
    \rH \Mscr (x) \ge c r_0(x),
  \end{equation*}
 for all $x \in M$.
\end{proposition}
{\colour[B5]
  \begin{proof}
    For $x \in M$ given, we apply Thm.~0.3
    of~\cite{anderson-cheeger:92} to the specific choice of $\mathcal
    B=(B,g)$ with $B := B(x,r_0(x))$ (which now replaces the manifold
    $M$ of~\cite{anderson-cheeger:92}). The constant $\lambda$ of
    Thm.~0.3 is then replaced by $1/r_0(x)$; notice that
    $\Ric_\Mscr^-(y) \ge -1/r_0(x)^2$ for all $y \in B$ by the very
    definition of $\inf_\delta$, for $\delta = r_0(x)$. Applying their
    theorem, we obtain that (with constants $c_1$ and $c_2$ depending
    only on $Q$, $n$ and $p$)
    \begin{equation*}
      \rH \Mscr(x) \ge \rH {\mathcal B}(x)
      \ge \min\{c_1 \lambda^{-1}, c_2 \iota_{\mathcal B} (x) \}
      \ge \min\{ c_1 r_0(x), c_2 r_0(x)/2 \} 
      = c r_0(x) 
    \end{equation*}
    with $c := \min\{c_1, c_2/2\}$; here we have
    used~\eqref{eq:hom.injrad2} with $r=r_0(x)$.
 \end{proof}}
{\colour[B5] As an alternative to the above proof of \Prp{lower.bd.rh}
  one could just follow the proof of Theorem~0.3
  in~\cite{anderson-cheeger:92}. Indeed, Thm.~0.3
  in~\cite{anderson-cheeger:92} is of a purely local nature (as noted
  by the authors), and the first step in their proof of Thm.~0.3
  consists in a reduction to geodesic balls.}

Note that Theorem~0.3 in~\cite{anderson-cheeger:92} is purely local
(as noted by the authors). Indeed, the first step in their proof of
Theorem~0.3 consists in a reduction to geodesic balls.

Without loss {\colour[B10] of generality} we may assume in the sequel,
as we have already done in \Prp{lower.bd.rh}, that the function $r_0$,
serving as lower bound for the injectivity radius and the Ricci
curvature, is bounded from above by $1$; this assumption is convenient
in the proofs of \App{ptw.bd}.

For further reference, we use the following notation:
\begin{definition}
  \label{def:harm.rad.met}
  For a continuous positive function $\map {r_0} M {(0,1]}$, we denote
  by $\Met_{r_0}(M)$ the set of smooth metrics $g$ on $M$ that satisfy 
  the lower bounds~\eqref{eq:lower.bd}.
\end{definition}

Let us mention a particularly simple situation where the
\emph{homogenized} injectivity radius and Ricci curvature can be
estimated from below using a \emph{pointwise} lower bound on the
injectivity radius and Ricci curvature itself.

\begin{proposition}
  \label{prp:bd.geo}
  Assume that 
  \begin{equation}
    \label{eq:bd.geo}
    \Ric_{(M,g)}^-(x) \ge -\frac 1{\beta(x)^2}  \quad \text{and} \quad
    \inj_{(M,g)}(x) \ge r(x), 
  \end{equation}
  for all $x \in M$, where $\map {\beta,r} M {(0,1]}$ are
  $\Contspace[1]$-functions enjoying the following properties: $r$ has
  bounded derivative, and $\beta$ satisfies an estimate
  $\abs{\beta'(x)} \le C \beta(x)^3$ for all $x \in M$, for some
  constant $C \ge 0$.

Then the lower bound~\eqref{eq:lower.bd} on the homogenized
injectivity radius and the homogenized Ricci curvature holds with
  \begin{equation*}
     r_0(x):=\min \Bigl\{1,\frac{r(x)}{1+\norm[\infty]{r'}},
                        \frac {\beta(x)}{\sqrt {1+2C}} \Bigr\}, 
  \end{equation*}
  i.e., $g \in \Met_{r_0} M$.
  In particular, if $\beta$ and $r$ are constant, then
  $r_0(x)=\min\{1,r, \beta\}$ can be chosen as a constant function.
\end{proposition}
\begin{proof}
  Note first that $f \ge g$ implies $\unifL \delta f \ge \unifL \delta
  g$.  Applying~\eqref{eq:unif.inj.lip} to the function $-\beta^{-2}$
  with $\delta \in (0,1]$ we obtain
  \begin{equation*}
    \unifL \delta \Bigl(-\frac 1 {\beta^2} \Bigr)(x)
    \ge -\frac 1 {\beta(x)^2} - L \delta 
    \ge -\frac 1 {\beta(x)^2} - 2C
    \ge -\frac {1+2C} {\beta(x)^2}
  \end{equation*}
  with $L = \norm[\infty]{(\beta^{-2})'}\le 2C$ and similarly,
  $\iota_{(M,g)}(x) \ge \min \{\delta, r(x) - \norm[\infty]{r'} \delta
  \}$.  The latter expression is greater than or equal to $\delta$ iff
  $r(x) - \norm[\infty] {r'} \delta \ge \delta$.  This inequality
  yields the inequality $r(x) (1+\norm[\infty] {r'})^{-1} \ge r_0(x)$
  on $r_0(x)=\delta$.
\end{proof}

\begin{remark}
  \label{rem:regularity}
  We may weaken the regularity assumptions on the metric $g$ as
  follows: It is sufficient to assume that $g \in \Contspace
  [1,\alpha]_\loc \subset \SobWspace [p,\loc] 1$, i.e., we assume that
  there is a covering with (for simplicity, smooth) charts, such that
  the metric tensor $(g_{ij})$ in each of these charts is of class
  $\Contspace [1,\alpha]_\loc$.  In this
  case, the Ricci curvature is still defined and $\Ric_{(M,g)} \in
  \Contspace [0,\alpha]_\loc$ (cf.\ the paper~\cite{akklt:03} for a
  detailed discussion of related ideas and results).

  More precisely, we can argue as follows:
  \begin{myenumerate}{\alph}
  \item 
    \label{reg.a}
    The H\"older regularity $g \in \Contspace [1,\alpha]_\loc$ allows
    us to apply the results of deTurck and
    Kazdan~\cite{deturck-kazdan:81}, notably their Lemma~1.2 which
    states that, for any point $x \in M$, there exist harmonic
    coordinate charts of class $\Contspace [2,\alpha]_\loc$ near $x$
    and that, moreover, \emph{all} harmonic charts near $x$ have this
    regularity.  In particular, the metric tensor $g_{ij}$ in any
    harmonic chart has regularity $\Contspace [1,\alpha]_\loc$.

  \item
    \label{reg.b}
    As before, Theorem~0.3 of~\cite{anderson-cheeger:92} yields a
    lower bound for the $\SobWspace[p,\loc] 1$-harmonic radius at $x
    \in M$ as in \Prp{lower.bd.rh}. (Note that~\cite{anderson-cheeger:92} seem to consider smooth $g_{ij}$, but a
    simple approximation argument allows to establish the result of
    Theorem~0.3 in~\cite{anderson-cheeger:92} under the weaker
    assumption $g_{ij} \in \SobWspace[p,\loc] 1$).

  \item
    \label{reg.c}
    By~\itemref{reg.a}, the metric tensor $g_{ij}$ is locally
    Lipschitz and, by~\itemref{reg.b}, the
    estimates~\eqref{eq:2.3}--\eqref{eq:2.4} hold. This quality of the
    $g_{ij}$ is required for an application of elliptic regularity
    theory in $\Lpspace$ to the Laplacian, expressed in harmonic
    coordinates, cf.~\cite{anderson-cheeger:92} and \App{ptw.bd}.
    Note that the Laplacian, written in harmonic coordinates, has no
    first-order terms (cf.~\cite{deturck-kazdan:81},
    \cite{anderson-cheeger:92}), an important simplification.
  \end{myenumerate}
\end{remark}

%
 \section{Existence and completeness of the wave operators} 
\label{sec:ex.wo}
%

We are now going to derive a criterion for the existence and
completeness of the wave operators for the Laplacian of two
(non-compact) manifolds that are close to one another in a suitable
sense. It is our aim to find conditions that only involve geometric
quantities and do not assume a particular structure of the unperturbed
situation. This problem has recently been studied
in~\cite{mueller-salomonsen:07} where the (relative) smallness of the
perturbation at the ends of the manifold is expressed in terms of
bounds on the curvature tensor and its derivatives. Here we propose an
approach which, in several respects, is even closer to the geometry
and, hopefully, even simpler. To this end, we advocate the use of the
(local) harmonic radius (cf.\ \Sec{harm.rad}) as a basic geometric
quantity which can be used to express conditions on the perturbed
metric that translate into trace-class conditions for the difference
of resolvents (or, more precisely, the difference of suitable powers
of resolvents).  Note that the manifolds considered in this section
are more general than what we discuss later on where we will restrict
our attention to manifolds with ends.

Suppose that $M$ is an $n$-dimensional, smooth, oriented manifold with
two metrics $g_1$, $g_2$ such that $\Mscr_k := (M,g_k)$ is a complete
$n$-dimensional Riemannian manifold, for $k=1,2$.  Let us first
compare the corresponding norms defined with respect to $g_1$ and
$g_2$.  In \App{dist.met}, we define a quasi-distance
\begin{equation}
  \label{eq:def.q-dist}
  \wt d_\infty(g_1,g_2)
  := \sup_{x \in M} \wt d(g_1,g_2)(x),
\end{equation}
where
\begin{equation}
  \label{eq:def.q-dist.loc}
  \wt d(g_1,g_2)(x) 
  := 2 \sinh \Bigl(\frac n 4 \cdot \max_k |\ln \alpha_k(x)| \Bigr)
  {\colour = \max_k \bigabs{\alpha_k(x)^{n/4}-\alpha_k(x)^{-n/4}}},
\end{equation}
and where $\alpha_k(x)$ is the $k$-th eigenvalue of the positive
definite endomorphism $A(x) \in \BdOp{T^*_x M}$ given by
\begin{equation}
  \label{eq:the-matrix-A(x)}
    g_2(x) (\xi, \zeta)= g_1(x)(A(x) \xi, \zeta), \qquad
          \xi, \zeta \in T^*_x M.
\end{equation}

Let us give a simple example.  If $g_2=\e^{2\mu}g_1$ for some
(bounded) function $\map \mu M \R$, then $g_2$ is a \emph{conformal
  perturbation} of $g_1$.  In this case, $A=\e^{2\mu}$, $\wt
d_\infty(g_1,g_2)=2 \sinh (n \norm[\infty] \mu/2)$ and
$d_\infty(g_1,g_2)=2 \norm[\infty] \mu$ (cf.\ \App{dist.met} for a
definition of $d_\infty(g_1,g_2)$), i.e., $\wt d(g_1,g_2)$ measures
the distortion rate of the conformal factor.

The following definition is standard (cf.,
e.g.,~\cite{mueller-salomonsen:07}).

\begin{definition}
  \label{def:quasi}
  We say that the metrics $g_1$, $g_2$ are \emph{quasi-isometric} if
  there exists a constant $\eta>0$ such that
  \begin{equation*}
    \eta g_1(x)(\xi,\xi) \leq g_2(x)(\xi,\xi)
    \le \eta^{-1} g_1(x)(\xi,\xi)
  \end{equation*}
  for all $\xi \in T^*M$ and $x \in M$.
\end{definition}
Note that $g_1,g_2$ are quasi-isometric if and only if
\begin{equation*}
  \wt d_\infty(g_1,g_2) < \infty
\end{equation*}
(see \Rem{q-iso} in \App{dist.met}).

The following equivalence of norms follows immediately from the
definition and~\eqref{eq:met.est}--\eqref{eq:met.est2}:
\begin{proposition}
  \label{prp:eq.norms}
  If the metrics $g_1$, $g_2$ on $M$ are quasi-isometric, i.e., if the
  quasi-distance fulfills $\wt d_\infty=\wt d_\infty(g_1,g_2)<\infty$,
  then
 \begin{gather*}
    (1+\wt d_\infty)^{-1} \norm[\HS_1] u
    \le \norm[\HS_2] u
    \le (1+\wt d_\infty) \norm[\HS_1] u,\\
    (1+\wt d_\infty)^{-1} \norm[\hat \HS_1] \omega
    \le \norm[\hat \HS_2] \omega
    \le (1+\wt d_\infty) \norm[\hat \HS_1] \omega,\\
    (1+\wt d_\infty)^{-1} \norm[\HS_1^1] u
    \le \norm[\HS_2^1] u
    \le (1+\wt d_\infty) \norm[\HS_1^1] u
  \end{gather*}
   for $u$ resp.\ $\omega$ in the appropriate spaces.  Here, $\HS_k :=
  \Lsqr{M,g_k}$, $\hat \HS_k := \Lsqr {T^*M,g_k}$ and $\HS_k^1 :=
  \Sob{M,g_k}$.  In particular, the spaces for $k=1$ and $k=2$ are
  identical as vector spaces and have equivalent norms.
\end{proposition}

Let us assume for the rest of this section that $\wt d_\infty(g_1,g_2)
< \infty$, i.e., that $g_1$, $g_2$ are quasi-isometric
(cf.~\Rem{q-iso}).

We let $\map I {\Lsqr {\Mscr_1}} {\Lsqr {\Mscr_2}}$, $If_1 := f_1$, denote
the natural identification operator; its adjoint, $I^*$, is given by
$I^*g = \rho g$, where $\rho=\rho_{g_2,g_1}$ is the density of
$\dvol_{g_2}$ with respect to $\dvol_{g_1}$, i.e.,
\begin{equation*}
  \dvol_{g_2} = \rho_{g_2,g_1} \dvol_{g_1}.
\end{equation*}
Note that $\rho=(\det A)^{-1/2}$.  Finally, we let $H_k$ denote the
(self-adjoint, non-negative) Laplacian operator $\Delta_{\Mscr_k}$
acting in $\Lsqr{\Mscr_k}$; for simplicity of notation, we write $R_k
:= (H_k+1)^{-1}$ for the resolvents.

The aim of the following is to {\colour[B6] find conditions which will
  allow us to show that, for sufficiently large $m \in \N$, the
  operators}
\begin{equation}
  \label{eq:2.5}
     V := R_2^m (H_2 I - I H_1) R_1^m
\end{equation}
can be written as a sum of products of Hilbert-Schmidt operators; this
will be achieved in \Lem{res.diff} {\colour[B6] and \Prp{hs}}. We
begin with some technicalities.  Let
\begin{align*}
  \map{S &:= 2 \sinh \frac 12 \ln \rho = \rho^{1/2} - \rho^{-1/2}}
      M \R,\\
  \map{\hat S &:= 2 \sinh \frac 12 \ln (\rho A) =
   (\rho A)^{1/2} - (\rho A)^{-1/2}}
      M {\BdOp {T^*M}},
\end{align*}
where $\BdOp{T^*M}$ denotes the vector bundle of (fiberwise)
endomorphisms of $T^*M$.  We denote the corresponding multiplication
operators on $\HS_k=\Lsqr{M,g_k}$ (resp.\ $\hat \HS_k=
\Lsqr{T^*M,g_k}$) by $S_k$ (resp.\ $\hat S_k$). (Strictly speaking,
$\hat S(x)$ acts as an endomorphism on the fiber $T_x^* M$, but we
call $\hat S_k$ also a multiplication operator.)  We use the
(pointwise) polar decomposition $S = \abs S (\sgn S)$ and $\hat S =
\abs {\hat S} (\sgn \hat S)$, where $\abs S(x) = \abs {S(x)} \ge 0$
and $\abs{\sgn S(x)}=1$, and where $\abs {\hat S}(x)$ is a
non-negative endomorphism and $\sgn \hat S(x)$ is unitary on $T_x^*
M$.

The following pointwise estimates will be used in \Prp{hs}.
\begin{lemma}
  \label{lem:s.est}
  We have the pointwise estimate
  \begin{equation*}
    0 \le \abs S, \abs[\BdOp{T^* \Mscr}]{\hat S} \le \wt d(g_1,g_2),
  \end{equation*}
  where $\abs[\BdOp{T^* \Mscr}]{\cdot}$ denotes the pointwise operator
  norm in $\BdOp{T^* \Mscr}$.
\end{lemma}
\begin{proof}
  Let us prove the estimate for $\hat S$, the estimate $0 \le \abs S
  \le \wt d(g_1,g_2)$ can be seen in the same way.  We have
  \begin{equation*}
    \abs[\BdOp{T^* \Mscr}]{\hat S}
    = \bigabs[\BdOp{T^*\Mscr}] {(\rho A)^{1/2} - (\rho A)^{-1/2}}
    = 2 \sinh \Bigabs[\BdOp{T^*\Mscr}] {\frac 12 \ln (\rho A)}.
  \end{equation*}
  Moreover, the $i$-th eigenvalue of $\ln (\rho A)$ is given by
  \begin{equation*}
    - \sum_{k=1}^n \frac {\ln \alpha_k} 2 + \ln \alpha_i.
  \end{equation*}
  Let $k_0$ be such that $\abs{\ln \alpha_{k_0}}= \max_k
  \abs{\ln\alpha_k}$.  Then
  \begin{equation*}
    \Bigl|-\sum_{k=1}^n \frac{\ln \alpha_k} 2 + \ln \alpha_i \Bigr|
    \le  \frac n2  \abs{\ln \alpha_{k_0}}.
  \end{equation*}
  Therefore, we have
  \begin{equation*}
    \abs[\BdOp{T^* \Mscr}]{\hat S}
    \le 2 \sinh \frac n4 \abs{\ln \alpha_{k_0}}
    = \wt d (g_1,g_2).\qedhere
  \end{equation*}
\end{proof}

Denote by $D_k$ the exterior derivative viewed as a closed operator from
$\HS_k$ into $\hat \HS_k$ with domain $\dom D_k =
\HS_k^1=\Sob{M,g_k}$.  We also set
\begin{align*}
  \map U {&\HS_1} {\HS_2}, \qquad U u := (\sgn S) \rho^{-1/2} u,\\
  \map {\hat U} {&\hat \HS_1} {\hat \HS_2}, \qquad 
     \hat U \omega := (\sgn \hat S)(\rho A)^{-1/2} \omega.
\end{align*}
It is easily seen that $U$ and $\hat U$ are unitary.
We now define $V$ as a quadratic form and provide a decomposition of
$V$ which involves two terms, each of them a product containing
factors of operators $B_k^{(m)}$ or $\hat B_k^{(m)}$, as defined
below. We will show in a second step that these factors are
Hilbert-Schmidt operators, provided $m$ is large enough.
\begin{lemma}
  \label{lem:res.diff}
  For any $m \in \N$ the operator $V$ in eqn.~\eqref{eq:2.5} can be
  written as
  \begin{equation*}
    V = \bigl(\hat B_2^{(m)}\bigr)^* \hat U \hat B_1^{(m)}
            - \bigl(B_2^{(m)}\bigr)^* U B_1^{(m-1)} H_1 R_1,
  \end{equation*}
  where
  \begin{equation*}
    \map{B_k^{(m)} := \abs{S_k}^{1/2} R_k^m}{\HS_k}{\HS_k}
    \quad\text{and}\quad
    \map{\hat B_k^{(m)}
      := \abs{\hat S_k}^{1/2} D_k R_k^m} {\HS_k}{\hat \HS_k}.
  \end{equation*}
\end{lemma}
\begin{proof}
  Let us first introduce a second identification operator $\map {I'}
  {\HS_2} {\HS_1}$ given by $I'f_2=f_2$ for $f_2 \in \HS_2$; note that
  $I'=I^{-1}$.  Letting $\qf h_k(u)=\normsqr[\hat \HS_k]{D_k u}$
  denote the quadratic forms of the operators $H_k$, \Prp{eq.norms}
  implies $I(\dom \qf h_1) = \dom \qf h_2$.  We now consider $f_k \in
  \HS_k$ and write $h_k := R_k^m f_k$. With $I^* h_2 = \rho h_2$ we
  then compute
  \begin{align*}
    \iprod[\HS_2]{V f_1} {f_2}
    &= (\qf h_2 (Ih_1, h_2) - \qf h_1 (h_1, I' h_2))
       - \iprod[\HS_2] {(I - (I')^*)H_1 h_1} {h_2} \\
    &= \int_M \bigl( \bigiprod[g_2] {(1- \rho^{-1} A^{-1}) dh_1}
                                                           {dh_2} 
                   - (1 - \rho^{-1}) (H_1 h_1) \, \conj h_2
              \bigr) \dvol_{g_2}\\
    &= \iprod[\hat \HS_2] {\hat U \abs{\hat{S}_1}^{1/2} \de h_1} 
                                       {\abs{\hat{S}_2}^{1/2} \de h_2}
        - \iprod[\HS_2] { U \abs{S_1}^{1/2} H_1 h_1} {\abs{S_2}^{1/2} h_2}
  \end{align*}
  and the desired factorization follows.
\end{proof}
The advantage of using the identification operator $I'$ instead of
$I^*$ in the quadratic form is to avoid a condition on $\de \rho$.

We define another ``distance'' between two metrics $g_1$, $g_2$ on the
$n$-dimensional manifold $M$: for a given continuous function $\map
{r_0} M {(0,1]}$ we let
\begin{equation}
  \label{eq:met.l1}
  \wt d_1(g_1,g_2)
  := \int_M \wt d(g_1,g_2)(x) \cdot r_0(x)^{-(n+2)} \cdot   
         (1+\rho_{g_2,g_1}(x)) \dvol_{g_1} (x)  
\end{equation}
denote the \emph{weighted $\Lpspace[1]$-quasi-distance} of $g_1$ and
$g_2$.  The factor $ 1+\rho_{g_2,g_1}$ is introduced in order to make
$\wt d_1$ symmetric.  We will be interested in situations where $\wt
d_1(g_1,g_2)$ is finite, and we will show in \Rem{d1.q-met} that $\wt
d_1$ is actually a quasi-distance (see \Def{q-metric}) on a suitable
subset.

We next discuss a Hilbert-Schmidt property of the operators
$B_k^{(m)}$ and $\hat B_k^{(m)}$, for $m$ large:
\begin{proposition}
  \label{prp:hs}
  Let $\map {r_0} M {(0,1]}$ be a continuous function. Suppose we are given
  two quasi-isometric metrics $g_1$, $g_2$ on $M$ (i.e., $\wt
  d_\infty(g_1,g_2)<\infty$) which have $r_0$ as a common lower bound
  for the harmonic radii, i.e., $\rH {(M,g_k)} (x) \ge r_0(x)$ for all
  $x \in M$ and $k = 1,2$.

  {\colour[B7] If $\wt d_1(g_1,g_2) < \infty$ then} for any $m \in \N$
  with $m \ge [n/4] + 2$, the operators $B_k^{(m)}$ and $\hat
  B_k^{(m)}$, defined in \Lem{res.diff}, are
  Hilbert-Schmidt. Furthermore, their Hilbert-Schmidt norms satisfy
  the estimate
  \begin{equation*}
    \bignormsqr[\HSspace] {B_k^{(m)}}, \;
    \bignormsqr[\HSspace] {\hat B_k^{(m)}}
    \le C \wt d_1(g_1,g_2), \qquad k = 1,2, 
  \end{equation*}
  where $C$ depends only on $m$, $n$, $p$ and $Q$.
\end{proposition}
\begin{proof}
  By the Riesz Representation Theorem and \Thm{b.1} it is easy to see
  that the resolvents $R_k^m$ are integral operators with (measurable)
  kernels $G_k^{(m)}(x,y)$ satisfying
  \begin{equation*}
    \int_M \abssqr{G_k^{(m)}(x,y)}  \dvol_{g_k}(y)
    \le C (\min\{1,\rH{(M,g_k)}(x)\})^{-n}  
    \le C r_0(x)^{-n},
  \end{equation*}
  where $C$  depends only on $m$, $n$, $p$, and $Q$. 
 We similarly obtain for the kernels $\de_x G_k^{(m)}(x,y)$ of the $D_k R_k^m$ 
 that 
   \begin{equation*}
     \int_M \abssqr{\de_x G_k^{(m)}(x,y)}  \dvol_{g_k}(y)
    \le C (\min\{1,\rH{(M,g_k)}(x)\})^{-n-2}
     \le C r_0(x)^{-n-2}.
  \end{equation*}
  The Hilbert-Schmidt norm of $B_k^{(m)}$ is given by
   \begin{equation*}
     \normsqr[\HSspace]{B_k^{(m)}}
    = \int_{M \times M} \abs{S_k(x)} \abs{G_k^{(m)}(x,y)}^2 
      \dvol_{g_k}(y) \dvol_{g_k}(x),
  \end{equation*}
  hence with the previous estimate on the kernel and the pointwise
  estimate $\abs S \le \wt d(g_1,g_2)$ from \Lem{s.est} we obtain
  \begin{align*}
    \bignormsqr[\HSspace] {B_1^{(m)}}
    &\le C \int_M {\wt d}(g_1,g_2)(x) r_0(x)^{-n} \dvol_{g_1}(x)
    \le C \wt d_1(g_1,g_2),\\
    \bignormsqr[\HSspace] {B_2^{(m)}}
    &\le C \int_M {\wt d}(g_1,g_2)(x) r_0(x)^{-n} 
            \rho_{g_2,g_1}(x) \dvol_{g_1}(x)
    \le C \wt d_1(g_1,g_2).
  \end{align*}
  The estimate on $\bignormsqr[\HSspace] {\hat B_k^{(m)}}$ follows in a
  similar fashion.
\end{proof}

Combining the last two propositions leads to the following result:
\begin{corollary}
  \label{cor:res.diff}
  Under the assumptions of the preceding \Prp{hs}, the operator $V :=
  R_2^m (H_2 I - I H_1) R_1^m$ is trace class. Furthermore,
  $\norm[\TRspace] V$, the trace norm of $V$, satisfies the estimate
  \begin{equation*}
    \norm[\TRspace] V \le 2C  \wt d_1(g_1,g_2), 
  \end{equation*}
  where the constant $C$ depends only on $m, n, p$, and $Q$.
\end{corollary}

We are now ready for the main theorem of this section.  Recall that
$\Met_{r_0} (M)$ consists of the set of metrics $g$ on $M$ that
satisfy the lower bounds~\eqref{eq:lower.bd} for the injectivity
radius and the Ricci-curvature in terms of a continuous function $\map
{r_0} M {(0,1]}$.  Also recall the definition of the wave operators
$W_\pm(H_2,H_1,I)$ in \Sec{two.hs.scatt}.

\begin{theorem}
  \label{thm:main1}
  Suppose we are given a smooth manifold $M$ and a continuous function
  $\map{r_0} M {(0,1]}$.  Let $g_1, g_2 \in \Met_{r_0}(M)$ denote two
  quasi-isometric Riemannian metrics on $M$.  Furthermore, we assume
  that the difference between $g_1$ and $g_2$ satisfies the
  $r_0$-dependent weighted integral condition $\wt
  d_1(g_1,g_2)<\infty$, with $\wt d_1$ as in \eqref{eq:met.l1}.
  
  Then the wave operators $W_\pm(H_2, H_1, I)$ exist and are complete.
  Furthermore, the $W_\pm(H_2, H_1, I)$ are partial isometries with
  initial space $\HS_\ac(H_1)$ and final space $\HS_\ac(H_2)$.
\end{theorem}

\begin{proof}
  From \Prp{lower.bd.rh}, we obtain the lower bound $\rH{(M,g_k)}(x)
  \ge c r_0(x)$ on the harmonic radii.  {\colour[A5]%
    We are now going to check the assumptions of the
    Birman-Belopol'skii theorem as given in \Thm{birman-belopolskii}
    or~\cite[Thm.~XI.13]{reed-simon-3}:} That $I$ is bounded and has a
  bounded inverse is nothing but the equivalence of the
  $\Lsqrspace$-norms on $\Lsqr {\Mscr_k}$, cf.~\Prp{eq.norms}.  Note
  that $I^{-1}$ is simply the identification $I^{-1}h = h$, for $h \in
  \Lsqr{\Mscr_2}$.

  For the trace class condition, let $m \in \N$ satisfy $m \ge [n/4] +
  2$ and let $E_{\mathbb I}(H)$ denote the spectral projection of $H$
  associated with a bounded interval $\mathbb I \subset \R$.  As
  $E_{\mathbb I}(H)(H+1)^m$ is bounded, it follows from \Cor{res.diff}
  that $E_{\mathbb I}(H_2) (H_2 I - I H_1) E_{\mathbb I}(H_1)$ is
  trace class.

  Moreover, $V_1:=(I^*I - 1) R_1^m = (\rho - 1) R_1^m = T_1
  B_1^{(m)}$, where $T_1=\rho^{1/2} \sgn S \abs S^{1/2}$ is a bounded
  multiplication operator.  From \Prp{hs} we see, using again $\wt
  d_1(g_1,g_2)<\infty$, that $B_1^{(m)}$ (and therefore $V_1$) is
  Hilbert-Schmidt.  In particular, $(I^*I - 1)E_{\mathbb I}(H_1)$ is
  compact.

  We also have $I(\dom \qf h_1) = \dom \qf h_2$ (by \Prp{eq.norms}),
  where $\dom \qf h_k$ is the quadratic form domain of $H_k$.  Note
  that we have used the quasi-isometry of $g_1$ and $g_2$ here.  The
  desired results now follow from \Thm{birman-belopolskii}.
\end{proof}

\begin{remark}
  \label{rem:ends}
  In typical applications to manifolds with ends it is evident that
  $\rho(x) \to 1$ ``at infinity'' in the sense that for each $\eps >
  0$ there is a compact subset $M_\eps \subset M$ such that $|1 -
  \rho(x)| < \eps$ for all $x \notin M \setminus M_\eps$.  The
  compactness of $(I^*I - 1) (H_1 + 1)^{-m}$ is then immediate by
  local compactness of the Laplacian.
\end{remark}

\begin{remark}
  \label{rem:mueller-salomonsen}
  Let us comment on the result by M\"uller and
  Salomonsen~\cite{mueller-salomonsen:07}, mentioned in the
  introduction: We first note that the assumptions of Theorem~0.1
  of~\cite{mueller-salomonsen:07} already imply the assumptions of our
  \Thm{main1}, namely, $g_1 \sim_\beta^2 g_2$ in the sense
  of~\cite{mueller-salomonsen:07} implies that the metrics are
  quasi-isometric (this follows
  from~\cite[Lem.~1.7]{mueller-salomonsen:07}) and that $\wt
  d(g_1,g_2)(x) \le C_1 |g_1(x)-g_2(x)|_{g_1(x)} \le C_2 \beta(x)$.
  Here, $\beta$ is a function called \emph{of moderate decay}
  (actually, $\beta$ is a function of $d_{g_1}(x,p)$ for some
  reference point $p \in M$), and in particular bounded.

  {\colour[B8] M\"uller and Salomonsen require a \emph{weighted bound}
    on a modified injectivitiy radius which implies that} $\wt
  d_1(g_1,g_2) < \infty$: From condition~(iii) in their Thm.~0.1,
  it follows that the injectivity radius $\inj_{(M,g_1)}(x)$ is
  bounded from below by $r_0(x):=C_3 \beta(x)^{2a/(3n(n+2))}$ for some
  $a \le 1$.  Since the curvature is assumed to be bounded in Thm.~0.1
  of~\cite{mueller-salomonsen:07}, our condition on the Ricci
  curvature in~\eqref{eq:lower.bd} is automatically fulfilled.
  Moreover, by~(ii) in Thm.~0.1, one has $\beta^{b/3} \in
  \Lp[1]{M,g_1}$ with $b=2-a \ge 1$.  The latter condition, together
  with $d(g_1,g_2)(x) \le C_2 \beta(x)$, implies that
  \begin{equation*}
    \wt d_1(g_1,g_2) =
    \int_M {\wt d(g_1,g_2)(x)}{r_0(x)^{-(n+2)}} (1+\rho_{g_2,g_1}(x)) \dd x
    \le C_4 \int_M \beta(x)^{1-2 a/(3n)} \dd x.
  \end{equation*}
  Since $1-2 a/(3n)=1-4/(3n)+2b/(3n) \ge b/3$ for $n \ge 2$, and since
  $\beta$ is bounded, it follows from $\beta^{b/3} \in \Lp[1]{M,g_1}$
  that $\wt d_1(g_1,g_2) < \infty$.

  Let us also note that~\cite{mueller-salomonsen:07} require more
  regularity on the deviation of the metric $g_2$ from $g_1$. In fact,
  $g_1 \sim_\beta^2 g_2$ means that the derivatives up to order $2$
  have to be close to each other with respect to a weight function
  $\beta$; furthermore, the manifolds $(M,g_i)$ {\colour are supposed
    to have curvature bounded up to order $2$}.  Our assumptions here
  are weaker in the sense that there are only \emph{relative}
  conditions on the metrics, i.e., the metrics $g_1$ and $g_2$ only
  have to be quasi-isometric and that $\wt d_1(g_1,g_2) < \infty$.  No
  \emph{global} boundedness assumption on the curvature of $(M,g_i)$
  has to be made, and \emph{no} condition on the derivatives of $g_1$
  and $g_2$.
\end{remark}


%
\section{Manifolds with ends}
\label{sec:mfds.ends}
%

The general setup presented in this section is fairly standard and has
been used in a similar way by many authors; {\colour[A8] cf.\
  \Rem{mfd.with.ends} below.}  Let $M$ be a smooth, orientable,
connected $n$-dimensional manifold and let $g$ be a metric on $M$ such
that the Riemannian manifold $\Mscr = (M,g)$ is complete.  Our
manifolds with ends are characterized by geometric and spectral
assumptions. We first describe the geometry:

\begin{assumption}
  \label{ass:geo}
  We assume that $M$ can be decomposed into $\ell + 1$ open
  submanifolds $M_k$, $k=0, \dots, \ell$, where $\clo M_0$ is compact
  and $\clo M_k$, $k=1,\dots, \ell$, are non-compact.  More precisely,
  we assume that the boundaries {\colour[B9] $\Sigma_k := \partial M_k
    = \clo M_k \cap \clo M_0$, $k = 1, \ldots, \ell$, are pairwise
    disjoint, smooth and compact manifolds of dimension $n-1$; in
    particular,}
  \begin{equation*}
    M = \clo M_0 \cup M_1 \cup \dots \cup M_\ell.
  \end{equation*} 
\end{assumption}

We denote the corresponding Riemannian manifolds by $\Mscr_k=
(M_k,g)$, $k=0,\dots,\ell$.  In addition, we denote the boundaries,
now considered as Riemannian manifolds, by $\bd \mathcal M_k =
\mathcal S_k := (\Sigma_k, \iota^* g)$, $k=1,\dots, \ell$, where
$\embmap \iota \Sigma M$ denotes the natural embedding and $\iota^* g$
the induced metric.

\begin{remark} 
  If we allow metrics $g$ of class $\Contspace [1,\alpha]_\loc(M)$,
  then the induced metric $\iota_k^* g$ is of class $\Cont
  [1,\alpha]{\partial M_k}$ which is sufficient for our purposes,
  especially for the proof of \Prp{4.3}; cf.\ also~\cite{akklt:03} for
  a related discussion of regularity properties.
\end{remark}

We next turn to the spectral assumptions.  Let $H$ denote the
Laplacian of $\Mscr$ and let $H_k$ denote the Laplacian of $\Mscr_k$,
$k = 0, \dots, \ell$, where $H_0$ satisfies Dirichlet boundary
conditions on
\begin{equation*}
  \Sigma := \Sigma_1 \cup \dots \cup \Sigma_\ell = {\colour \bd M_0},
\end{equation*}
while the $H_k$ satisfy Dirichlet boundary conditions on $\Sigma_k$.
We can define $H_0$ as the self-adjoint operator associated with the
quadratic form $\qf h_0$ (cf.~\cite[Thm.~VI.2.1]{kato:66}), where $\qf
h_0$ is given by $\qf h_0(u) := \normsqr[\Lsqr{T^*\Mscr_0}] {\de u}$
and $\dom \qf h_0 := \Sobn{\Mscr_0}$.  Moreover, the Sobolev space
$\Sobn{\Mscr_0}$ is the completion of $\Cci{M_0}$ (functions with
support away from $\bd M_0$) in the norm of $\Sob{\Mscr_0}$.  The
operators $H_k$, $k=1,\dots, \ell$, are defined in a similar way.
 
We will need the decoupled Laplacians $H_\dec = \mybigoplus_{k=0}^\ell
H_k$ and $H'_\dec = \mybigoplus_{k=1}^\ell H_k$.  Finally, let
$\map{I_k} {\Lsqr{\Mscr_k}}{\Lsqr \Mscr}$ denote the natural
embedding, for $k = 0, \ldots, \ell$, and, similarly, let $I$ denote
the natural embedding of $\mybigoplus_{k=1}^\ell \Lsqr {\Mscr_k}$ into
$\Lsqr \Mscr$.  The Laplacian $H_0$ of $\Mscr_0$ has compact resolvent
and thus purely discrete spectrum.

\begin{assumption}
  \label{ass:spec}
  In addition to \Ass{geo}, we assume that each of the decoupled
  Laplacians $H_k$, $k = 1, \dots, \ell$, has a (non-trivial)
  absolutely continuous part.
\end{assumption}

\begin{remark}
  Given \Ass{geo}, this is an assumption on the metric $g$.  If
  \Ass{spec} is satisfied, each of the ends of $\Mscr$ will constitute
  a scattering channel.  \Ass{spec} is made chiefly for simplicity of
  notation later on.  Indeed, it would be easy to adapt our results to
  the case where some of the ends do \emph{not} participate in the
  scattering at all; e.g., an (infinite) horn may have purely discrete
  spectrum if it shrinks fast enough (cf.~\cite{bruening:89}).  This
  point will be discussed further in \Sec{example}.
\end{remark}
{\colour[A8]
\begin{remark}
  \label{rem:mfd.with.ends}
  As mentioned above, this setup can be considered standard. In fact,
  in many papers it is even assumed right from the beginning that each
  end of the manifold $\Mscr$ is given as a warped product on $N
  \times (0,\infty)$ where $N$ is a compact manifold (cf.,
  e.g.,~\cite{dbhs:92} for an early, and~\cite{kenichi-nakamura:10}
  for a recent reference). In particular, many examples are
  constructed using simple coordinate systems of this type, and our
  examples in \Sec{example} are no exception to that.

  Let us also note that the concept of a scattering manifold in the
  sense of Melrose fits into our general scheme.  Indeed, Melrose
  removes from a compact, smooth manifold a finite number of open sets
  $U_k$ (with $\partial U_k$ smooth and $\overline{U}_k$ pairwise
  disjoint, $k = 1, \ldots, \ell$), requiring the metric to become
  singular near $\partial U_k$ in a specific way. Here we may consider
  open sets $V_k \supset \overline{U}_k$, again with $\partial V_k$
  smooth and $\overline{V}_k$ pairwise disjoint, to model the ends on
  $M_k := V_k \setminus \overline{U}_k$; of course, the metric has to
  be singular near $\partial U_k$ in a certain sense.
\end{remark} 
}
We next discuss the decoupling by Dirichlet boundary conditions on the
submanifold $\Sigma$.  Similar decoupling arguments have been used
(for the smooth case) in Birman~\cite{birman:62, birman:63},
Weder~\cite{wed:84}, Yafaev~\cite{yafaev:92}, Hempel and
Weder~\cite{hempel-weder:93}, and Carron~\cite{carron:02}, to name
just a few.  {\colour[A9] There are different approaches to check that
  a trace class condition is satisfied; for example,
  Carron~\cite{carron:02} uses techniques from pseudo-differential
  operators for showing that the resolvent difference (denoted by $V$
  below) is trace class.  Here, we rather use techniques from elliptic
  regularity (see \Prp{b.5} in the appendix). Our method requires
  minimal regularity of the Riemannian metric only and uses methods
  from \App{ptw.bd} that are indispensable for our analysis anyway.}

\begin{proposition}
  \label{prp:4.3}
  Under the above \Ass{geo}, 
  the wave operators
  \begin{equation}
    \label{eq:3.1} 
    W_\pm(H, H_\dec)
    = \slim_{t \to \pm \infty} \e^{\im tH} \e^{-\im tH_\dec}
                                P_\ac(H_\dec)
  \end{equation}
  and
  \begin{equation}
    \label{eq:3.2} 
    W_\pm(H, H'_\dec, I)
    = \slim_{t\to \pm\infty} \e^{\im tH} I \e^{-\im tH'_\dec}
                             P_\ac(H'_\dec)
  \end{equation}
  exist, are complete, and partially isometric. In particular,
  \begin{equation}
    \label{eq:3.3}
    \ran W_\pm(H, H_\dec) = \HS_\ac(H),
  \end{equation}
  where
  \begin{equation}
    \label{eq:3.4}  
    \ran W_\pm(H, H_\dec) = \ran W_\pm(H, H'_\dec, I)
    = \mybigoplus_{k=1}^\ell \ran W_\pm(H, H_k, I_k).
  \end{equation}
\end{proposition}
\begin{proof} The proof is a modification of the proof of \Thm{main1}
  and we only give a sketch.  We first need some notation.  For $k =
  1, \ldots, \ell$, let $\mathcal U_k =(U_k,g)$ be a collar
  neighborhood of $\partial \mathscr M_k=\mathcal S_k$ in $\mathscr
  M_k$.  Similarly, we let $\mathcal U_0$ denote a collar neighborhood
  of $\partial \Mscr_0 = \mathcal S_1 \cup \dots \cup \mathcal S_\ell$
  in $\Mscr_0$.  Without loss {\colour[B10] of generality} we may
  assume that $U_0, \dots, U_\ell$ are relatively compact.  For
  brevity, we denote the $(\ell + 1)$-tuple of these collar
  neighborhoods as $\mathcal U$ and we let $\Sob[2]{\mathcal U}$
  denote the $(\ell + 1)$-tuple of the second order Sobolev spaces
  $\Sob[2]{\mathcal U_k}$ with $k = 0, \dots, \ell$.  We define the
  boundary operators
  \begin{align*}
    \map{\Gamma}{\Sob \Mscr} {\mathcal G}, &&
    \Gamma f &:= f \restr \Sigma,\\
    \map{\Gamma'}{\Sob[2] {\mathcal U}} {\mathcal G}, &&
    \Gamma' f &:= \partial_{\mathrm n_+} f \restr \Sigma +
                \partial_{\mathrm n_-}f \restr \Sigma,
  \end{align*}
  where $\mathcal G= \Lsqr {\mathcal S} = \mybigoplus_k \Lsqr
  {\mathcal S_k}$, $\partial_{\mathrm n_\pm} f =\de f \cdot \mathrm
  n_\pm$, and where $\mathrm n_\pm$ are unit vector fields on
  $\mathcal U_k$ resp.\ $\mathcal U_0$, normal to $\Sigma_k$ and
  pointing outwards of $\Mscr_k$ resp.\ of $\Mscr_0$.

  Defining $V$ via $\iprod {V\wt f} {\wt h} := \iprod{R_\dec^m \wt
    f}{H R^m \wt h} - \iprod{H_\dec R_\dec^m \wt f}{R^m \wt h}$ for
  $\wt f, \wt h \in \HS = \Lsqr \Mscr$ and fixing $m \in \N$, $m \ge
  [n/4] + 2$, we then have
  \begin{equation*}
    \iprod {V \wt f} {\wt h}
    = \iprod f {H h} - \iprod {H_\dec f} h
    = \iprod[\mathcal G] {\Gamma' f}{\Gamma h}
    = \iprod {(\Gamma R^m)^* (\Gamma' R_\dec^m) \wt f} {\wt h}
  \end{equation*}
  by Green's formula, where $f=R_\dec^m \wt f$ and $h=R^m \wt h$.  We
  have to show that $A := \Gamma R^m$ and $B:= \Gamma' R_\dec^m$ are
  Hilbert-Schmidt operators.  Their integral kernels are given by
  $\Gamma_x G^{(m)}(x,y)$ and $\Gamma_x'G^{(m)}_\dec(x,y)$
  respectively, for $x \in \Sigma$, $y \in M$, where $G^{(m)}$ resp.\
  $G^{(m)}_\dec$ denote the kernel of $R^m$ resp.\ $R_\dec^m$.
  {\colour[B11] Let $\Cont[1]{\Mscr}$ denote $\Cont[1] M$ equipped
    with the usual Fr\'echet-topology generated by the semi-norms
    $p_K(u) := \sup_{x\in K}(|u(x)| + |\dd u(x)|_g)$, with $K \subset
    M$ compact. Similarly, we let $\Cont[1]{\overline{\Mscr}_s}$
    denote the space of functions $u \in \Cont[1]{\overline{M}_s}$
    with the same family of semi-norms as above. The restrictions
    $\chi_{\overline{U}_s}$ define continuous embeddings of
    $\Cont[1]{\overline{\Mscr}_s}$ into $
    \Cont[1]{\overline{\Uscr}_s}$, for $s = 0, \ldots, \ell$.  By
    \Prp{b.5}, $ R^m$ maps $\HS$ continuously into $\Cont[1]{\Mscr}$,
    while $\chi_{\overline{U}_s} R_\dec^m$ maps $\HS$ continuously
    into $\Cont[1]{\overline{\Uscr}_s}$.}  As in the proof of
  \Prp{hs}, it follows by the Riesz Theorem that there is a constant
  $\wt C>0$ such that for all $x \in \Sigma$
  \begin{equation*}
    \int_M \bigabssqr{\Gamma_x G^{(m)}(x,y)} \dvol_g(y) \le \wt C, 
   \qquad
    \int_M \bigabssqr{\Gamma'_x G^{(m)}(x,y)} \dvol_g(y) \le \wt C.
  \end{equation*}
  Since $\Sigma$ is compact, $A$ and $B$ are Hilbert-Schmidt and the
  result follows.
\end{proof}

We now discuss incoming and outgoing states. It follows
from~\eqref{eq:3.3}--\eqref{eq:3.4} that for each state $f \in
\HS_\ac(H)$ there are vectors $g_{k;\pm} \in \HS_{\rm ac}(H_k)$, $k=1,
\dots, \ell$, such that
\begin{equation}
  \label{eq:3.5}
  \Bignorm {\e^{\pm \im tH} f  -  
            \mybigoplus_{k=1}^\ell I_k \e^{\pm \im t H_k} g_{k;\pm}} 
  \to 0, \qquad t \to \pm \infty. 
\end{equation}
We call $\ran W_\pm(H, H_k, I_k)$ the \emph{outgoing} (for $+$) and
the \emph{incoming} subspace (for $-$) of the $k$-th scattering
channel.  In this sense, scattering on $\Mscr$ can be understood as an
interaction of $\ell$ scattering channels. If $K \subset M$ is compact
and $f \in \HS_\ac(H)$, then $\chi_K \e^{\im tH}f \to 0$, as $t \to
\pm \infty$, by the usual arguments.

In many cases (and in particular in concrete examples), one would like
to describe the long-time behavior of the system by suitable
\emph{asymptotes:} for any given channel $\Mscr_k$, $k \in \{1,
\ldots, \ell\}$, there may exist a (simple) \emph{comparison dynamics}
given by an operator $h_k$, acting in a Hilbert space $\HS_k$, which
can be used to describe the asymptotic evolution of $\e^{-\im tH_k}f$
as $t \to \pm\infty$, for $f$ in the absolutely continuous subspace of
$H_k$.  More precisely, suppose there is a bounded operator
$\map{j_k}{\HS_k} {\Lsqr {\Mscr_k}}$ such that the wave operators
\begin{equation}
  \label{eq:3.6} 
  W_\pm(H_k, h_k, j_k)
  := \slim_{t \to \pm \infty} \e^{\im t H_k} j_k
  \e^{-\im t h_k} P_\ac(h_k)  
\end{equation}
exist, are partially isometric, and complete, so that, in particular,
\begin{equation*}
  \ran W_\pm(H_k, h_k, j_k) = \HS_\ac(H_k). 
\end{equation*}
If we assume, for the moment, that such reference operators $h_k$ exist for
\emph{all} $k = 1, \ldots, \ell$, the chain rule implies that
\begin{align*}
    W_\pm(H, \mybigoplus_{k=1}^\ell h_k, Ij)
    & = W_\pm(H, H'_\dec, I) 
         \circ W_\pm(H'_\dec, \mybigoplus_{k=1}^\ell h_k, j)\\
    & = W_\pm(H,H'_\dec,I)
         \circ 
           \Bigl( \mybigoplus_{k=1}^\ell W_\pm (H_k, h_k, j_k) \Bigr),
\end{align*}
where $j$ is the direct sum of the $j_k$. In view of~\eqref{eq:3.4} we
then see that
\begin{equation}
  \label{eq:3.7}
  \HS_\ac(H) = \mybigoplus_{k=1}^\ell \ran W_\pm(H, h_k, I_k j_k). 
\end{equation}
Note that the definition of the incoming/outgoing subspaces is
independent of the choice of the comparison dynamics given by the pairs
$(\HS_k, h_k)$.

Standard examples for ends include (asymptotically) Euclidean ends
where we may take the (flat) Laplacian on $\R^n$ to define the
comparison dynamics, half-cylinders, cusps or horns, funnels, etc. In
many of these examples the notion of ``incoming/outgoing'' corresponds
to the geometric notion of coming in from infinity or going out to
infinity; cf.\ \Eq{3.5}. Some of these examples will be discussed in
more detail in \Sec{example}.

Abstracting from the above situation, we use the following
terminology:
\begin{definition}
  \label{def:ref.op} 
  Let $H$ and $\Mscr$ satisfy \Asss{geo}{spec}.  Suppose we are given
  Hilbert spaces $\HS_k$ and self-adjoint operators $h_k$ acting in
  $\HS_k$, for $k =1, \ldots, \ell$.  We say that $h := \mybigoplus_k
  h_k$ is a \emph{reference operator with $\ell$ channels} for the
  Laplacian $H$ on $\Mscr$, if there are bounded operators
  $\map{J_k}{\HS_k}{\Lsqr{\mathcal M_k}}$ such that the following
  holds:

  If we define the \emph{identification operator} $\map J
  {\HS=\mybigoplus_k \HS_k} {\Lsqr \Mscr}$ by $Jf := \sum_k J_k f_k$
  then the wave operators $W_\pm(H, h, J)$ exist and are complete, and
  they are partially isometric with initial space $\mybigoplus_k
  \HS_\ac(h_k)$ and final space $\HS_\ac(H)$.

  We then define the associated \emph{scattering operator}
  by
  \begin{equation}
    \label{eq:3.8} 
    \map{S=S(H,h,J) := W_+^*(H,h,J) \circ W_-(H,h,J)} \HS \HS. 
  \end{equation}
\end{definition}

\begin{remark}
  \label{rem:ref.op}
  Note that we always have the trivial choice $h_k = H_k$, $\HS_k =
  \Lsqr{\Mscr_k}$ which leads to the scattering operator $S(H,
  H_\dec',I)$ with $H_\dec'$ and $I$ as in the beginning of this
  section.  While the scattering operators $S(H, h, J)$ depend on the
  choice of the reference operators, the chain rule implies that they
  are all unitarily equivalent to $S(H, H_\dec', I)$.

  {\colour In many instances a ``natural'' choice of a reference
    dynamics can be derived from the geometry of the ends, cf.\ e.g.\
    the recent work of Ito and Skibsted~\cite{kenichi-skibsted:13b}
    and the literature quoted there.}
\end{remark}

In order to study the interaction of the channels, we introduce the
\emph{scattering matrix} $(S_{ik})_{i,k =1, \ldots, \ell}$ where
\begin{equation*}
  \map{S_{ik} =  W_+^*(H,h_i,J_i) \circ W_-(H,h_k,J_k)}
  {\HS_k} {\HS_i};
\end{equation*}
in particular, $Sf = \bigl(\sum_{k=1}^\ell S_{ik}
f_k)_{i=1,\dots,\ell}$ for all $f \in \HS$.  Clearly, $S_{ik} \ne 0$
for a given pair of indices $(i,k)$ is equivalent to the fact that
some incoming states which for $t \to -\infty$ asymptotically lie in
the $k$-th end have a non-zero asymptotic part in the $i$-th end, as
$t \to +\infty$.  Put differently, $S_{ik} \ne 0$ is equivalent to
having non-zero transmission from the $k$-th into the $i$-th channel
or end, i.e., the $k$-th scattering channel is open to the $i$-th
channel.
 
Note that $S_{ik} \ne 0$ is equivalent to $\ran W_-(H,h_k,J_k) \cap
\ran W_+(H,h_i,J_i) \ne \{0\}$.  Furthermore, $\ran W_-(H,h_k,J_k)=
\ran W_+(H,h_k,J_k)$ implies that $S_{ik}=0$ for all $i \ne k$.
Conversely, suppose that $S_{ik} \ne 0$ for some $i \ne k$, then,
necessarily, $\ran W_-(H,h_k,J_k) \ne \ran W_+(H,h_k,J_k)$. 

The property $S_{ik} \ne 0$ is symmetric in the indices $i$ and $k$,
as we will show now: Denoting by $T$ the operator of complex
conjugation, $T\varphi(x)= \conj {\varphi(x)}$ for complex-valued
functions $\varphi$ {\colour[A10] ($T$ is the operator of time
  reversal in Quantum Mechanics),} we clearly have
\begin{equation*}
  T \e^{-\im t H}= e^{\im t H}T, \qquad 
  T \e^{-\im t h_k}= \e^{\im t h_k}T, 
\end{equation*}
and it follows that 
\begin{equation*}
  T W_{\pm}(H, h_k, J_k) = W_{\mp}(H,h_k, J_k) T, \qquad
  T W^\ast_\pm(H, h_k, J_k) = W^\ast_\mp (H,h_k, J_k) T. 
\end{equation*}
Then,
\begin{equation*}
  T S_{ik}
  = T W_+^*(H,h_i,J_i) \circ W_-(H,h_k,J_k)
  = W_-^*(H,h_i,J_i) \circ W_+(H,h_k,J_k) T
  =  S_{ki}^\ast T;  
\end{equation*} 
as a consequence, $S_{ik} \neq 0 \; \Longleftrightarrow \; S_{ki}^\ast
\neq 0 \; \Longleftrightarrow \; S_{ki}\neq 0$.  We are thus justified
in saying that the $i$-th and the $k$-th scattering channel are open
to one another.

In the next section, we will show that the property $S_{ik} \ne 0$ is
stable under small perturbations of the metric.

%
\section{Non-zero transmission under perturbation of the metric}
\label{sec:non-zero.transm}
%

In this section, we prove strong continuity of the scattering matrix
$S$ with respect to perturbations of the metric.  This immediately
implies the result mentioned at the end of \Sec{mfds.ends} on the
stability of the openness of scattering channels.
 
Let $M$ be as in \Ass{geo} and let $g_0$ be a metric on $M$ satisfying
\Ass{spec}; we consider $g_0$ as being fixed.  Let $r_0 > 0$ be a
function satisfying~\eqref{eq:lower.bd} for $\Mscr_0 = (M, g_0)$, and
let $H_0$ denote the Laplacian of $\Mscr_0$. Let
$h_0=\mybigoplus_{k=1}^\ell h_k$ denote a reference operator for $H_0$
with $\ell$ scattering channels as in \Def{ref.op}; in particular,
there is a Hilbert space $\HS_0$ and a bounded operator $\map J
{\HS_0} {\Lsqr \Mscr}$ such that the wave operators $W_\pm(H_0, h_0,
J)$ exist and are complete.  We may then define the associated
scattering operator $S(H_0, h_0, J)$ as in \Eq{3.8}.

We first describe the set of admissible metrics that are close to the
metric $g_0$.  Recall the definitions of $\wt d_\infty$ and $\wt d_1$
in \Sec{ex.wo}, eqns.~\eqref{eq:def.q-dist},
\eqref{eq:def.q-dist.loc}, and~\eqref{eq:met.l1}.  For $\gamma>0$ and
$\eps>0$ we set
 \begin{equation}
  \label{eq:met.ball}
  \Met_{r_0}(M,g_0, \gamma, \eps)
  := \bigset{g \in \Met_{r_0}(M)}
        {\wt d_\infty(g_0,g) \le \gamma, \quad
        \wt d_1(g_0,g) \le \eps},
\end{equation}
i.e., $\Met_{r_0}(M,g_0,\gamma,\eps)$ is the set of
smooth\footnote{For less regular metrics, cf.~\Rem{regularity}.}
metrics $g$ enjoying the following properties:
\begin{enumerate}
\item The homogenized Ricci curvature and the homogenized injectivity
  radius of $g$ are controlled \emph{locally} from below by the
  function $r_0$, cf.~\eqref{eq:lower.bd}.
\item The metric $g$ is quasi-isometric to $g_0$ ($\wt d_\infty(g_0,g)
  \le \gamma < \infty$, cf.~\Rem{q-iso}).
\item The weighted $\Lpspace[1]$-quasi-distance $\wt d_1(g_0,g)$ is
  smaller than or equal to $\eps$.  The quasi-distance is defined
  in~\eqref{eq:met.l1} with respect to the weight function $r_0$.
\end{enumerate}
We comment on the structure of the space $\Met_{r_0}(M,g_0,\gamma)$ of
admissible metrics later on in \Rem{d1.q-met}. 

Let us now present a stability result for the scattering operator
under perturbation of the metric:

\begin{theorem}
  \label{thm:main2}
  Let $g_0$, $H_0$, $h_0$ and $S(H_0,h_0,J)$ as above.  For $g_\eps
  \in \Met_{r_0}(M,g_0, \gamma, \eps)$, denote by $H_\eps$ the
  Laplacian of $\Mscr_\eps = (M, g_\eps)$ and by $I_\eps \colon {\Lsqr
    {\Mscr_0}} \to {\Lsqr {\Mscr_\eps}}$ the natural identification.
  Then $h_0$ is also a reference operator for $H_\eps$, and the
  associated scattering operators $S(H_\eps, h_0, I_\eps J)$ converge
  strongly to $S(H_0, h_0, J)$, as $\eps \to 0$.
\end{theorem}

\begin{remark}
  \label{rem:main2}
  In particular, it follows that, for $\eps > 0$ small, all $\ell$
  channels constitute a scattering channel for $H_\eps$. Moreover, if
  the $k$-th and the $i$-th channels are open to each other for $H$
  they will also be open to each other for $H_{\eps}$. Conversely,
  suppose there are $k_0, i_0 \in \{1, \ldots, \ell\}$ such that the
  $k_0$-th and the $i_0$-th channels are \emph{not open to each other}
  for a sequence $(H_{\eps_j})_{j\in\N}$ where $\eps_j \to 0$ as $j
  \to \infty$.  Then, the $k_0$-th and the $i_0$-th channels are also
  not open to each other for $H$.
\end{remark}

\begin{proof}[Proof of \Thm{main2}]
  It follows from \Thm{main1} that the wave operators
  $W_\pm(H_\eps,H_0,I_\eps)$ exist, are complete and partially
  isometric.  The same is true for the wave operators
  $W_\pm(H_0,h_0,J)$ by \Def{ref.op}, with $\map J {\HS_0}
  {\Lsqr{\Mscr_0}}$ as above.  By the chain rule, we have existence
  and completeness of the wave operators
  \begin{equation*}
    W_\pm(H_\eps, h_0, I_\eps J)
    = W_\pm(H_\eps, H_0, I_\eps) \circ W_\pm(H_0, h_0, J).
  \end{equation*}
  In particular, $h_0$ is also a reference operator for $H_\eps$, and
  the associated scattering operator is given by
  \begin{equation*}
  S_\eps :=  S(H_\eps, h_0, I_\eps J)
    = \bigl(W_+(H_\eps, h_0, I_\eps J) \bigr)^*
         \circ W_-(H_\eps, h_0, I_\eps J). 
  \end{equation*}
  For  $u, v \in \HS_\ac(h_0)$, we then have
  \begin{align*}
  &{\iprod [\HS_0] {S_\eps u} v} \\
  & =  \bigiprod [\Lsqr{\Mscr_\eps}] 
            { W_-(H_\eps,H_0,I_\eps) \cdot W_-(H_0, h_0, J) u}
            { W_+(H_\eps,H_0,I_\eps) \cdot W_+(H_0, h_0, J) v}
            \\
  & = \bigiprod [\Lsqr{\Mscr_0}] 
           {I_\eps^* \cdot W_-(H_\eps,H_0,I_\eps) \cdot W_-(H_0, h_0, J) u}
           {I_\eps^* \cdot W_+(H_\eps,H_0,I_\eps) \cdot W_+(H_0, h_0, J) v}. 
  \end{align*}
  Using \Lem{5.6} below, we find that
  \begin{equation*}
    {\iprod  [\HS_0] {S_\eps u} v}
    \to \bigiprod [\Lsqr{\Mscr_0} ] {P_\ac(H_0) \circ W_-(H_0, h_0, J)u}
               {P_\ac(H_0) \circ W_+(H_0, h_0, J)v} 
    =   {\iprod [\HS_0] {S_0 u} v} 
  \end{equation*}
  as $\eps \to 0$, where $S_0 = S(H_0, h_0, J)$.   
 The strong convergence follows from the fact that
  the operators $S(H_\eps)$ are unitary.
\end{proof}

As we observed in \Rem{main2}, \Thm{main2} immediately gives the
following stability result for the scattering matrix:
 
\begin{corollary}
  \label{cor:main2}
  Let $M$, $g_0$ and $H_0$ as above and suppose that for a given pair
  of indices $i,k \in \{1, \ldots, \ell \}$ we have
  \begin{equation}
    \label{eq:scatt-non-triv}
    S_{ik}(H_0, h_0, J) \ne 0.
  \end{equation}
  Then for any $\gamma > 0$ fixed, there exists $\eps_0 > 0$ such that
  $S_{ik}(H_\eps, h_0, I_\eps J) \ne 0$ for all metrics $g_\eps \in
  \Met_{r_0}(M,g_0,\gamma,\eps)$ and all $0< \eps \le \eps_0$.
\end{corollary}

\begin{remark}
  \begin{myenumerate}{\alph}
  \item As discussed in \Rem{ref.op}, the property $S_{ik}(H_0,h_0,J)
    \ne 0$ is independent of the choice of the reference operator
    $h_0$.

  \item It is not easy to establish property~\eqref{eq:scatt-non-triv}
  in concrete situations. We will give simple examples with $\ell=2$ in
  \Sec{example} where we exploit rotational symmetry.
 \end{myenumerate}
\end{remark}

{\colour[A11] The following lemma has been used in the proof of
  \Thm{main2}.  We consider here a family of metrics $(g_\eps)_\eps$
  converging to a metric $g_0$ in the sense of~\eqref{eq:met.ball}, so
  that, in particular, $\wt d_1(g_\eps,g_0) \le \eps$.}

\begin{lemma}
  \label{lem:5.6}
  With the assumptions and notation of \Thm{main2}, let $H_\eps$ be
  the Laplacian of $\Mscr_\eps=(M, g_\eps)$ associated with a metric
  $g_\eps \in \Met_{r_0}(M,g_0,\gamma,\eps)$, for $0 < \eps < \eps_0$.
  Then
  \begin{equation}
    \label{eq:5.3}
    \slim_{\eps \to 0} I_\eps^* W_\pm(H_\eps, H_0, I_\eps) = P_\ac(H_0).
  \end{equation}
\end{lemma}
\begin{proof}
  For $\eps \ge 0$ set $R_\eps := (H_\eps + 1)^{-1}$, and let $m :=
  [n/4] + 2$.  Defining $V_\eps$ in analogy with $V$ in~\eqref{eq:2.5}
  by
  \begin{equation*}
    V_\eps := R_\eps^m (H_\eps I_\eps - I_\eps H) R_0^m,
  \end{equation*}
  it follows from our assumptions and \Cor{res.diff} that
  $\norm[\TRspace]{V_\eps} \to 0$ as $\eps \to 0$.  This implies
  (cf.~the Corollary following Theorem~XI.7 in~\cite{reed-simon-3}) 
  that
  \begin{equation}
    \label{eq:5.4}
    \normsqr{W_\pm(H_\eps,H_0, R_\eps^m I_\eps R_0^m) \phi 
           - R_\eps^m I_\eps R_0^m P_\ac(H_0) \phi} 
    \le 16 \pi \norm[\TRspace]{V_\eps} \cdot \triplenormsqr \phi \cdot
           \norm{I_\eps} \to 0
  \end{equation}
  as $\eps \to 0$, for all $\phi \in \mathcal M(H_0)$ that satisfy
  $\triplenorm \phi < \infty$; the notation is as
  in~\cite{reed-simon-3}.  Note that, by assumption and
  \Prp{eq.norms}, $\norm {I_\eps} \le 1+\wt d_\infty(g_0,g_\eps) \le
  1+\gamma$ independently of $\eps$.  By the intertwining
  relations~\cite{reed-simon-3},
  \begin{equation}
    \label{eq:5.5}
    W_\pm(H_\eps, H_0, I_\eps) R_0^{2m}
    = W_\pm(H_\eps, H_0, R_\eps^m I_\eps R_0^m).
  \end{equation}
  Furthermore, the arguments used in \Sec{ex.wo} yield
  \begin{equation}
    \label{eq:5.6}
    \norm{R_\eps^m I_\eps R_0^m - I_\eps R_0^{2m}} \to 0, \qquad \eps
    \to 0;
  \end{equation}
  a proof of \Eq{5.6} will be given below.  We conclude from
  \eqref{eq:5.4}--\eqref{eq:5.6} that for all $\phi$ with
  $\triplenorm\phi < \infty$
  \begin{equation*}
    \norm{W_\pm(H_\eps, H_0, I_\eps) R_0^{2m} \phi
          - I_\eps P_\ac(H_0) R_0^{2m}\phi } \to 0,
    \qquad \eps \to 0.
  \end{equation*}
  The set of vectors $\set{R_0^{2m} \phi} {\phi \in \Lsqr{\Mscr_0},
    \triplenorm \phi < \infty}$ is dense and~\eqref{eq:5.3} follows.

  It remains to prove~\eqref{eq:5.6}.  Here we first note that, by a
  standard expansion,
  \begin{equation*}
    (R_\eps^m I_\eps - I_\eps R_0^m) R_0^m
    = - \sum_{j=1}^m R_\eps^j (H_\eps I_\eps - I_\eps H_0) R_0^{2m - j + 1};
  \end{equation*}
  therefore, it is clearly enough to show that $\norm{ R_\eps (H_\eps
    I_\eps - I_\eps H_0) R_0^{m+1}} \to 0$, as $\eps \to 0$.  By a
  simple variant of \Lem{res.diff}, we have
  \begin{equation*}
    R_\eps (H_\eps I_\eps - I_\eps H_0) R_0^{m+1}
    = \bigl(\hat B_\eps^{(1)}\bigr)^* \hat U \hat B_\eps^{(m+1)}
            - \bigl(B_\eps^{(1)}\bigr)^* U B_\eps^{(m)} R_0H_0,
  \end{equation*}
  where $B_\eps^{(m)} = \abs{S_\eps}^{1/2} R_\eps^{m}$ and $\hat
  B_\eps^{(m+1)} = \abs{\hat S_\eps}^{1/2} D_\eps R_\eps^{m+1}$ (see
  \Sec{ex.wo}).  From \Prp{hs} we now conclude 
  \begin{equation*}
    \norm[\HSspace]{R_\eps (H_\eps I_\eps - I_\eps H_0) R_0^{m+1}}
    \le  2 (\gamma C \eps)^{1/2} \to 0
  \end{equation*}
  as $\eps \to 0$, using the norm bound $\normsqr {B_\eps^{(1)}} \le
  \norm[\infty] {S_\eps} \le \gamma$ (see \Lem{s.est}) and similarly
  for $\hat B_\eps^{(1)}$.  Moreover, the convergence in
  Hilbert-Schmidt norm implies the convergence in operator norm
  and~\eqref{eq:5.6} follows.
\end{proof}

To conclude this section, let us comment on the metric structure of
the spaces of metrics used so far.
\begin{remark}
  \label{rem:d1.q-met}
  \begin{myenumerate}{\alph}
  \item On $\Met(M,g_0,\gamma):=\set{g \in \Met(M)}{\wt
      d_\infty(g_0,g) \le \gamma}$, the distance function $\wt
    d_\infty$ is a quasi-distance with constant $\tau=1+\gamma/2$
    (cf.\ \Def{q-metric} and what is said before that definition).  A
    quasi-distance induces a unique topology and a uniform structure
    (cf.\ the comment after \Def{q-metric}).

  \item
  \label{d1.q-met.b}
  Moreover, $\wt d_1$ is a quasi-distance on the subspace of the
  metrics $g \in \Met(M,g_0,\gamma)$ for which $\wt d_1(g_0,g)$ is
  finite. Indeed, straightford calculations, using~\eqref{eq:met.est}
  for the estimates on $\rho$ and the fact that $\wt d$ is a
  quasi-distance with constant $\tau=1+\gamma/2$, show that
  \begin{align*}
    \wt d_1(g_1,g_3) \le  \tau (1 + \gamma)^4
      \bigl(\wt d_1(g_1,g_2) + \wt d_1(g_2,g_3) \bigr). 
  \end{align*}
    
  On $\Met(M,g_0,\gamma)$, we can also work with the quasi-distance
  \begin{equation}
    \label{eq:eq.q-met.d1}
    \wt d_1^*(g_1,g_2)
    := \int_M \wt d(g_1,g_2) r_0^{-(n+2)} \dvol_{g_0}
  \end{equation}
  (with constant $\tau$) which is equivalent to $\wt d_1$.  We prefer
  to work with $\wt d_1$ since this is the quasi-distance appearing
  naturally in the estimates of \Prp{hs} and \Cor{res.diff}.

  The set $\Met(M,g_0, \gamma, \eps)=\set{g \in
    \Met(M,g_0,\gamma)}{\wt d_1(g_0,g)\le \eps}$ is now the closed
  $\eps$-ball with respect to the quasi-distance $\wt d_1$, and the
  set $\Met_{r_0}(M,g_0,\gamma,\eps)$ defined above is the
  intersection of this $\eps$-ball with the space $\Met_{r_0}(M)$ of
  metrics fulfilling the (local) lower bounds~\eqref{eq:lower.bd} on
  the homogenized Ricci curvature and injectivity radius.

\item
  \label{d1.q-met.c}
  On the $\eps$-ball $\Met_{r_0}(M,g_0,\gamma,\eps)$, the wave
  operators $W_\pm(\Delta_{(M,g_2)},\Delta_{(M,g_1)},I)$ exist and are
  complete for any two metrics $g_1,g_2$ in this ball due to
  \Thm{main1}.  \Thm{main2} can be restated as saying that the map $g
  \mapsto S(\Delta_{(M,g)},h_0, IJ)$ (with $\map I
  {\Lsqr{M,g_0}}{\Lsqr{M,g}}$, $If=f$), associating to a metric the
  scattering operator of $\Delta_{(M,g)}$ and a reference operator
  $h_0$ (see \Def{ref.op}), is \emph{continuous} with respect to the
  quasi-distance $\wt d_1$ and the topology of strong convergence of
  operators in $\BdOp{\HS_0}$.
    
  Moreover, \Cor{main2} can be restated by saying that the set of
  metrics $g$ such that the $k$-th and the $i$-th channel of the
  scattering operator $S(\Delta_{(M,g)},h_0, IJ)$ are open to each
  other, is a neighborhood of $g_0$ (i.e., it contains a ball of
  radius $\eps_0$ around $g_0$) in the quasi-distance space
  $(\Met_{r_0}(M,g_0,\gamma), \wt d_1)$.

  It would be interesting to analyze in more detail the structure of
  $\Met(M,g_0,\gamma,\eps)$ and $\Met_{r_0}(M,g_0,\gamma,\eps)$.
  \end{myenumerate}
\end{remark}

{\colour[A2] 
  We conclude this section with some remarks of a more general nature
  concerning the question of openness of scattering channels.
\begin{remark}
  \label{rem:when.channels.closed}
  \begin{myenumerate}{\alph}
  \item One might conjecture that an end can only be closed if the
    (decoupled) Laplacian of this end has no absolutely continuous
    spectrum.  In particular, this would mean that there are no
    geometric or topological obstructions that might prevent sending
    wave packets from the $i$-th channel into any of the other
    channels. We are not aware of any counter-examples. 

  \item Another conjecture would say that---as long as the ends have
    some absolutely continuous spectrum---condition~\eqref{eq:ch.open}
    holds in a generic sense. Our work provides a step in this
    direction since we show that the set of metrics enjoying
    property~\eqref{eq:ch.open} for all pairs $(i,k)$ is open in the
    sense of \Remenum{d1.q-met}{d1.q-met.c}; whether this set is also
    dense is a rather difficult question.

  \item In this paper we consider the openness of the scattering
    channels in a global sense, i.e., we ask---without imposing any
    restrictions on the energy of the wave packets---whether the
    scattering channels are open.  From a physical point of view, the
    following question of \emph{local openness} would also be of great
    interest. Suppose we restrict our attention to a \emph{compact set
      with respect to energy,} i.e, we consider wave packets which are
    localized in Fourier space.  If such a wave packet comes in from
    the $i$-th channel, how does it split up as $t$ tends to
    $+\infty$?  Here it might happen that scattering channels that are
    open for wave packets with energy in a compact set, $K$, are no
    longer open for such wave packets if the metric is perturbed in
    such a way that the set $K$ and the absolutely continuous spectrum
    of the ``receiving channel'' become disjoint.

  \item Strong distortions of the metric on one of the ends may
    destroy the a.c.\ spectrum there and then this end (with the new
    metric) would be closed for scattering. It is easy to construct
    examples of the type discussed in \Sec{example} below where the
    radial function $s \mapsto r(s)$ is distorted in such a way that
    the corresponding potential $w$ in eqn.~\eqref{eq:6.3} changes
    from short range to a potential that tends to $+\infty$ as $s \to
    \infty$; consider, e.g., $r(s) := e^{-s^2}$ for the distorted
    radial function.

    It is a much harder question to ask whether a channel can close
    while staying within a given class $\Met(M,g_0,\gamma,\eps_0)$,
    maybe with some large $\eps_0$; note that for metrics from this
    class the a.c.\ spectrum is stable.
  \end{myenumerate}
\end{remark} %
}
%
\section{Examples}
\label{sec:example}
%

In this section, we present some examples where our main results,
\Thm{main1} and \Cor{main2}, can be applied. We only consider
manifolds with two ends; dealing with more than two ends would require
additional efforts.

\subsection{Surfaces of revolutions and warped products}
\label{sec:warped}
As a preparation, we first recall and establish some facts on
manifolds which are symmetric with respect to rotation around some
axis (\emph{surfaces of revolution} for $n=2$).  Such manifolds are
special cases of warped products~\cite{oneill:83}.  %
In some of our examples such manifolds are used to find a suitable
function $r_0$ and in other examples we consider perturbations of
manifolds of this type.

A particularly simple case is given by the following situation: Let
$\Mscr_0$ denote an $n$-dimensional manifold in $\R^{n+1}$ which is
symmetric with respect to rotation around the $x_{n+1}$-axis and
homeomorphic to $\R \times \Sphere^{n-1}$.  Writing $\bar x = (x_1,
\ldots, x_n)$, we introduce coordinates $(s, \xi)$, $\xi = |\bar x
|^{-1} \bar x \in \Sphere^{n-1}$ for $\bar x \ne 0$, where $s$ is
arc-length along any line $\xi = \const$.  Defining $r(s) = \abs{\bar
  x}$, we assume that $r$ is a positive function of class $\Contspace
[2]$.  In these coordinates, the Riemannian manifold $\Mscr_0 =
(M,g_0)$ is given by $M=\R \times \Sphere^{n-1}$ and $g_0 = \de s^2 +
r(s)^2 g_{\Sphere^{n-1}}$, where $g_{\Sphere^{n-1}}$ denotes the
standard metric on $\Sphere^{n-1}$.  More generally, we can start with
a metric $g_0$ and define the Riemannian manifold $(M,g_0)$
abstractly, without referring to the ambient space $\R^{n+1}$ (this is
necessary, e.g., if $r(s)$ grows fast).  Some of the results in this
section remain true if we replace $\Sphere^{n-1}$ by any compact
Riemannian manifold $Y$; for simplicity, we only treat the case
$Y=\Sphere^{n-1}$ here.

As is well known, there is a unitary operator $\map U {\Lsqr \Mscr}
\Lsqr{\R , \Lsqr {\Sphere^{n-1}}}$ with the property that
\begin{equation}
  \label{eq:6.1}
  H = U^* {\tilde H} U, \qquad
  \tilde H = \mybigoplus_{m=0}^\infty \tilde H_m, 
\end{equation}
where
\begin{equation}
  \label{eq:6.2}
   \tilde H_m = - \frac {\dd^2} {\dd s^2}+ \pot(r)  +
     \frac{\lambda_m} {r ^2}, \qquad m \in \N_0, 
\end{equation}
a self-adjoint operator in $\Lsqr{\R, E_m}$, and
\begin{equation}
  \label{eq:6.3}
   \pot(r) = \frac{n-1}{2} \dot q
      + \Bigl(\frac {n-1}2 \Bigr)^2 q^2,
    \qquad
   q = \frac{\dot r} r;
\end{equation}
here, $E_m$ denotes the eigenspace associated with $\lambda_m$, the
$m$-th eigenvalue of the Laplacian on the sphere $\Sphere^{n-1}$.  We
have $\dim E_0 = 1$ and $\lambda_0 = 0$ with constant eigenfunction.
The Ricci curvature (viewed as a symmetric tensor on $T^* M$) is given
by
\begin{equation*}
  \Ric = -(n-1) (\dot q + q^2) \dd s^2 + 
  \Bigl(\frac {n-2} {r^2} - \bigl( \dot q + (n-1) q^2 \bigr) \Bigl)
          r^2 g_{\Sphere^{n-1}}. 
\end{equation*}
Similarly, we obtain for the sectional curvature 
\begin{equation*}
  K(\partial_s,\partial_{\xi_j})
  = -\frac{\ddot r}{r}
   \qquad\text{and}\qquad
  K(\partial_{\xi_j},\partial_{\xi_k})
  = \frac{1-\dot r^2}{r^2}
\end{equation*}
for $i \ne j$ (provided $n \ge 3$), where $\{\partial_{\xi_j}\}_j$ is
a basis of $T_x \Sphere^{n-1}$; cf., e.g.,
\cite[p.~209ff]{oneill:83}. In particular, if $n=2$, we have
\begin{equation*}
  \Ric = -(\dot q + q^2) g = -\frac{\ddot r}{r} g
   \qquad\text{and}\qquad
  K(\partial_s,\partial_{\xi})
  = -\frac{\ddot r}{r}.
\end{equation*}

\subsection{Reference operators on the two ends}
\label{sec:ref.op.2}
We decouple $\Mscr=(M,g)$ into just two pieces by a surface $\Sigma
\subset M$ which corresponds to $s = 0$.  With the same unitary
operator $U$ as in~\eqref{eq:6.1}, the decoupled operators $H_\dec$
and $H_{\dec,m}$ satisfy
\begin{equation*}
  H_\dec = U^* {\tilde H}_\dec U, \qquad
  \tilde H_\dec = \mybigoplus_{m=0}^\infty \tilde H_{\dec,m},  
\end{equation*}
with
\begin{equation*}
  \tilde H_{\dec, m}
  = \tilde H_{-, m} \oplus \tilde H_{+, m}; 
\end{equation*}
here $\tilde H_{-,m}$ and $\tilde H_{+, m}$ are,
respectively, the operator~\eqref{eq:6.2} in $\Lsqr{\R_-, E_m}$ and in
$\Lsqr{\R_+, E_m}$ with Dirichlet boundary condition at $s=0$, and
$\R_- :=(-\infty,0)$ and $\R_+ :=(0,\infty)$. We then have:

\begin{myenumerate}{\alph}
\item Applying a celebrated result of Deift and
  Killip~\cite{deift-killip:99} to our situation, we see that $\tilde
  H_{+, m}$ (respectively, $\tilde H_{-, m}$) has absolutely
  continuous (a.c.)  spectrum $[0,\infty)$, provided the potential
  $\pot + \lambda_m/r^2$ is square integrable over $(1,\infty)$
  (respectively, over $(-\infty, -1)$).

\item If $r(s) \to 0$ as $s \to \pm \infty$, the decoupled operators
  $H_{\pm, m}$ have purely discrete spectrum for $m \ge 1$ (while
  $H_{\pm,0}$ still may have a non-trivial a.c.\ part).
  Br\"uning~\cite{bruening:89} provided conditions for Laplacians on
  cusps that guarantee purely discrete spectrum.

\item We will mostly work with the following assumption on $\pot$: we
  say that $\pot$ is \emph{short range} on $\R$ if there exists a
  constant $\alpha > 1$ such that
  \begin{equation}
    \label{eq:s-range}
    \abs{\pot(s)} \leq C(1 + |s|)^{-\alpha}.
  \end{equation}
  Under this assumption, the Enss method yields existence and
  completeness of the wave operators for the pair $(h_0, h_0 + \pot)$
  and the absence of singular continuous spectrum for $h_0 + \pot$,
  with $h_0$ denoting the unique self-adjoint realization of $-
  \dd^2/\dd s^2$ in $\Lsqr \R$ with domain $\Sobsymb^2(\R)$; cf.,
  e.g., \cite{enss:78}, \cite{ simon:79c}.  Analogous results hold for
  the operators $\tilde H_{\pm,m}$, $m \ge 1$, provided $\pot +
  \lambda_m/r^2$ is short-range.
\end{myenumerate}

The cases where $r$ is of the form $r(s) = s^\beta$, for some $\beta \in
\R$ and $s \ge 1$, are particularly simple. Let
\begin{equation*}
  r(s):= \tau s^\beta, \qquad s \ge 1, 
\end{equation*}
for some $\beta \in \R$ and $\tau >0$. In this case,
\begin{equation*}
  q(s) = \frac \beta s,  \qquad 
  \pot(s) = \Bigr(\frac {n-1} 2 \beta - 1 \Bigl)
         \frac{n-1}2 \beta \frac 1 {s^2},  
\end{equation*}
and the potentials $\pot$ are short-range.  Let us describe the geometry
of such an end in more detail in the case $n=2$.  The Ricci curvature
is then given by
\begin{equation*}
  \Ric = -(\beta-1) \beta  \frac 1 {s^2} \cdot g.
\end{equation*}
In particular, the curvature tends to zero as $s \to \infty$.

If $\beta>1$, the end is large and negatively curved.  If $\beta=1$,
we have a Euclidean (flat) end for $\tau=1$ and a cone for $0<\tau<1$.
If $0<\beta<1$, we have a positively curved parabolic-type end.  If
$\beta=0$, we have a flat cylinder of radius $\tau>0$, and if $\beta <
0$, we have a negatively curved shrinking horn.

On the real line, one may combine different asymptotics of the above
type for $s \ge 1$ and $s \le -1$. More precisely, we may consider
$\Mscr_0 = (M, g_0)$ with two ends and the decomposition $M= M_- \cup
\clo M_0 \cup M_+$ with $M_- = (-\infty, -1) \times \Sphere^{n-1}$,
$M_0=(-1,1) \times \Sphere^{n-1}$, and $M_+ = (1, \infty) \times
\Sphere^{n-1}$.  We then consider $\Contspace [2]$-functions on the
real line given by
\begin{equation*}
  r(s) =
  \begin{cases}
    \tau_- \abs s^{\beta_-},& \text{if $s \le -1$,}\\
    \tau_+  s^{\beta_+},& \text{if $s \ge 1$,}
  \end{cases}
\end{equation*}
for some constants $\beta_\pm \in \R$ and $\tau_\pm>0$; for $-1<s<1$,
the function $r > 0$ can be defined arbitrarily. In all these cases we
have two scattering channels, one for $\R_+$ and one for $\R_-$.

\subsection{Perturbations of warped products: Existence and
  completeness of wave operators (Applications of \Thm{main1})}
\label{sec:ex.thm.main1}

We begin with the slightly simpler case where we assume a global 
bound for the sectional curvature. In dimension $n=2$, this 
assumption covers warped products with at most exponential growth 
or decay, while it excludes the cases $r(s) \to 0$ in dimension $n \ge 3$. 

\begin{example}[Existence and completeness of wave operators]
  \label{ex:6.1}
  Suppose we are given a manifold $\Mscr_0 = (M,g_0)$ with a warped
  product metric $g_0 = \dd s^2 + r(s)^2 \dd g_{\Sphere^{n-1}}$
  satisfying the following conditions: $r \in \Cont[2] \R$, the
  functions $\ddot r/r$ and, for $n \ge 3$, $(1 + \dot r^2)/r^2$, are
  bounded, and there exists a constant $m \ge 1$ such that
  \begin{equation}
    \label{eq:6.7}
    \frac 1 m r(s_0)
    \le r(s)
    \le m r(s_0),
    \qquad \forall s \in [s_0 -2, s_0 + 2], 
  \end{equation}
  for all $s_0 \in \R$.  Notice that $\Mscr_0$ has bounded sectional
  curvature.  It follows from \Lem{d.2} and \Rem{d.5} that
  $\iota_{\Mscr_0}(s)$, the (homogenized) {\colour[A3] injectivity
    radius} of $\Mscr_0$ at points $(s,y) \in M$, is bounded below by
  $r_0(s) := c_0\min\{r(s), 1\}$, where $c_0 > 0$ is a suitable
  constant; without loss of generality, we may assume $c_0 \le 1$.

  Now consider two metrics $g_1$ and $g_2$ on $M$ with the following
  properties:
  \begin{enumerate}
  \item
    \label{met1}
    $g_1$ and $g_2$ are quasi-isometric to $g_0$.
  \item
    \label{met2}
    $\Mscr_k = (M, g_k)$ have bounded sectional curvature.
    
  \item
    \label{met3}
    We assume that
    \begin{equation}
      \label{eq:6.8}
      \wt d_1^*(g_1,g_2)
      = \int_M {\wt d(g_1,g_2)(x)}{r_0(x)^{-(n+2)}} \dvol g_0(x) < \infty, 
    \end{equation}
    with the function $r_0$ obtained above, i.e., we assume that the
    quasi-distance $\wt d_1^*(g_1,g_2)$ as defined in
    eqn.~ is finite.  Note that this
    quasi-distance is equivalent to the original quasi-distance $\wt
    d_1$ defined in eqn.~\eqref{eq:def.q-dist.loc}, 
    cf.~\Remenum{d1.q-met}{d1.q-met.b}.
  \end{enumerate}
  Then the wave operators for the Laplacians $H_1$ and $H_2$,
  associated with the metrics $g_1$ and $g_2$, exist and are complete.

  Indeed, it follows from~\itemref{met1},~\itemref{met2}
  and~\cite[Prop.~2.1]{mueller-salomonsen:07} (cf.~also \Prp{d.1})
  that the {\colour[A3] injectivity radius} for $\Mscr_{1,2}$ is
  bounded from below by $c r_0(x)$ for some positive constant $c$.
  By~\eqref{eq:6.7}, a similar lower bound holds for the homogenized
  {\colour[A3] injectivity radius} of $\Mscr_{1,2}$.  Then
  \Prp{lower.bd.rh} yields the lower bound for the harmonic radius
  required in \Thm{main1}.  (Note that neither $g_1$ nor $g_2$ have to
  be close to $g_0$ at infinity.)
 \end{example}
 {\colour[A1] 
   In the following remark we show that in some particularly simple
   cases the distance between two metrics $g_1$ and $g_2$ can be
   computed more or less explicitly.

\begin{remark}
  \label{rem:computation-of-metric} 
  Suppose we are given two quasi-isometric metrics $g_1$ and $g_2$ on
  $M = \R \times \Sphere^{n-1}$ of the form $g_i = \dd s^2 + r_i(s)^2
  \dd g_{\Sphere^{n-1}}$, for $i = 1,2$, where the functions $r_i$
  satisfy the same conditions as $r_0$ in \Ex{6.1}. In particular, the
  manifolds $\Mscr_i = (M,g_i)$ have bounded sectional curvature and
  their (homogenized) injectivity radius is bounded below by $c
  r_i(s)$ for some constant $c > 0$.  We may thus work with $r_0(s) :=
  \min\{1, r_1(s), r_2(s)\}$ in eqn.~\eqref{eq:met.l1}; note that the
  function $r_0$ defined here may be different from the function $r_0$
  of \Ex{6.1}.

  It is easy to see that the matrix $A$ from
  eqn.~\eqref{eq:the-matrix-A(x)} has one eigenvalue $1$ while the
  other $n-1$ eigenvalues are equal to $(r_2/r_1)^2$ so that
  \begin{equation*}
    \wt d(g_1,g_2)
    = \Bigabs{\Bigl(\frac{r_2}{r_1}\Bigr)^{n/2}-
      \Bigl(\frac{r_1}{r_2}\Bigr)^{n/2}}.  
  \end{equation*}
  We can now compute $ \wt d_1^*(g_1,g_2)$ (as in
  eqn.~\eqref{eq:eq.q-met.d1}) as
  \begin{equation*}
    \wt d_1^*(g_1,g_2) = 
    \omega_{n-1}
    \int_{-\infty}^\infty  \Bigabs{\Bigl(\frac{r_2}{r_1}\Bigr)^{n/2}-
      \Bigl(\frac{r_1}{r_2}\Bigr)^{n/2}}
    \frac 1 {(\min\{1,r_1,r_2\})^{n+2}} \dd s;   
  \end{equation*}
  since $g_1$ and $g_2$ are quasi-isometric, there is a constant $ c_1
  \ge 1$ depending only on the quasi-isometric distance $\wt
    d_\infty(g_1,g_2)$ such that $ \wt d_1^*(g_1,g_2) \le \wt
  d_1(g_1,g_2) \le c_1 \wt d_1^*(g_1,g_2)$.
\end{remark}%
}

We next discuss examples without a global curvature bound:
\begin{example} (Existence and completeness of wave operators without
  global curvature bound).
  \label{ex:6.2}
  We start with a manifold $\Mscr_0$ with a warped product metric
  obtained from a function $r \in \Cont[2] \R$ satisfying
  condition~\eqref{eq:6.7}. We define the function $\kappa_0 \colon \R
  \to \R_+$ as in eqn.~\eqref{eq:d.5}.

  Given $r$, we then consider metrics $g_1$, $g_2$ as in \Prp{d.3}.
  In particular, $g_1$ and $g_2$ have to be quasi-isometric to $g_0 =
  \dd s^2 + r(s)^2 g_{\Sphere^{n-1}}$ and there is a (more
  complicated) condition on the sectional curvature of $\mathcal
  M_1:=(M,g_1)$ and $\mathcal M_2:=(M,g_2)$ expressed in terms of
  functions $\kappa_1$ and $\kappa_2$ on the real line, defined as in
  assumption~\itemref{prp.d.3.ii} of \Prp{d.3}.  By this proposition
  and \Rem{d.5}, the (homogenized) {\colour[A3] injectivity radius}
  and then also the (homogenized) harmonic radius of $\Mscr_1$ and
  $\Mscr_2$ at points $(s,y) \in M$ are bounded below by
  \begin{equation*}
    r_0(s)
    := C \min \bigl\{r(s),
         \min  \bigset{\bigl(\sqrt{\kappa_i(s)}\bigr)^{-1}} 
         {i = 0, 1, 2} \bigr\},
  \end{equation*}
  for some positive constant $C$. Proceeding now as in the above
  example, we find that the wave operators for the Laplacians $H_1$
  and $H_2$ exist and are complete if $g_1$ and $g_2$ satisfy
  condition~\eqref{eq:6.8}.
\end{example}

It may be of interest to note that lower bounds for the radius of 
injectivity are the main limitation and difficulty in the 
application of our theorems to concrete examples.

\subsection{Non-trivial scattering for warped products}
\label{sec:non-triv.scatt}

 Here we provide examples of functions $r$ on the real line which yield
 scattering channels that are open to each other.  We use the notation
 from the beginning of this section concerning ${\tilde H}_{\pm,m}$
 etc.  We give the details in the case of the scattering channels,
 $\tilde{H}_{\pm,0}$, with zero angular momentum. The case of
 scattering channels with non-zero zero angular momentum follows in the
 same way.

 We require that the potential $\pot$, defined in eqn.~\eqref{eq:6.3},
 is short-range, i.e., that it satisfies~\eqref{eq:s-range}.  Note 
 that the ends can be flat, horns, etc., as long as $\pot$ is short-range.

 We first introduce some notation that we need (see \App{line}).
 Denote by $h_0$ the unique self-adjoint realization of $-\dd^2/\dd
s^2$ in $\mathcal H:=\Lsqr \R$ with domain $\Sob [2] \R$.  We define
$ h: = h_0+w$, a self-adjoint operator in $\mathcal H$, with domain
$\dom h = \Sob [2] \R$.  Let $h_{\pm,0}$ be the self-adjoint
realizations of $-\dd^2/\dd s^2$ in $\Lsqr{\R_\pm}$ with Dirichlet
boundary condition at zero, i.e., with domain $\dom h_{\pm,0} =
\Sobn{\R_\pm} \cap \Sob [2] {\R_\pm}$. We denote by $\iota_\pm$ the
natural embeddings of $\Lsqr{\R_\pm}$ into $\Lsqr \R$ (extension by zero). 
  We then consider the wave operators
\begin{equation*}
  \Lambda^+_\pm := \slimpm  \e^{\im th} \iota_+ \e^{-\im t h_{+,0}},
  \quad\text{and}\quad
  \Lambda^-_\pm := \slimpm  \e^{\im th} \iota_- \e^{-\im t h_{-,0}}.
\end{equation*}
We define a mapping $j_0$ which associates with a function $u \in
\Lsqr \R$ the function $j_0 u \in \Lsqr{\R, E_0}$, defined by
\begin{equation*}
  (j_0 u)(s,y) = \frac 1 {\sqrt{\omega_{n-1}}} u(s) , 
  \end{equation*}
where $\omega_{n-1} := \vol_{n-1}(\Sphere^{n-1})$, 
and define $J := U^* j_0$, an isometry from $\Lsqr \R$ into $\Lsqr
\Mscr$.

A natural choice of a reference operator for the ends are the
operators $h_{\pm,0}$.  We then define the wave operators for the
right end as
\begin{equation}
  \label{eq:5.7}
  W_\pm^+ := \slimpm \e^{\im t H} J\iota_+ \e^{-\im th_{+,0}}, 
\end{equation}
and for the left end as
\begin{equation}
  \label{eq:5.8}
  W_\pm^- := \slimpm \e^{\im t H} J\iota_- \e^{-\im th_{+,0}}. 
\end{equation}
By the chain rule, we have
\begin{equation}
  \label{eq:5.9}
  W^+_\pm= W_\pm(H,h, J) \circ \Lambda^+_\pm
  \quad\text{and}\quad
  W^-_\pm= W_\pm(H,h, J) \circ \Lambda^-_\pm.
\end{equation}
Furthermore, as $U J = UU^* j_0 = j_0$, 
\begin{equation}
  \label{eq:5.10}
  W_\pm(H,h, J)
  = U^* W_\pm({\tilde H}, h, j_0)
  = U^* W_\pm({\tilde H}_0, h, j_0) = U^* j_0 = J, 
\end{equation}
since $\e^{\im t{\tilde H}_0} j_0 u = j_0 \e^{\im th} u$.
By~\eqref{eq:5.9} and~\eqref{eq:5.10}, we have
\begin{equation}
  \label{eq:5.11}
  W^+_\pm= J \circ \Lambda^+_\pm
  \quad\text{and}\quad
  W^-_\pm= J \circ \Lambda^-_\pm.
\end{equation}
Since $J$ is an isometry, it follows from \Lem{c.3}
and~\eqref{eq:5.11} that 
\begin{equation}
  \label{eq:5.12}
  \ran W^+_+ \ne \ran W^+_-
  \quad\text{and}\quad
  \ran W^-_+ \ne \ran W^-_-,
\end{equation}
which shows that the left and the right scattering channels are open
to each other, i.e.,
\begin{equation}
  \label{eq:5.13}
  S_{+,-} :=\bigl(W^+_+\bigr)^* W^-_- \neq 0.
\end{equation}
Let us give an example of a high-velocity asymptotic state which comes
in from the left end for large negative times, and which travels to
the right end as time tends to infinity, with a reflected part that is
very small if the velocity is large.  As in \App{line} we define
\begin{equation*}
  \phi^+_v(s):= \e^{\im v s} \phi_0,
  \quad \text{with\ } \mathcal F \phi_0 \in \Cci \R, 
\end{equation*}
where $\mathcal F$ denotes the Fourier transform.  Since
\begin{equation*}
  \bigl(\mathcal F \phi^+_v\bigr)(k)
  = \bigl(\mathcal F\phi_0 \bigr)(k-v),
\end{equation*}
this state has large velocity if $v>0$ is taken large enough.  We
set
\begin{equation*}
  \phi^+_{v,0} :=  \phi^+_v(s)-\phi^+_v(-s), \quad  s>0,
\end{equation*}
then by~\eqref{eq:5.11} and~\eqref{eq:c.54}--\eqref{eq:c.55} we have
\begin{equation*}
  S_{+,-} \phi^+_{v,0} =
  \bigl(\Lambda^+_+\bigr)^*  \Lambda^-_-\phi^+_{v,0}
  =\phi^+_{v,0}+ \mathcal O \Bigl(\frac 1v \Bigr),
\end{equation*}
and $ S_{+,-}\phi^+_{v,0}   =\bigl(\Lambda^+_+\bigr)^* \Lambda^-_-
\phi^+_{v,0} \ne 0$ for $v$ large enough.

\subsection{Perturbations of warped products: Open scattering channels
  (Applications of \Cor{main2})}
\label{sec:ex.cor.main2}

We now consider perturbations of the rotationally symmetric 
situation discussed in the preceding subsection. 

\begin{example}[Open scattering channels]
  \label{ex:6.3}
  For simplicity, let $n=2$ and let us assume that the metric $g_0$ is
  a warped product obtained from a function $r \colon \R \to \R_+$ of
  class $\Contspace [2]$ satisfying $r(s) = |s|$ for $s \le -1$ and
  that $r(s) = \tau s^{\beta}$, for $s \ge 1$, for some $\beta < 0$
  and $\tau > 0$.  We let $\Mscr_0 = (M,g_0)$, a manifold of bounded
  sectional curvature. Let $K \ge 0$ denote an upper bound for the
  curvature of $\Mscr_0$.  By the above discussion, the scattering
  channels for the Laplacian of $\Mscr_0$ are open.  As in \Ex{6.1},
  we find that the (homogenized) {\colour[A3] injectivity radius} of
  $\Mscr_0$ has a lower bound of the form $c_0 \min\{1, r(s) \}$ for
  $s \in \R$, with a constant $c_0 > 0$; cf.~\Lem{d.2}.  Since the
  Ricci curvature is of order $s^{-2}$ as $s \to \infty$, and as it is
  equal to zero for $ s\leq-1$, we may choose $r_0(s)= c_0$ for $s
  \leq -1$ and $r_0(s)= \tau s^\beta$ for $s \geq 1$
  in~\eqref{eq:6.8}.

  For some $\gamma > 0$ and $K$ as above, let us consider the class
  $\Met_{r_0}(M, g_0, \gamma, K, \eps)$ consisting of all metrics $g
  \in \Met_{r_0}(M, g_0, \gamma, \eps)$ with sectional curvature
  bounded by $K$.  Since any $g \in \Met_{r_0}(M, g_0, \gamma, K,
  \eps)$ is quasi-isometric to $g_0$ (with relative constants
  depending only on $\gamma$), Proposition~2.1
  of~\cite{mueller-salomonsen:07} (cf.\ also \Prp{d.1}) implies that
  the {\colour[A3] injectivity radius} of $g$ is bounded below by $c
  r_0$, for some constant $c > 0$; it is easy to see that a similar
  estimate then holds for the homogenized {\colour[A3] injectivity
    radius} as well.  As in \Ex{6.1}, the wave operators for the pair
  $H$ and $H_0$ exist and are complete. From \Cor{main2} we infer that
  there exists $\eps_0 > 0$ such that the scattering channels for the
  Laplacian of $\Mscr = (M,g)$ are open for any metric $g \in
  \Met_{r_0}(M, g_0, \gamma, K, \eps)$, provided $\eps < \eps_0$.
\end{example}

Here is, finally, an example for \Cor{main2} without a global bound on
the curvature:

\begin{example}[Open scattering channels without global curvature
  bound]
  \label{ex:6.4}
  Let $n \in \N$, $n \ge 2$. For $r \colon \R \to \R_+$, $r \in
  \Contspace [2]$ satisfying~\eqref{eq:6.7}, let $g_0 := \dd s^2 +
  r(s)^2 g_{\Sphere^{n-1}}$ and assume that the scattering channels
  are open.  Define $\kappa_0 \colon \R \to \R_+$ as in
  eqn.~\eqref{eq:d.5}.  Let $\kappa \colon \R \to \R_+$ be continuous
  and satisfy $\kappa(s) \ge \kappa_0(s)$ and let $\Met_{r_0}(M, g_0,
  \gamma, \kappa, \eps)$ denote the set of all metrics $g \in
  \Met_{r_0}(M, g_0, \gamma, \eps)$ which satisfy the curvature
  condition \eqref{prp.d.3.ii} of \Prp{d.3}.  By this proposition and
  \Rem{d.5}, the (homogenized) harmonic radius of any $\Mscr = (M,g)$
  with $g \in \Met_{r_0}(M, g_0, \gamma, \kappa, \eps)$ at points
  $(s,y) \in M$ is bounded below by
  \begin{equation*}
    r_0(s) := C \min \Bigl\{ r(s), \frac 1 {\sqrt\kappa(s)} \Bigr\}.
  \end{equation*}
  By \Cor{main2}, there exists $\eps_0 > 0$ such that the scattering
  channels of any $\Mscr = (M,g)$ with $g \in \Met_{r_0}(M, g_0,
  \gamma, \kappa, \eps)$ are open, provided $0 < \eps < \eps_0$.
\end{example}


\appendix
%
\section{Pointwise distance functions on the set of metrics}
\label{app:dist.met}
%
Let us introduce two pointwise distance functions on the set of
metrics $\Met (M)$ on a manifold $M$.  We use the terminology
``distance (function)'' for what is usually called ``metric'' in the
sense of metric spaces.  Let $V$ be an $n$-dimensional $\C$-vector
space and denote by $\Sesq_+ (V)$ the set of all positive definite
sesquilinear forms on $V$.  Given $g_1$, $g_2 \in \Sesq_+(V)$, we  
define a positive definite endomorphism $A=A_{g_2,g_1}$ on $V$ via
\begin{equation}
  \label{eq:2.met}
  g_2(\xi, \zeta) = g_1(A \xi, \zeta)
\end{equation}
for all $\xi$, $\zeta \in V$, the \emph{relative distortion} of $g_2$
with respect to $g_1$.  The \emph{distance} of $g_1$ to $g_2$ is
defined as
\begin{equation*}
  d(g_1,g_2) :=  \max_k \abs{\ln \alpha_k} , 
\end{equation*}
where $\alpha_1, \dots, \alpha_n$ denote the $n$ positive eigenvalues
of $A$;  hence $d(g_1,g_2)$ equals the operator norm of $\ln A$.
Moreover, setting $\overrightarrow {g_1g_2} := \ln A_{g_2,g_1}$, the
set $\Sesq_+(V)$ becomes an \emph{affine} space with associated vector
space $\BdOp V$ (endomorphisms on $V$).  Moreover, we have
$d(g_1,g_2)=\abs[\BdOp V]{\ln A_{g_2,g_1}}$, which shows that $d$ is
indeed a distance function on $\Sesq_+(V)$.

When dealing with Riemannian metrics (and especially with our trace
class estimate in \Sec{ex.wo}), it will be convenient to work with the
following modified distance function, namely,
\begin{equation}
  \label{eq:def.q-met}
  \wt d(g_1,g_2) := 2 \sinh \Bigl(\frac n 4 \cdot d(g_1,g_2) \Bigr).
\end{equation}
Note that $\wt d$ is symmetric and definite (i.e., $\wt
d^*(g_1,g_2)=0$ implies $g_1=g_2$), but does not fulfill the triangle
inequality.  Instead, using the addition theorem 
\begin{equation*}
  \sinh(u+v)
  = \sinh u \sqrt{\sinh^2 v + 1} + \sinh v \sqrt{\sinh^2 u + 1}
\end{equation*}
and the triangle inequality for $d$, one can see that
\begin{equation*}
  \wt d(g_1,g_3) \le \mu(\wt d(g_1,g_2), \wt d(g_2,g_3))
\end{equation*}
with $\mu(a,b) = a\sqrt{(b/2)^2+1} + b\sqrt {(a/2)^2+1}$.  Since
$\mu(a,b) \le a + b + ab =: \wt \mu(a,b)$, we can also use $\wt \mu$
instead of $\mu$.  Replacing $d$ with the equivalent (uniformly)
\emph{bounded} distance 
\begin{equation*}
  d^*(g_1,g_2) := \min \{d(g_1,g_2),\delta\}
\end{equation*}
for some fixed $\delta>0$,
we see that
\begin{multline*}
  \wt d^*(g_1,g_3)
  \le \wt d^*(g_1,g_2) + \wt d^*(g_2,g_3) 
  + \wt d^*(g_1,g_2) \wt d^*(g_2,g_3)\\
  \le \Bigl(1 + \sinh \frac {n\delta} 4 \Bigr) 
  \bigl( \wt d^*(g_1,g_2) + \wt d^*(g_2,g_3) \bigr),
\end{multline*}
where $\wt d^*$ is defined as $\wt d$ in~\eqref{eq:def.q-met} but with
$d^*$ instead of $d$, i.e., the triangle inequality is fulfilled up to
a factor $\tau=1+\sinh(n\delta/4)$.  Note that if $\wt d$ is bounded
by $\gamma$, then we can choose $\tau = 1+\gamma/2$.
\begin{definition}
  \label{def:q-metric}
  A function $\wt d^* \ge 0$ is called \emph{quasi-distance} if it is
  symmetric ($\wt d^*(g_1,g_2)=\wt d^*(g_2,g_1)$), definite ($\wt
  g^*(g_1,g_2)=0$ implies $g_1=g_2$), and fulfills the weak triangle
  inequality
  \begin{equation*}
    \wt d^*(g_1,g_3) \le \tau(\wt d^*(g_1,g_2)+ \wt d^*(g_1,g_2))
  \end{equation*}
  for some factor $\tau \ge 1$.
\end{definition}
Usually, a quasi-distance is called a \emph{quasi-metric}, but we
prefer the terminology ``distance'' in order not to interfere with the
word ``metric'' for the points of the space, being Riemannian metrics.
Sometimes, such distance functions are also called
\emph{semi-metrics}.  For more details on such spaces we refer
to~\cite[Sec.~14]{heinonen:01} (see also~\cite{xia:09} for a more
recent list of references).  Let us just mention Proposition~14.5
from~\cite{heinonen:01} stating that a \emph{power} of a quasi-distance
$\wt d^*$ is equivalent to a metric, i.e., for any $\eps \in
(0,\eps_0)$ there is a metric $\bar d_\eps$ and a constant
$C=C(\eps,\tau)>0$ such that
\begin{equation*}
  C^{-1}(\wt d^*(g_1,g_2))^\eps
  \le \bar d_\eps(g_1,g_2)
  \le C (\wt d^*(g_1,g_2))^\eps
\end{equation*}
for all $g_1, g_2$, where $\eps_0=\ln 2/(2 \ln \tau)$.  In particular,
a quasi-distance uniquely determines a topology and a uniform structure
\emph{independent} of $\eps$.  Hence, the notions of \emph{convergence}
and \emph{completeness} are well-defined on a space with a
quasi-distance.

Let us now pass to metrics $g_1$, $g_2$ on a manifold $M$, i.e., to
sections in the bundle $\Met(M) := \Sesq_+ T^* M$. (For the following
considerations, no smoothness assumptions on $g_1$, $g_2$ are needed.)
In particular, $g_1$, $g_2$ induce a section $A=A_{g_2,g_1}$ into the
bundle of positive definite endomorphisms on $T^*M$
applying~\eqref{eq:2.met} pointwise.  Denote by $\map{\alpha_k} M
{(0,\infty)}$ the (pointwise) eigenvalues of $A$.  Comparing the two
volume measures $\dvol_{g_1}$ and $\dvol_{g_2}$ yields
\begin{equation}
  \label{eq:2.vol}
  \dvol_{g_2} = \rho \dvol_{g_1}
  \qquad \text{with} \qquad
  \rho=\rho_{g_2, g_1}=(\det A)^{-1/2} =
  (\alpha_1 \cdot \ldots  \cdot \alpha_n)^{-1/2}.
\end{equation}
Denote by $d(g_1,g_2)(x) := d(g_1(x),g_2(x))$ the \emph{pointwise}
distance function of $g_1, g_2 \in \Met(M)$, and similarly for $\wt
d$.  It is easily seen that the pointwise estimate
\begin{equation}
  \label{eq:met.est}
  (\wt d(g_1,g_2) + 1 )^{-2}
  \le \e^{-\frac n 2 d(g_1,g_2)}
  \le \rho
  = \e^{-\frac 12 \sum_k \ln \alpha_k}
  \le \e^{\frac n 2 d(g_1,g_2)}
  \le (\wt d(g_1,g_2) + 1)^2
\end{equation}
holds, the inequality being interpreted pointwise in the sense of
quadratic forms.  The same estimate holds with $\rho$ replaced by
$\rho A$, where we use the estimate
\begin{equation}
  \label{eq:met.est2}
  \e^{-\frac n 2 d(g_1,g_2)}
  \le \e^{\min_i (-\frac 12 \sum_{k \ne i} \ln \alpha_k 
                + \frac 12 \ln \alpha_i)}
  \le \rho A
  \le \e^{\max_i (-\frac 12 \sum_{k \ne i} \ln \alpha_k 
                + \frac 12 \ln \alpha_i)}
  \le \e^{\frac n 2 d(g_1,g_2)}.
\end{equation}

For two elements $g_1, g_2 \in \Met(M)$, denote by
\begin{equation*}
  d_\infty(g_1, g_2) := \sup_{x  \in M} d(g_1,g_2)(x) 
  \qquad \text{resp.} \qquad
  \wt d_\infty(g_1, g_2) := \sup_{x  \in M} \wt d(g_1,g_2)(x) 
   \label{eq:met.est3}
\end{equation*}
the uniform (quasi-)distance.  Note that $d$ is indeed a metric on
$\Met(M)$ (except that $d_\infty (g_1,g_2)=\infty$ is possible).

 \begin{remark}
  \label{rem:q-iso}
  It is easy to see that
  \begin{equation*}
    d_\infty(g_1,g_2)<\infty 
    \quad\Longleftrightarrow \quad 
    \wt d_\infty(g_1,g_2) < \infty 
    \quad\Longleftrightarrow\quad 
    \text{$g_1$ quasi-isometric to $g_2$,}
  \end{equation*} 
  where the last statement means that there exists $\eta>0$ such that
  \begin{equation*}
    \eta^{-1} g_2(x) \le g_1(x) \le \eta g_1(x)
  \end{equation*}
  for all $x \in M$ in the sense of quadratic forms (one may choose
  $\eta=\exp d_\infty(g_1,g_2)$). Recall that the notion
  ``quasi-isometric'' was already defined in \Def{quasi} in
  \Sec{ex.wo}.
\end{remark}

For a further analysis of the topological (or uniform) structure of
$\Met(M)$, we refer to~\cite{eichhorn:07}.

%
\section{Pointwise bounds for $u(x)$ and $\de u(x)$}
\label{app:ptw.bd}
%

Let $\Mscr = (M,g)$ denote a {\colour[B12] complete} Riemannian
manifold satisfying the assumptions of \Sec{mfds.sob.lap}.  The metric
$g$ is assumed to be smooth (or to have the ``minimal'' regularity of
\Rem{regularity}).  Let $H$ denote the {\colour[B13] self-adjoint and
  non-negative extension of the Laplacian of $\Mscr$.}
We provide pointwise bounds for $u$ and $\de u$ where $u = (H +
1)^{-m}f$, $f \in \Lsqr \Mscr$, and $m$ is sufficiently large.  For
{\colour $n < p < \infty$} and $1 < Q \le 2$ let $r(x) = \rH
\Mscr(x,p,Q)$ denote the harmonic radius at $x \in M$ as in
\Sec{harm.rad}.  We then have the following theorem.
\begin{theorem}
  \label{thm:b.1} 
  Let $H$ and $r(x)$ be as above, let $f \in \Lsqr \Mscr$, and let $u
  := (H+1)^{-m} f$, where $m \in \N$, $m \ge [n/4] + 2$.  Then
  {\colour $u \in \Cont[1]\Mscr$ and} there exist constants $C>0$,
  depending only on $m$, $n$, $p$, and $Q$, such that
  \begin{subequations}
    \begin{align}
      \label{eq:b.1a}
      \abs{u(x)} &
      \le C (\min\{1, r(x)\})^{-n/2} \norm[\Lsqr \Mscr] f
    \intertext{and}
      \label{eq:b.1b}
    \abs[g] {\de u(x)} &
    \le C (\min\{1,r(x)\})^{-n/2 -1} \norm [\Lsqr \Mscr] f.
    \end{align}
  \end{subequations}
\end{theorem}
While the first estimate is well-known (cf., e.g., \cite{cgt:82}), the
gradient estimate requires some additional work. The most crucial
ingredient is an estimate in elliptic regularity theory (eqn.~(0.10)
in~\cite{anderson-cheeger:92}) which we adapt to our situation. We
will employ elliptic regularity theory in $\Lsqrspace$ for equations in
divergence form as well as elliptic regularity theory in $\Lpspace$
for strong solutions.  The first one will allow us to show that weak
solutions are in fact strong solutions in the sense of
\cite{gilbarg-trudinger:83}; the actual estimates will then be
obtained from Theorem~9.11 in~\cite{gilbarg-trudinger:83}.

Let us recall how the equation $(H + 1) u = f$ reads in local
coordinates. Note first that
\begin{equation}
\label{eq:b.111}
  \abssqr[g] {\de u}
  = \sum_{i,j=1}^n g^{ij} \partial_i u \partial_j \conj u . 
\end{equation} 
 Therefore, by the very definition of the operator $H$ in
\Sec{mfds.sob.lap}, the equation $Hu + u = f$ (in the weak sense) means
 that $u$ belongs to $\Sob \Mscr$  and satisfies
\begin{equation}
  \label{eq:b.2}
  \sum_{i,j=1}^n
       \int g^{ij} \partial_i u \,\partial_j \phi  \sqrt{g}  \dd x 
     + \int u \phi \sqrt{g} \dd x 
  = \int f \phi \sqrt{g} \dd x,  
\end{equation}
for all $\phi \in \Cci U$, where $U$ is as in \Sec{harm.rad} and where
$\sqrt g = \sqrt{\det(g_{ij})}$.  (Note that we do not
distinguish in the notation between $u$ and $u \circ \Phi^{-1}$ if
$\Phi$ is a coordinate map.)  Therefore, the \emph{weak form} of the
partial differential equation is formally given by
\begin{equation}
  \label{eq:b.3}
  - \sum_{i,j=1}^n \partial_j (\sqrt g g^{ij} \partial_i u) + \sqrt g u
  = \sqrt g f.    
\end{equation}
We will see shortly that, under suitable assumptions, $u$ belongs to
$\Sobx [2] \loc U$ and is also a \emph{strong} solution. In the
special case of harmonic coordinates the first-order terms cancel, and
we see that $u$ then satisfies the partial differential equation  
\begin{equation}
\label{eq:b.4}
   - \sum_{i,j=1}^n g^{ij} \partial_i \partial_j u +  u = f; 
 \end{equation}
 {\colour[B14] see, e.g., \cite{deturck-kazdan:81}.}
Since, in harmonic coordinates, $g$ is close to $1$ in its coordinate
patch (cf.\ eqn.~\eqref{eq:2.3}), the extra factors of $\sqrt g$ in
eqn.~\eqref{eq:b.3} pose no problem.

As a preparation, we recall some facts from (interior) elliptic
regularity theory in a bounded domain $\Omega \subset \R^n$ where we
now use the symbols $u$ and $f$ in a different context.  We begin with
\emph{weak} solutions $u \in \Sobx [1] \loc \Omega$ of an elliptic
equation $L_w u = f$,
\begin{equation}
  \label{eq:b.5}
  L_w u := -\sum_{i,j=1}^n \partial_j A_{ij}\partial_i  u + \gamma u, 
\end{equation}
where the coefficient matrix $(A_{ij})$ is uniformly positive definite
with $(A_{ij}) \ge 1/2$, the $A_{ij}$ are of class $\Cont[0,1]
\Omega$, and $\gamma$ is bounded; finally, we assume $f \in \Lp[q]
\Omega$ for some $q \ge 2$.
From~\cite[Thm.~8.8]{gilbarg-trudinger:83} we then infer that $u \in
\Sobx [2] \loc \Omega$ and that $u$ is a \emph{strong solution}
of eqn.~\eqref{eq:b.5} in the sense of~\cite{gilbarg-trudinger:83}.
Furthermore,~\cite[Lemma 9.16]{gilbarg-trudinger:83} yields that $ u
\in \SobW[q,\loc] 2 \Omega$.  This type of regularity will be
needed later on.

We next consider \emph{strong} solutions $u \in
\SobW[q,\loc] 2 \Omega$ of an elliptic equation $L u = f$ with
\begin{equation}
  \label{eq:b.6}
  L u := -\sum_{i,j=1}^n a_{ij}\partial_i \partial_j u + \gamma u
\end{equation}
where, as above, $q \ge 2$, $(a_{ij}) \ge 1/2$, $a_{ij} \in \Cont[0,1]
{\Omega}$, and $\gamma$ is bounded.
In view of the Sobolev Embedding Theorem, we define the exponent
\begin{equation}
  \label{eq:b.7}
  \sigma(q) :=
  \begin{cases}
    \dfrac{qn}{n - 2q}, & 2q < n,\\
    q+1, & 2q = n,\\
    \infty, & 2q > n,  
  \end{cases}
\end{equation}
for $q \in [1,\infty]$.  For $a \in \Cont[0,\alpha] {\clo \Omega}$
(the space of \emph{uniformly} H\"older-continuous functions), we
denote the $\alpha$-H\"older-constant of $a$ by $[a]_{0,\alpha}$.  We
then have the following lemma.

\begin{lemma}
  \label{lem:b.2} 
  Let $\Omega := B_2 \subset \R^n$ and let $L$ as in~\eqref{eq:b.6}
  with $(a_{ij}) \ge 1/2$, $a_{ij} \in \Cont[0,1] {\Omega}$, and
  $\gamma$ bounded. {\colour Let $\alpha \in (0,1]$ and} let $\Lambda
  > 0$ be such that $\norm[\infty]{a_{ij}} \le \Lambda$,
  $[a_{ij}]_{0,\alpha} \le \Lambda$ and $\norm[\infty]{\gamma} \le
  \Lambda$.  Let $q \in [2, \infty)$ and let $u \in \SobW[q] 2 {B_2}$.
  We then have:
  \begin{enumerate}
  \item
    \label{b.2.i}
    There exists a constant $C_1 \ge 0$, depending only on $n$,
    $\alpha$, and $\Lambda$, such that
    \begin{equation}
      \label{eq:b.8}
      \norm [{\SobW[q] 2 {B_1}}] u
      \le C_1 \bigl(\norm[{\Lp[q]{B_2}}] {L u}
        + \norm[\Lsqr {B_2}]u \bigr).
    \end{equation}
  \item
    \label{b.2.ii}
    Let $q_1 := \sigma(q)$ as in~\eqref{eq:b.7}.  Then there
    exists a constant $C_2 \ge 0$, depending only on $n$, $\alpha$,
    and $\Lambda$, such that
    \begin{equation}
      \label{eq:b.9}
      \norm[{\Lp[q_1]{B_1}}] u
      \le C_2 \bigl( \norm[{\Lp[q]{B_2}}]{L u}
                   + \norm[\Lsqr{B_2}] u \bigr).
    \end{equation}
  \end{enumerate}
\end{lemma}

\begin{remark}
  \begin{myenumerate}{\alph}
  \item It is essential for later applications that the last term in
    eqns.~\eqref{eq:b.8} and~\eqref{eq:b.9} is an $\Lsqrspace$-norm,
    as in eqn.~(0.10) in~\cite{anderson-cheeger:92}.

  \item If $L$ would also contain first order terms $b_i
    \partial_i u$, it appears that we would need to require the
    coefficients $b_i$ to be bounded. In general, the first order
    terms of the Laplacian contain derivatives of the $g^{ij}$ and we
    would need an assumption like $\norm[{\SobWspace[\infty] 1}] {g^{ij}}
    \le \Lambda$, while the harmonic coordinates only come with an
    estimate for $\norm[{\SobWspace[p] 1}]{g^{ij}}$.
  \end{myenumerate}
\end{remark}

\begin{proof} The proof of \Lem{b.2} combines elliptic regularity in
  $\Lpspace$ with a simple bootstrap argument for which we fix a
  sequence of radii $2 > \rho_1 > \rho_2 > \dots > 1$.

  \begin{myenumerate}{\roman}
  \item Interior elliptic regularity in $\Lsqr {B_2}$ as
    in~\cite[Thm.~9.11]{gilbarg-trudinger:83} gives us a constant
    $c_1$, depending only on $\Lambda$ and $\rho_1$ such that
    \begin{equation*}
      \norm[{\Sob [2] {B_{\rho_1}} }] u
      \le c_1 \bigl(   \norm[\Lsqr{B_2}] {Lu}
      + \norm[\Lsqr{B_2}]  u
      \bigr).
    \end{equation*}

  \item Let $p_1 := \min\{q, \sigma(2)\}$.  By the Sobolev Embedding
    Theorem, there is a constant $c_2$, depending only on $n$,
    $\alpha$, $\Lambda$, and $q$, such that
    \begin{equation*}
      \norm[{\Lp[p_1]{B_{\rho_2}}}] u
      \le c_2 \norm[{ \Sob [2] {B_{\rho_1}} }] u
      \le c_1 c_2 \bigl(
      \norm[\Lsqr {B_2}]{Lu}
      + \norm[\Lsqr {B_2}] u
      \bigr).
    \end{equation*}

  \item We apply~\cite[Thm.~9.11]{gilbarg-trudinger:83} in
    $\Lp[p_1]{B_{\rho_2}}$ to obtain a constant $c_3$, depending only
    on $\Lambda$, $\rho_1$, $\rho_2$, $\rho_3$, and $p_1$ such that
    \begin{align*}
      \norm[{ \SobW[p_1] 2 {B_{\rho_3}} }] u & \le c_3 \bigl(
      \norm[{\Lp[p_1]{B_{\rho_2}}}] {Lu} +
      \norm[{\Lp[p_1]{B_{\rho_2}}}] u
      \bigr)\\
      & \le c_4 \norm[{\Lp[p_1]{B_2}}] {Lu} + c_5 \norm[\Lsqr{B_2}] u.
    \end{align*}
  \end{myenumerate}
  If $p_1 = q$ the proof of the first inequality is finished.
  Otherwise we continue with another application of the Sobolev
  Embedding Theorem. The proof terminates after a finite number of
  steps (which depends only on $n$ and $q$). The second inequality
  follows by Sobolev. 
\end{proof}

By a simple scaling argument we now transfer the
estimate~\eqref{eq:b.9} from $B_2$ to $B_{2r}$ where $0 < r \le 1$:

\begin{lemma}
  \label{lem:b.4}
  Let $0 < r \le 1$ and consider $\Omega := B_{2r}$. {\colour Let $L$
    be as in~\eqref{eq:b.6} with $(a_{ij}) \ge 1/2$, $a_{ij} \in
    \Cont[0,1] {\Omega}$, and $\gamma$ bounded.  Let $\Lambda > 0$ be
    such that $\norm[\infty]{a_{ij}} \le \Lambda$,
    $[a_{ij}]_{0,\alpha} \le \Lambda r^{-\alpha}$, and
    $\norm[\infty]{\gamma} \le \Lambda$ on $B_{2r}$.}  Let $q \in
  [2,\infty)$, $q_1 := \sigma(q)$, and let $u \in \SobW[q] 2
  {B_{2r}}$.

  Then there exists a constant $C \ge 0$, depending only on $n$,
  $\alpha$, and $\Lambda$, such that
  \begin{equation}
    \label{eq:b.10}
    \norm [{\Lp[q_1]{B_r}}] u
    \le C r^2 \cdot r^{-n(1/q - 1/q_1)} 
            \cdot \norm[{\Lp[q]{B_{2r}}}] {Lu}
          + C \cdot r^{-n(1/2 - 1/q_1)} \cdot \norm[{\Lsqr{B_{2r}}}] u.
 \end{equation}
\end{lemma}
\begin{proof}
  Write $f := Lu \in \Lp[q]{B_{2r}}$ and scale out (i.e., set $y :=
  x/r$) to obtain
  \begin{equation*}
    {\tilde u}(y) := u(ry), \qquad y \in B_2.
  \end{equation*}
  Similarly, write ${\tilde a}_{ij}(y) := a_{ij}(ry)$, ${\tilde \gamma}(y)
  := \gamma(ry)$, and ${\tilde f}(y) := f(ry)$.  Defining
  \begin{equation*}
    {\tilde L} = - \sum {\tilde a}_{ij} \partial_{y_i} \partial_{y_j} +
    r^2 {\tilde \gamma},
  \end{equation*}
  the equation $Lu = f$ in $\Lsqr{B_{2r}}$ is then equivalent with
  \begin{equation}
    \label{eq:b.11}
    \tilde L  \tilde u = r^2 {\tilde f}(y)
  \end{equation} 
  in $\Lsqr{B_2}$.  Applying \Lem{b.2} to eqn.~\eqref{eq:b.11} in
  $B_2$, we find that
  \begin{equation*}
    \norm[{\Lp[q_1]{B_1}}]{\tilde u}
    \le C r^2 \norm [{\Lp[q]{B_2}}]{\tilde f}
       + C \norm [\Lsqr {B_2}] {\tilde u},
  \end{equation*} 
  with a constant $C$ depending only on $n$, $\alpha$, and $\Lambda$;
  note that $[{\tilde a}_{ij}]_{0,\alpha} {\colour \le
    r^{\alpha}[a_{ij}]_{0,\alpha} \le \Lambda}$.  Scaling back yields
  \begin{equation*}
    r^{-n/q_1} \norm [{\Lp[q_1]{B_r}}] u
    \le C r^2 r^{-n/q} \cdot
        \norm [{\Lp[q]{B_{2r}}}] f  + r^{-n/2} \norm [\Lsqr {B_r}] u
  \end{equation*}
  and the result follows.
\end{proof}

The estimate~\eqref{eq:b.10} is slightly more precise than {\colour what
  follows from} the estimate~(0.10) in~\cite{anderson-cheeger:92} in
the sense that the dependence on the (local) harmonic radius is made
explicit in eqn.~\eqref{eq:b.10}.

The proof of \Thm{b.1} will be based on an iteration of \Lem{b.4}.  To
illustrate the idea, let $f \in \Lsqr {B_{2r}}$ and let $u_1$, $u_2$
satisfy $L u_1 = f$ and $L u_2 = u_1$ in $B_{2r}$, in the sense of a
strong solution; in particular, we assume $u_1, u_2 \in \Sobx [2] \loc
{B_{2r}}$.  By Sobolev, we then have $u_1 \in \Lploc[q_1] {B_{2r}}$,
where $q_1 := \sigma(2)$ and, by Lemma~9.15
in~\cite{gilbarg-trudinger:83}, $u_2 \in \SobW[q_1,\loc] 2 {B_{2r}}$.
Again, we assume $0 < r \le 1$. Then, \Lem{b.4} again yields
\begin{equation}
\label{eq:b.12}
  \norm [{\Lp[q_1]{B_r}}] {u_1}
   \le C  \cdot r^{2-n (1/2 - 1/q_1)}
          \cdot \norm [\Lsqr {B_{2r}}] f
     + C \cdot  r^{-n (1/2 - 1/q_1)}
          \cdot \norm [\Lsqr{B_{2r}}]{u_1}. 
\end{equation}
Similarly, we see that $u_2 \in \Lp[q_2]{B_{r/2}}$ with $q_2 :=
\sigma(q_1)$.  Again, \Lem{b.4} gives 
\begin{equation}
  \label{eq:b.13}
  \norm [{\Lp[q_2] {B_{r/2}}}]{u_2} 
  \le C \cdot r^{2-n(1/q_1 - 1/q_2)}
        \cdot \norm [{\Lp[q_1]{B_r}}]{u_1}  +
      C \cdot r^{-n (1/2 -  1/q_2)}\cdot
          \norm[\Lsqr{B_r}] {u_2}.
 \end{equation}
 Inserting~\eqref{eq:b.12} into~\eqref{eq:b.13} we obtain
 \begin{align}
  \nonumber
  \norm [{\Lp[q_2]{B_{r/2}}}] {u_2}
  & \le C r^{-n (1/2 - 1/q_2)}
        \cdot
        \bigl(
             C r^4 \norm[\Lsqr {B_{2r}}] f
             + C   \norm[\Lsqr {B_{2r}}] {u_1}
             +     \norm[\Lsqr {B_r}] {u_2}
        \bigr)\\
  \label{eq:b.14}
   & \le C' r^{-n/2}
        \bigl(
               \norm[\Lsqr {B_{2r}}] f
             + \norm[\Lsqr {B_{2r}}] {u_1}
             + \norm[\Lsqr {B_r}] {u_2}
        \bigr); 
\end{align}
note that the powers $r^{1/q_1}$ have dropped out. Clearly,
analogous\hiddenfootnote{Begr\"undung f\"ur ``analogous'':
  \begin{align*}
   \norm[{\Lp[q_2] {B_{r/2}}}]{u_2}
      & \le C^2\cdot r^4 \cdot r^{-n(1/2 -  1/q_2)} \cdot
                                      \norm[\Lsqr{B_{2r}}]{f}  \\
      & + C^2 \cdot  r^{-n(1/2 -  1/q_2)} \cdot
                                       \norm[\Lsqr {B_{2r}}]{u_1}\\
      & + C \cdot  r^{-n(1/q_1 -  1/q_2)}
                        \norm[\Lsqr{B_r}]{u_2}.
\end{align*}
}
estimates hold for (finite) chains of equations where $L u_{k+1} =
u_k$, for $k = 0, \ldots, m$.

Finally, we now return to manifolds $\Mscr$ with Laplacian $H$.  In
view of \Def{harm.rad} we fix \emph{once and for all} some $p \in
(n,\infty)$ and let $\alpha := 1 - n/p$.  We also fix some $1 < Q \le 2$
close enough to $1$ to ensure that any of the neighborhoods $U$ of
\Sec{harm.rad} contains a Euclidean ball of radius $\rH \Mscr (x)/2$;
here we use eqn.~\eqref{eq:2.3} and the usual formula for the length
of curves in local coordinates. In the sequel, we will suppress the
dependence on the constants $p$, $Q$ and $n$ in the notation.

 Let $f \in \Lsqr \Mscr$ and consider
\begin{equation*}
  u_k := (H+1)^{-k}f, \qquad k \in \N_0;   
\end{equation*}
we then have $u_k \in \dom H$,
\begin{equation}
  \label{eq:b.15}
  (H+1) u_{k+1} = u_k, \qquad k \in \N_0, 
\end{equation}
and 
\begin{equation*}
  \norm [\Lsqr \Mscr] {u_k}
  \le \norm [\Lsqr \Mscr] f, \qquad k \in \N. 
\end{equation*}
Let $r(x) = \rH\Mscr(x,p,Q)$ denote the harmonic radius at $x \in M$.
Passing to harmonic coordinates in a geodesic ball $B = B_\Mscr(x,r(x)) 
\subset M$ of radius $r(x)$ around $x$, the functions $u_k$ (or, more
precisely, $u_k \circ \Phi^{-1}$ etc.)  are weak solutions of the
divergence form equation $-\partial_j \sqrt g g^{ij} \partial_i
u_{k+1} + \sqrt g u_{k+1} = \sqrt g u_k$, but then, as explained
above, they are also strong solutions of
\begin{equation}
  \label{eq:b.16}
  - g^{ij} \partial_i \partial_j u_{k+1} + u_{k+1} = u_k 
\end{equation}
in $\Phi(B) \subset \R^n$, for $k \in \N_0$; note that we may apply
Theorem~8.8 of~\cite{gilbarg-trudinger:83} to the weak equation since
$\sqrt g g^{ij}$ is in $\SobWspace [\infty,\loc] 1$ and thus locally
Lipschitz.

\begin{proof}[Proof of \Thm{b.1}]
  The $g^{ij}$ satisfy the
  estimates~\eqref{eq:2.3}--{\colour \eqref{eq:gij.hoelder}} and we
  see that $\Lambda$ (defined as in \Lem{b.4}) depends only on $n$ and
  $p$.  Applying \Lem{b.4} successively to the
  equations~~\eqref{eq:b.16}, as indicated above, we obtain the
  estimate~\eqref{eq:b.1a}, which, in fact, holds for $m \ge [n/4] +
  1$.

  As for the gradient estimate, we let $k_0 := [n/4] + 1$ and consider
  the equation $L u_{k_0+1} + u_{k_0+1} = u_{k_0}$ where, by the
  above, $u_{k_0}$ and $u_{k_0+1}$ are locally bounded with
  estimates
  \begin{equation*}
    |u_{k_0}(x')|, \> |u_{k_0+1}(x')| 
           \le C r(x)^{-n/2} \norm[\Lsqr \Mscr] f,
     \qquad |x-x'| \le 2^{-k_0-1} r(x).
  \end{equation*}
  Scaling out as in the proof of \Lem{b.4}, but now with a factor of
  $2^{k_0+3}r(x)$, we find that the scaled function ${\tilde
    u}_{k_0+1}$ satisfies an equation
  \begin{equation*}
    {\tilde L} {\tilde u}_{k_0+1} = 4^{-k_0 - 3} r(x)^2 {\tilde u}_{k_0}
  \end{equation*}
  in $B_2 \subset \R^n$, where ${\tilde L} = - {\tilde g}^{ij}
  \partial_i\partial_j +  4^{-k_0 - 3} r(x)^2$. 

  Fix some $q \in (n,\infty)$, e.g., $q := n+1$.  As above, we have
  ${\tilde u}_{k_0+1} \in \SobW[q,\loc] 2 {B_2}$ and \Lem{b.2}
  yields an estimate
  \begin{equation*}
     \norm [{\SobW[q] 2 {B_1}}] {{\tilde u}_{k_0+1}}
        \le c
           \bigl(
               \norm [{\Lp[q]{B_2}}] {{\tilde u}_{k_0}} 
                +\norm[\Lsqr{B_2}] {{\tilde u}_{k_0+1}}
           \bigr)
        \le c' r^{-n/2} \norm [\Lsqr \Mscr] f.
  \end{equation*}
  where the constants $c$, $c'$ depend only on $n$, $p$, and $Q$.

  By the Sobolev Embedding Theorem, we now conclude that {\colour
    ${\tilde u}_{k_0+1} \in \Cont[1]{B_1}$ and}
  \begin{equation*}
    |\nabla {\tilde u}_{k_0+1}(x)|
    \le C' \norm [{\SobW[q] 2 {B_1}}] {{\tilde u}_{k_0+1}}
    \le c'' r^{-n/2} \norm [\Lsqr \Mscr] f, \qquad {\colour x \in B_1,} 
  \end{equation*}
  with $C'$ depending only on $n$, and $c'' := C'c'$.

  Scaling back gives the estimate $|\nabla u_{k_0+1}(x)| \le c
  r^{-n/2-1} \norm [\Lsqr \Mscr] f$, with a constant $c$ depending
  only on $n$, $p$, and $Q$.  We conclude by
  combining~\eqref{eq:b.111} with the estimate~\eqref{eq:2.3}.
\end{proof}

{\colour[B11] Since it fits well into the context of this appendix, we
  indicate here how to deal with (smooth) boundaries $\Sigma$ and some
  of the mapping properties of $R^m = (H+1)^{-m}$ and of $R_\dec^m =
  (H_\dec + 1)^{-m}$, as required in the proof of Proposition 4.5.
  Recall from Section 4 that $\Sigma_s = \bd M_s$, for $s = 1, \ldots,
  \ell$, and $\Sigma = \bigcup_{s=1}^\ell \Sigma_s = \bd M_0$ where we
  now label by the index $s$ instead of $k$. The spaces
  $\Cont[1]{\Mscr}$ and $\Cont[1]{\clo \Mscr_s}$, for $s = 0, \ldots,
  \ell$, are as in \Sec{mfds.ends}.
 
\begin{proposition} 
  \label{prp:b.5}
  For $m \in \N$, $m \ge [n/4] + 2$ we have the following:  
  \begin{myenumerate}{\alph}
  \item
    \label{prp.b.5a}
    $R^m$ is a bounded operator from $\Lsqr \Mscr$ to
    $\Cont[1]{\Mscr}$;

  \item
    \label{prp.b.5b}
    $R_\dec^m$ is a bounded operator from $\Lsqr \Mscr$ to
    $\Cont[1]{\clo \Mscr_0} \times \cdots \times
    \Cont[1]{\clo \Mscr_\ell}$.  Furthermore, for all $f \in
    \Lsqr \Mscr$ we have $(R_\dec^m f) \restriction \Sigma = 0$.

  \item
    \label{prp.b.5c}
    For any $K \subset M$ compact there exists a constant $C_K$ such
    that
    \begin{align}
      \label{eq:b.18}
      \sup_{x\in K} (\abs{u(x)} + \abs[g] {\de u(x)}) \le C
      \norm[\Lsqr \Mscr] f,
    \end{align}
    for all $f \in \Lsqr \Mscr$ and $u = R^m f$ or $u = R_\dec^m f$.
  \end{myenumerate}
\end{proposition} 

One might say that this result is a routine consequence of elliptic
regularity theory, and, indeed, its proof is very similar to the proof
of \Thm{b.1}. Let again $\map {r_0} M {(0,1]}$ denote the (continuous)
function introduced in \Prp{lower.bd.rh}. Notice that, in \Prp{b.5},
we do not need to control the constants in our estimates as functions
of $r_0(x)$, which simplifies the argument as compared to the proof of
\Thm{b.1}; on the other hand, the presence of a boundary requires the
use of appropriate tools from elliptic regularity theory.
%
\begin{proof} 
  For $K \subset M$ compact, there exists a constant $\rho_0 > 0$ such
  that $r_0(x) \ge \rho_0$ for all $x \in {K}$.  The desired results
  in~\itemref{prp.b.5a} and~\itemref{prp.b.5c} for $R^m$ and $u =
  (H+1)^{-m}f$ are now immediate from \Thm{b.1} (but note that we
  could also use the simpler arguments given below).

  We next consider $R_\dec^m$.  For any $x \in {K}$ there exists an
  open neighborhood ${U}_x$ which admits a system of harmonic
  coordinates $\Phi_x$ that map ${U}_x$ diffeomorphically to a
  Euclidean ball $B(0,r_x)$ of radius $r_x > 0$.  By compactness,
  there exists a finite selection of points $x_1, \ldots, x_J \in {K}$
  such that the union of ${U}_{x_1}, \ldots, {U}_{x_J}$ covers
  ${K}$. We may assume, in addition, that for any $j$ we either have
  ${U}_{x_j} \cap \Sigma = \emptyset$ or $x_j \in \Sigma$.  We write
  $N_{s,j} := \Phi_{x_j}(M_s \cap {U}_{x_j}) \subset B(0,r_j)$, for $s
  = 0, \ldots, \ell$ and $j = 1, \ldots, J$.

  Letting $r_j := r_{x_j}$ and $\map{\Psi_j}{B(0,r_j)} M$ denote
  the inverse of $\Phi_{x_j}$, there exist radii $0 < r_j' < r_j$ with
  the property that the sets $\Psi_j(B(0,r_j'))$ also cover ${K}$. We
  therefore only have to produce the required bounds on Euclidean
  balls $B(0,r_j')$.

  As in eqn.~\eqref{eq:b.15}, we write $u_0 : = f$ and $u_k := (H_\dec
  + 1)^{-k} f$ for $k \in \N$; we also let $u_{k,j} := u_k \circ
  \Psi_j$.  We then have $\norm{u_k} \le \norm{f}$ and $\normsqr
  [\Lsqr {T^*\Mscr}] {\de u_k} \le C$ for all $k$.  Furthermore, the
  equation $(H_\dec + 1)u_{k+1} = u_k$ in $\Lsqr{\Mscr_s}$ implies
  that $u_{k+1,j}$ is a weak solution of the associated divergence
  form elliptic partial differential equation in local coordinates
  (cf.~eqn.~\eqref{eq:b.3}) in the sets $N_{s,j}$, satisfying
  Dirichlet boundary conditions on $B(0,r_j) \cap \Phi_{x_j} (\partial
  M_s \cap {U}_{x_j})$, for $s = 0, \ldots, \ell$.

  For simplicity of notation, let us assume that $r_j = 1$ and $r_j' =
  1/2$.  It will be convenient to introduce the radii $\rho_\nu :=
  \frac 1 2 + 2^{-\nu-1}$, for $\nu \in \N_0$, so that $1/2 <
  \rho_{\nu+1} < \rho_\nu < 1$ for all $\nu \in \N$.  We may assume
  without loss of generality that the domains $N_{s,j,\nu} := N_{s,j}
  \cap B(0,\rho_\nu)$ are Lipschitz so that the Sobolev Embedding
  Theorem in the form of~\cite[Eqn.~(7.30)]{gilbarg-trudinger:83} can
  be applied to each of the $N_{s,j,\nu}$ (with a constant which may
  depend on $\nu$).

  \begin{myenumerate}{\roman}
  \item 
    Now, applying~\cite[Thm.~8.12]{gilbarg-trudinger:83} in a suitable
    domain with $\Contsymb^2$-boundary yields $u_{k,j} \in \Sob[2] 
    {N_{s,j,1}}$ for $k \in \N$; furthermore, $u_{k,j}$ satisfies the
    estimate~(8.25) of~\cite{gilbarg-trudinger:83} with $\Omega :=
    N_{s,j,1}$ and $\phi = 0$. (If $x_i \in \Sigma$ the
    application of~\cite[Thm.~8.12]{gilbarg-trudinger:83} requires
    some care: we first pick a cut-off function $\psi \in
    \Cci{B(0,1)}$ which is $1$ on $B(0, \rho_1)$ and
    plug $\psi u_{k,j}$ into the p.d.e.\ satisfied by $u_{k,j}$.) 
    In addition, we may conclude that $u_{k,j}$ is a \emph{strong} 
    solution of the associated partial differential equation
    in $N_{s,j,1}$; cf.\  eqn.~\eqref{eq:b.16}. 
  \item 
    Since, by the first step, $u_{1,j} \in \Sob[2]{N_{s,j,1} }$, the
    Sobolev Embedding Theorem yields $u_{1,j} \in \Lp[q_1]{N_{s,j,1}}$
    with $q_1 := \sigma(2)$, where $\sigma(\cdot)$ is defined
    in~\eqref{eq:b.7}. Now~\cite[Theorem~9.13 and
    Lemma~9.16]{gilbarg-trudinger:83} imply that $u_{2,j} \in
    \SobW[q_1] 2 {N_{s,j,2}}$ together with an estimate as in
    \Lem{b.2}
     \begin{align*}
       \norm [{\SobW[q_1] 2 {N_{s,j,2}}}] {u_{2,j}}  & \le C
        \bigl( \norm [\Lsqr {N_{s,j,1}}] {(H_\dec + 1) u_{2,j}}
       + \norm [\Lsqr {N_{s,j,1}}] {u_{2,j}} \bigr) \\
       & \le C \bigl( \norm [\Lsqr  {N_{s,j,1}}] {u_{1,j}} 
                 + \norm [\Lsqr  {N_{s,j,1}}]
        {u_{2,j}} \bigr) \le C \norm [\Lsqr \Mscr] f.
    \end{align*}
    In a similar fashion we subsequently obtain an estimate of $ \norm
    [{\SobW[q_2] 2 {N_{s,j,3}}}] {u_{3,j}}$ in terms of
    $\norm [\Lsqr \Mscr] f$ etc.
  \item 
    Iterating the above steps $m$ times we arrive at $q_m > n/2$ with
    an estimate
    \begin{align*}
      \norm [{\SobW[q_m] 2 {N_{s,j,m+1}}}] {u_{m+1,j}}   \le C  \norm [\Lsqr
      \Mscr] f.  
    \end{align*}
    A variant of the Sobolev Embedding Theorem
    (\cite[Thm.~7.26]{gilbarg-trudinger:83}, which uses the Sobolev
    Extension Theorem) yields $u_{m+1,j} \in
    \Cont[1]{\overline{N}_{s,j,m+2}}$ and a corresponding estimate
    estimate of $|u_{m+1,j}(x)|$ and $|\nabla u_{m+1,j}(x)|$ for $x
    \in N_{s,j,m+2}$. The estimate of $|u(x)|$ is now immediate, while
    the estimate of $\abs[g] {\de u(x)}$ follows as in the proof of
    \Thm{b.1}.\qedhere 
  \end{myenumerate}
 \end{proof} 
}
%
\section{Scattering on the line}
\label{app:line}
%

Let us denote by $h_0$ the unique self-adjoint realization of
$-\dd^2/\dd x^2$ in $\HS:=\Lsqr \R$ with domain $\Sob [2] \R$.
Denoting by $\mathcal F$ the Fourier transform in $\Lsqr{\mathbb R}$,
we have
\begin{equation*}
  (\e^{-\im t h_0} \phi)(x)
  = \frac 1 {\sqrt{2\pi}} \int_\R  
         \e^{\im (kx-tk^2)} (\mathcal F\phi)(k) \dd k,
         \quad \phi \in \Schwartz \R.
\end{equation*}
We state below the standard stationary-phase estimate
(cf.~\cite[Corollary to Thm.~XI.14]{reed-simon-3}): 
Let $\phi \in \Lsqr \R$ satisfy $\mathcal F \phi \in \Cci \R$ with
$\supp \mathcal F \phi = K$.  Then, for any open set
$U$ such that $ K \subset U$ and for any $ m \in \N$, there
is a constant $C_m$ such that
\begin{equation}
  \label{eq:c.2}
  \bigabs{\e^{-\im t h_0} \phi(x)}
  \le C_m  \bigl(1+|x|+|t|\bigr)^{-m},
    \quad \text{for all $x \in \R$ such that $x/(2t) \notin U$.}
\end{equation} 
Similar estimates hold for the $x$-derivatives of $\e^{-\im t h_0}
\phi$.  Let $h_{\pm,0}$ be the self-adjoint re\-a\-li\-zations of
$-\dd^2/\dd x^2$ in $\Lsqr{\R_\pm}$ with Dirichlet boundary condition
at zero, i.e., with domain $\dom h_{\pm,0}= \Sobn{\R_\pm} \cap \Sob
[2] {\R_\pm}$.  We denote by $\iota_\pm$ the natural embeddings of
$\Lsqr{\R_\pm}$ into $\Lsqr \R$ (extension by zero) and by
\begin{equation*}
  \mathcal H_\pm
  := \mathcal F^{-1} \bigl(\iota_\pm \Lsqr{\R_\pm}\bigr)
  =\bigset{u \in \HS} { \hat u \restr {\R_\mp}= 0}.
\end{equation*}

We define the wave operators
\begin{equation*}
  \Theta^+_{\pm} := \slimpm \e^{\im th_0} \iota_+ \e^{-\im t h_{+,0}}
  \quad\text{and}\quad
  \Theta^-_\pm := \slimpm \e^{\im th_0} \iota_- \e^{-\im t h_{-,0}},
\end{equation*}
provided that the strong limits exist.

\begin{lemma}
  \label{lem:c.1}
  The wave operators $\Theta^+_{\pm}$ and $\Theta^-_{\pm}$ exist, are
  isometric and we have
  \begin{equation*}
    \ran \Theta^+_{\pm}= \mathcal H_{\pm}
    \quad\text{and}\quad
    \ran \Theta^-_{\pm}= \mathcal H_{\mp}.
  \end{equation*}
\end{lemma}
\begin{proof}
  We only give the proof for $\Theta^+_\pm$; the case of
  $\Theta^-_\pm$ is similar. We first prove the existence of the
  $\Theta^+_\pm$.  Let $j \in \Cont[2] \R$ satisfy $0 \le j(x) \leq
  1$, $j(x)=0$ for $x \leq 0$ and $j(x)=1$ for $x \ge 2$.  We also
  denote by $j$ the bounded operator from $\Lsqr{\R_+}$ into $\mathcal
  H $ given by multiplication by $j$. {\colour[B15] As a first step,
    we replace $\iota_+$ in the definition of $\Theta^+_\pm$ by $j$,
    using a well-known and simple argument; cf.~\cite[p.~35 and
    problem 18]{reed-simon-3}: The function $1 -j$, defined on $\R_+$,
    is bounded and has compact support; hence the Rellich local
    compactness theorem implies that $(1-j) (h_{+,0} + 1)^{-1}$ is
    compact.  Since $h_{+,0}$ is absolutely continuous, we have
    $\slimpm (1- j) (h_{+,0} + 1)^{-1} \e^{-\im t h_{+,0}}=0$ and
    therefore
\begin{equation*} 
    (1- j)  \e^{-\im t h_{+,0}}  \varphi = 
     (1- j) (h_{+,0} + 1)^{-1} \e^{-\im t h_{+,0}}
   (h_{+,0} + 1) \varphi \to 0 
   \end{equation*}
in norm, for all $\varphi \in \Cci {\R_+}$.  Therefore, it is enough to 
prove the existence of the wave operators   
 \begin{equation*}
    M_\pm:= \slimpm \e^{\im th_0} j \e^{-\im t h_{+,0}}. 
  \end{equation*}
 }
 For any $\phi \in \mathcal S$, we let $\phi_\odd$ denote the odd part of 
 $\phi$, given by $\phi_\odd (x) = \phi(x) - \phi(-x)$, and define 
  \begin{equation*}
    \mathcal D_0 := 
    \bigset {\phi_\odd} {\phi \in \mathcal S(\R), \mathcal F \phi \in \Cci{\R
        \setminus \{0\}}}.
  \end{equation*}
  As $\mathcal D_0 \restriction \R_+$ is dense in $\Lsqr{\R_+}$ it is enough
  to prove the existence of $M_\pm (\phi_\odd \restriction \R_\pm)$ 
  for $\phi_\odd \in \mathcal D_0$. Defining
  \begin{equation*}
    \phi_{\odd, t}(x)
    := \e^{-\im th_0} \phi_\odd(x)
    = \bigl(\e^{-\im t h_0} \phi \bigr)(x)
    -  \bigl(\e^{-\im t h_0} \phi \bigr)(-x),
    \quad \phi_\odd \in \mathcal D_0, 
  \end{equation*}
  we have $ \im \frac \partial {\partial t} \phi_{\odd, t} = - \frac
  {\partial^2}{\partial x^2} \phi_{\odd, t}$ and $\phi_{\odd, t}(0)=0$
  which implies
  \begin{equation*}
    \phi_{\odd, t}= \e^{-\im t h_{0,+}} \phi_\odd, \qquad x > 0. 
  \end{equation*}
  By Duhamel's formula, to prove the existence of $M_+ \phi_\odd$ it
  is enough to show that 
  \begin{multline*}
    \int_0^\infty \bignorm[\Lsqr \R] 
          {\e^{\im th_0} (h_0 j-j h_{0,+}) \e^{-\im t h_{0,+}} 
          \phi_\odd} \dd t \\
   = \int_0^\infty \Bignorm[\Lsqr \R]
     {\Bigl( j''+ 2 j' \frac \dd {\dd x} \Bigr) 
      \bigl( (\e^{-\im t h_0} \phi)(x)- (\e^{-\im t h_0} \phi)(-x) \bigr)}
                \dd t  < \infty, 
  \end{multline*}
  but this is immediate from~\eqref{eq:c.2}.  The case of $M_-
  \phi_\odd$ follows in the same way. $\Theta_\pm$ are
  isometric since $\e^{\im t h_0}$ and $\e^{-\im th_{0,+}}$
  are unitary.

  Set
  \begin{equation*}
    \hat {\mathcal D}_\pm
    := \bigset {\phi_\pm \in \Lsqr \R} {\mathcal F \phi_\pm \in \Cci{\R_\pm}} 
  \end{equation*}
  and
  \begin{equation*}
    \phi_{\pm, t}(x)
    := \bigl(\e^{-\im t h_0} \phi_\pm \bigr)(x)
    -  \bigl(\e^{-\im t h_0} \phi_\pm \bigr)(-x), 
    \quad \phi_\pm \in \hat{\mathcal D}_\pm.
  \end{equation*}
  Then, as above,
  \begin{equation*}
    \phi_{\pm, t} = \e^{-\im t h_{0,+}} \phi_{\pm,0},
    \quad \text{where} \quad 
    \phi_{\pm,0}(x)= \phi_\pm(x)-\phi_\pm(-x), \quad
    \phi_\pm \in \hat{\mathcal D}_\pm.
  \end{equation*}
  It follows from~\eqref{eq:c.2} that 
    $\lim_{t \rightarrow \pm \infty} 
       \bignorm[\Lsqr \R] {\e^{-\im th_0} \phi_\pm 
             - \iota_+ \e^{-\im t h_{0,+}} \phi_{\pm,0}}=0$
 and we see that 
  \begin{equation}
   \label{eq:c.19}
    \phi_\pm = \Theta^+_\pm \phi_{\pm, 0},  
    \quad   \forall \phi_\pm \in \hat{\mathcal D}_\pm.
  \end{equation}
  Then, $\hat{\mathcal D}_\pm \subset {\rm Ran}\Theta^+_\pm$, and as
  the $\Theta^+_\pm$ are isometric,
  \begin{equation}
    \label{eq:c.20}
    \mathcal H_\pm= \clo{\hat{\mathcal D}_\pm} 
    \subset  \ran  \Theta^+_\pm.
  \end{equation}
  We prove in the same way, using~\eqref{eq:c.2}, that 
   $ \hat{\mathcal D}_\pm
    \subset \bigl(\ran \Theta^+_\mp \bigr)^\perp$,
  and then that 
  \begin{equation}
    \label{eq:c.22}
    \mathcal H_\pm
    \subset \bigl(\ran \Theta^+_\mp \bigr)^\perp.
  \end{equation}
  Finally, as $\Lsqr \R= \mathcal H_- \oplus  \mathcal H_+$ and 
  $\Lsqr \R= \ran \Theta^+_\pm \oplus \bigl(\ran \Theta^+_\pm \bigr)^\perp$,
  equations~\eqref{eq:c.20} and~\eqref{eq:c.22} imply that 
  $ \mathcal H_\pm= \ran \Theta^+_\pm$.
\end{proof}

Let $w$ be a real-valued, bounded function defined on $\mathbb R$
satisfying
\begin{equation*}
  \abs{w(x)} \le  C(1+|x|)^{-1-\beta}
\end{equation*}
for some constants $C \ge 0$, $\beta > 0$, and let
\begin{equation*}
  h:= h_0+w.
\end{equation*}
The operator $h$ is self-adjoint in $\mathcal H$ with domain $\dom h=
\Sob [2] \R$.  We consider the wave operators
\begin{equation*}
  \Omega^+_\pm:= \slimpm \e^{\im th} \iota_+ \e^{-\im t h_0},
  \quad\text{and}\quad
  \Omega^-_\pm:= \slimpm \e^{\im th} \iota_- \e^{-\im t h_0}.
\end{equation*}
\begin{theorem}
  \label{thm:c.2}
  The wave operators $\Omega^+_\pm, \Omega^-_\pm$ exist and are
  partially isometric. Moreover, their initial subspaces,
  respectively, $\HS^+_{\ini, \pm}$ and $\HS^-_{\ini,
    \pm}$, are given by
  \begin{equation}
    \label{eq:c.29} 
    \HS^+_{\ini, \pm} = \HS_\pm
    \quad\text{and}\quad
    \HS^-_{\ini, \pm} = \HS_\mp.
  \end{equation}
  Furthermore,
  \begin{equation*}
    \ran \Omega^+_+ \ne \ran \Omega^+_-
    \quad\text{and}\quad
    \ran \Omega^-_+ \ne \ran \Omega^-_-.
  \end{equation*}
\end{theorem}
\begin{proof}
  The existence of $\Omega^+_\pm \phi$ and $\Omega^-_\pm\phi$ for $
  \mathcal F \phi \in \Cci{\R \setminus \{0\}}$ follows upon replacing
  $\iota_\pm$ by multiplication with a \emph{smooth} cut-off function
  together with Duhamel's formula and equation~\eqref{eq:c.2}.  We
  omit the details.  This proves the existence of the wave operators
  in a dense set, and by continuity in $\HS$.

  By~\eqref{eq:c.2}, we have for $\phi_+ \in \hat{\mathcal D}_+$,
 \begin{equation*}
   \bignorm{\iota_+ \e^{-\im t h_0} \phi_+} \to
    \begin{cases}
      \norm {\phi_+}, & \text{as $t \to \infty$},\\
      0, & \text{as $t \to -\infty$},
    \end{cases}
  \end{equation*}
  and similarly for $\phi_-  \in \hat{\mathcal D}_-$. It follows that
  \begin{equation}
    \label{eq:c.32}
    \norm{\Omega_\pm^+ \phi} =
    \begin{cases}
      \norm \phi , & \text{if $\phi \in  \hat{\mathcal D}_\pm  $},\\
      0, & \text{if $ \phi \in \hat{\mathcal D}_\mp$}.
    \end{cases}
  \end{equation}
  Since the subspaces $\hat{\mathcal D}_\pm$ are dense in $\mathcal
  H_\pm$, the first equality in~\eqref{eq:c.29} follows
  from~\eqref{eq:c.32}.  We prove the second equality
  in~\eqref{eq:c.29} in the same way.

  We now show that $\ran \Omega^+_+ \ne \ran \Omega^+_-$.  The
  intuition behind our proof is as follows.  Consider an incoming
  state with large negative mean velocity and assume that this state
  is localized near $+\infty$ for large negative times.  Such a state
  will be in the range of $\Omega^+_-$.  As time increases it will
  propagate to the left, but as it has large velocity it will ``go
  across'' the potential $w$ and will travel to $-\infty$ as time goes
  to infinity (with only a small reflected part of the state
  travelling to $+\infty$).  %
  Since this state has a non-trivial component localized near
  $-\infty$ for large positive times, it cannot be in the range of
  $\Omega^+_+$.

  Let us consider the following asymptotic state with high negative
  velocity,
  \begin{equation*}
    \phi^-_v:= \e^{-\im vx} \phi_0, 
    \quad\text{with $\mathcal F\phi_0 \in \Cci \R.$}
  \end{equation*}
  As
  \begin{equation*}
    \bigl(\mathcal F \phi^-_v \bigr)(k)
    = (\mathcal F\phi_0)(k+v),
  \end{equation*}      
  this state will have large negative velocity if $v>0$ is taken large
  enough.

  Let us introduce the wave operators 
  \begin{equation*}
    \Omega_\pm := \slimpm  \e^{\im t h} \e^{-\im t h_0}.
  \end{equation*}
  The existence of these wave operators follows from~\eqref{eq:c.2}
  and Duhamel's formula;  they are also complete (cf.~\cite{enss:78,
    simon:79c}).  Defining $\psi^-_v:= \Omega_-\, \phi^-_v$, we have 
  $\slimm \iota_- \e^{-\im t h_0} \phi^-_v= 0$,
  by~\eqref{eq:c.2},   and then
  \begin{equation}
    \label{eq:c.38}
    \psi^-_v := \slimm  \e^{\im th} \e^{-\im t h_0} \phi^-_v
    =           \slimm \e^{\im th} \iota_+ \e^{-\im t h_0} \phi^-_v
    = \Omega^+_ -\phi^-_v.
  \end{equation}
  This proves that $\psi^-_v \in \ran \Omega^+_-$.

  It follows from Corollary~2.3 of~\cite{enss-weder:95} and the
  intertwining relation $\e^{-\im th} \Omega_-=\Omega_- \e^{-\im t
    h_0}$ that
  \begin{equation*}
    \bignorm[\Lsqr \R] {\e^{-\im th} \psi^-_v 
                   - \e^{-\im t h_0} \phi^-_v}
    = \mathcal O(1/v),
  \end{equation*}
  uniformly in $t \in \R$.  Hence, by~\eqref{eq:c.2},
  \begin{align*}
    \liminf_{t \rightarrow \infty}
    \bignorm[\Lsqr \R] {\iota_- \e^{-\im th} \psi^-_v}
    & \ge \liminf_{t \rightarrow \infty}
      \bignorm[\Lsqr \R] {\iota_- \e^{-\im th_0} \phi^-_v}
      - \mathcal O(1/v)\\
    & = \bignorm[\Lsqr \R] {\phi^-_v} - \mathcal O(1/v)
    \ge \frac12 \bignorm[\Lsqr \R] {\phi^-_v} >0,
  \end{align*}
   for $v$ large enough, which proves that $\psi^-_v \notin \ran
  \Omega^+_+$.

  The asymptotic state $\phi^+_v$, defined as  
 $\phi^+_v(x):= \e^{\im vx} \phi_0$, with $\mathcal F\phi_0 \in \Cci \R$,
  has large positive velocity if $v>0$ is large enough. We define
  $\psi^+_v:= \Omega_+\, \phi^+_v$  and prove as above that  
  $\psi^+_v = \Omega^-_- \phi^+_v$, 
  so that $\psi^+_v \notin \ran \Omega^-_+$, which proves that
   $\ran \Omega^-_+ \ne \ran \Omega^-_-$.
\end{proof}
 We need yet another set of wave operators. 
Let us denote by $\Lambda^+_\pm$ and $\Lambda^-_\pm$ the wave 
operators
\begin{equation}
  \Lambda^+_\pm:= \slimpm \e^{\im th} \iota_+ \e^{-\im t h_{+,0}}
  \quad\text{and}\quad
  \Lambda^-_\pm:= \slimpm \e^{\im th} \iota_- \e^{-\im t h_{-,0}}.
\end{equation}
\begin{lemma}
  \label{lem:c.3}
  The wave operators $\Lambda^+_\pm$ and $\Lambda^-_\pm$ exist, are
  isometric and, furthermore,
  \begin{equation}
    \ran \Lambda^+_+ \ne \ran \Lambda^+_-,
    \quad\text{and}\quad
    \ran \Lambda^-_+ \ne \ran \Lambda^-_-.
  \end{equation}
\end{lemma}
\begin{proof}
  By the chain rule
  \begin{equation}
     \label{eq:c.46}
    \Lambda^+_\pm = \Omega^+_\pm  \circ  \Theta^+_\pm
    \quad\text{and}\quad
    \Lambda^-_\pm= \Omega^-_\pm  \circ  \Theta^-_\pm.
  \end{equation}
  Then, the lemma follows by \Lem{c.1} and \Thm{c.2}.
\end{proof}

\begin{remark}
  \label{rem:c.4}
  The above proofs yield explicit examples of asymptotic states that
  belong to $\ran \Lambda^\pm_-$ and do not belong to $\ran
  \Lambda^\pm_+$.  In fact, it follows from~\eqref{eq:c.19} that
  \begin{equation}
     \label{eq:c.47}
    \phi^\pm_v = \Theta^+_\pm \phi^\pm_{v,0}, 
    \quad\text{where}\quad
    \phi^\pm_{v,0}:= \phi^\pm_v(x)- \phi^\pm_v(-x).
  \end{equation}
  We prove in the same way that
  \begin{equation*}
    \phi^\pm_v= \Theta^-_\mp \phi^\pm_{v,0}.
  \end{equation*}
  By~\eqref{eq:c.38}, \eqref{eq:c.46} and~\eqref{eq:c.47} we have
  \begin{equation}
    \psi^-_v = \Omega^+_- \Theta^+_- \phi^-_{v,0}
    = \Lambda^+_- \phi^-_{v,0}.
  \end{equation}
  Moreover as $\psi^-_v \notin \ran \Omega^+_+$ it follows
  from~\eqref{eq:c.46} that $\psi^-_v \notin \ran \Lambda^+_+$.  We
  prove in the same way 
 that
  \begin{equation}
    \label{eq:c.50}
    \psi^+_v = \Omega^-_- \Theta^-_- \phi^+_{v,0}
    = \Lambda^-_- \phi^+_{v,0}.
  \end{equation}
  Moreover as $\psi^+_v \notin \ran \Omega^-_+$ it follows from
  \eqref{eq:c.46} that $\psi^{+}_{v} \notin \ran \Lambda^{-}_+$.

  By the definition of $\psi^+_v = \Omega_+\, \varphi^+_v$,
  equation~\eqref{eq:c.50}, and Corollary~2.3 of~\cite{enss-weder:95},
  we have
  \begin{equation}
    \label{eq:c.51}
    \Lambda^-_- \phi^ +_{v,0}
    =\Omega_ + \phi^+_v = \phi^+_v + \mathcal O\Bigl(\frac 1v \Bigr).
  \end{equation}
\end{remark}

As in the proof of~\eqref{eq:c.38} we prove that 
\begin{equation}
  \label{eq:c.52}
  \Omega^+_+ \phi^+_v= \Omega_+ \phi^+_v.
\end{equation}
Then, as above, it follows from \eqref{eq:c.46}, \eqref{eq:c.47},
\eqref{eq:c.52} and Corollary 2.3 of~\cite{enss-weder:95} that
\begin{equation}
  \label{eq:c.53}
  \Lambda^+_+ \phi^+_{v,0} = \phi^+_v+ \mathcal O \Bigl(\frac 1v \Bigr).
\end{equation}
By~\eqref{eq:c.51} and~\eqref{eq:c.53} and since $\Lambda^+_+$ is
isometric, we have
\begin{equation}
  \label{eq:c.54}
  \bigl(\Lambda^+_+\bigr)^* \Lambda^-_- \phi^+_{v,0}
  = \phi^+_{v,0}+ \mathcal O \Bigl(\frac 1v \Bigr).
\end{equation}
 Then 
\begin{equation}
  \label{eq:c.55}
 \bigl(\Lambda^+_+\bigr)^* \Lambda^-_- \phi^+_{v,0} \ne 0, \qquad v \gg 1,
 \end{equation}
proving that the $-$ and the $+$ channels are open to each other.
Note that $\phi^+_{v,0}$ is a high-velocity asymptotic state, coming
in from the left for large negative times; as time increases, it
travels to the right and is transmitted through the potential---with
a reflected part that is very small for $v$ large---and it goes to
$+\infty$ as time tends to $+\infty$.

%
\section{Lower bounds for the {\colour[A3] injectivity radius}}
\label{app:inj.rad}
%

The main purpose of this final appendix is to acquaint the reader with
a comparison result of M\"uller and
Salomonsen~\cite[Prop.~2.1]{mueller-salomonsen:07} for the injectivity
radius of two complete manifolds with bounded curvature.  We use their
result in a two-step process to deal with warped products (and
perturbed warped products) where the curvature does not obey a global
curvature bound, as happens for shrinking ends of dimension $n \ge 3$.
In some sense, we attempt to produce local versions
of~\cite[Prop.~2.1]{mueller-salomonsen:07}; note that the proof given
in their paper uses some non-local arguments and completeness is an
issue.  The basic idea is to extend finite sections of our manifolds
to complete manifolds with two cylindrical ends and to obtain a lower
bound for the {\colour[A3] injectivity radius} of these extended manifolds by
comparison with a straight cylinder.

We begin with the comparison theorem for the {\colour[A3] injectivity
  radius} of M\"uller and Salomonsen~\cite{mueller-salomonsen:07}. Our
statement displays a lower bound with explicit constants taken from
their proof; cf.~eqn.~\eqref{eq:d.2}.

\begin{proposition}[{\cite[Prp.~2.1]{mueller-salomonsen:07}}]
  \label{prp:d.1}
  Let $M$ denote a smooth $n$-dimensional manifold.  Suppose that the
  Riemannian manifolds $\Mscr_0 := (M,g_0)$ and $\Mscr_1 : = (M,g_1)$
  are complete with quasi-isometric metrics $g_0$ and $g_1$, i.e.,
  \begin{equation}
    \label{eq:d.1}
    \eta g_0 \le g_1 \le \eta^{-1} g_0,
  \end{equation}
  for some constant $\eta \in (0,1]$. Furthermore, suppose that the
  sectional curvature of $\Mscr_0$ and $\Mscr_1$ is bounded (in
  absolute value) by some constant $K \ge 0$. Let $\inj_{\Mscr_0}(x)$
  and $\inj_{\Mscr_1}(x)$ denote the {\colour[A3] injectivity radius}
  of $\Mscr_0$ and $\Mscr_1$, respectively, at the point $x \in M$.
  We then have
  \begin{equation}
    \label{eq:d.2}
    \inj_{\Mscr_1}(x)
    \ge  \frac1{2} \min \Bigl\{\frac{\eta^2 \pi}{\sqrt K}, 
                    \eta \inj_{\Mscr_0}(x) \Bigr\},
    \qquad  x \in M.
  \end{equation}
\end{proposition}

We next use the above comparison theorem to obtain lower bounds for
the {\colour[A3] injectivity radius} of manifolds $\Mscr = (M_+, g)$
where $M_+ = \R_+ \times \Sphere^{n-1}$, $\R_+ = (0,\infty)$, and $g$
is a warped product metric generated by a function $r$. We require a
local bound on the variation of $r$, cf.~eqn.~\eqref{eq:d.3}. From
this point on, we restrict our attention to the case $n \ge 3$. The
corresponding results in the case $n=2$ are obtained by dropping the
term $(1+{\dot r}(t)^2)/r(t)^2$ in the definition of $\kappa$ in
eqn.~\eqref{eq:d.5} below etc.

\begin{lemma}
  \label{lem:d.2}
  Let $\map r {\R_+} {\R_+}$, $r \in \Cont[2] {\R_+}$, satisfy the
  condition
  \begin{equation}
    \label{eq:d.3}
    \frac 1 m r(s_0)
    \le r(s) \le
    m r(s_0), \qquad s \in [s_0 -2, s_0 + 2], 
  \end{equation}
  for all $s_0 > 2$, where $m \ge 1$ is a constant. Let $\Mscr_+ =
  (M_+, \dd s^2 + r(s)^2 g_{\Sphere^{n-1}})$ with $ g_{\Sphere^{n-1}}$
  denoting the standard metric on $\Sphere^{n-1}$.
  
  Then $\inj_{\Mscr_+}(s)$, the {\colour[A3] injectivity radius} of
  $\Mscr_+$ at the points $(s,\omega) \in M_+$, satisfies the lower
  bound
  \begin{equation}
    \label{eq:d.4}
    \inj_{\Mscr_+}(s) \ge C_0\, \min \Bigl\{ \frac 1 {\sqrt{\kappa(s)}},
      r(s) \Bigr\},
    \qquad s > 2,
  \end{equation}
  where $C_0 > 0$ is a constant that is independent of $s$, and
  \begin{equation}
    \label{eq:d.5}
    \kappa(s) :=
    \max_{|s-t| \le 2} \max\Bigl\{ 
                \frac{|\ddot r(t)|}{r(t)},
                \frac{1 + \dot r(t)^2}{r(t)^2}, 1 \Bigr\}. 
    \end{equation}
  
\end{lemma}

\begin{proof}
  We fix a function $\phi \in \Cci{-2,2}$ satisfying $0 \le \phi \le
  1$ and $\phi(x) = 1$ for $-1 \le x \le 1$. Let $c_\phi :=
  \max\{\norm[\infty]{\phi'}, \norm[\infty]{\phi''}\}$.  Furthermore,
  let $\phi_{s_0} = \phi(\cdot - s_0)$ and define
  \begin{equation}
    \label{eq:d.6}
    \rho_{s_0}(s)
    := \phi_{s_0}(s) r(s) + (1 - \phi_{s_0}(s)) r(s_0),
    \qquad s \in \R, 
  \end{equation}
  for $s_0 \ge 2$; note that $\rho_{s_0}$ is defined on all of $\R$.
  We have
  \begin{equation}
    \label{eq:d.7}
    \frac 1 m  r(s_0) \le \rho_{s_0}(s) \le m r(s_0),
    \qquad s \in \R. 
  \end{equation}
  Let $\Mscr_{s_0} = (M, g_{s_0})$ with $M := \R \times \Sphere^{n-1}$
  and $g_{s_0} = \dd s^2 + \rho_{s_0}^2 g_{\Sphere^{n-1}}$.  Then
  $\Mscr_{s_0}$ is complete and, by~\eqref{eq:d.7}, $\Mscr_{s_0}$ is
  quasi-isometric (with a constant $\eta := 1/m \in (0,1)$) to $\wt
  \Mscr_{s_0} = (M, \dd s^2 + r(s_0)^2 g_{\Sphere^{n-1}})$, a cylinder
  of constant radius $r(s_0)$.

  In order to obtain a curvature bound for $\Mscr_{s_0}$ we first note
  that $|r(s) - r(s_0)| \le r(s)$ if $r(s) \ge r(s_0)$ while $|r(s) -
  r(s_0)| \le r(s_0) \le m r(s)$ if $r(s) < r(s_0)$; in both cases the
  estimate $|r(s) - r(s_0)| \le m r(s)$ is valid.  For $|s-s_0| \le 2$
  the derivatives of the function $\rho_{s_0}$ satisfy
  \begin{equation*}
    |\dot \rho_{s_0}(s)|
    \le \norm[\infty]{\phi'} |r(s) - r(s_0)| + |\dot r(s)| 
    \le c_\phi m r(s) + \sqrt{\kappa(s_0)} r(s) 
    \le c_1 \sqrt{\kappa(s_0)} r(s_0), 
  \end{equation*}
  by definition of $c_\phi$ and $\kappa$. Similarly, we have for
  $|s-s_0| \le 2$
  \begin{align*}
     |\ddot \rho_{s_0}(s)|
     &\le  \norm[\infty]{\phi''} |r(s) - r(s_0)| 
                 + 2 \norm[\infty]{\phi'} |\dot r(s)| + |\ddot r(s)| \\
     &\le c_\phi m r(s) + 2 c_\phi \sqrt{\kappa(s_0)}r(s)  
             + \kappa(s_0) r(s) \le c_2 \kappa(s_0) r(s_0). 
  \end{align*}
  We now find for $|s-s_0| \le 2$  
  \begin{equation*}
    \frac{|1 - \dot\rho_{s_0}(s)^2|}{\rho_{s_0}(s)^2}
    \le \frac{m^2}{r(s_0)^2} + c_1^2 \kappa(s_0) \le c_2 \kappa(s_0), 
   \qquad 
    \frac{|\ddot \rho_{s_0}(s)|}{\rho_{s_0}(s)} \le c_3 \kappa(s_0). 
  \end{equation*}
  Thus the sectional curvature of $\Mscr_{s_0}$ is bounded by $C_1
  \kappa(s_0)$ with some constant $C_1 \ge 0$ which is independent of
  $s_0 \ge 2$.

  We may now apply \Prp{d.1} to $\Mscr_{s_0}$ and $\wt \Mscr_{s_0}$
  with $\eta := 1/m$ and $K := C_1 \kappa(s_0)$ to obtain
  \begin{equation*}
    \inj_{\Mscr_{s_0}} \ge
    C_0\min \Bigl\{ \frac{1}{\sqrt{\kappa(s_0)}}, r(s_0) \Bigr\} . 
  \end{equation*}
  Since the manifolds $\Mscr_+$ and $\Mscr_{s_0}$ have the same metric
  for $|s - s_0| \le 1$ it is clear that $\inj_{\Mscr_+}(s_0) \ge
  \min\{\hbox{\rm inj}_{\Mscr_{s_0}}, 1\}$, and the desired result
  follows.
\end{proof}

The idea of proof used in obtaining \Lem{d.2} can easily be
generalized to perturbations of a warped product metric. For
simplicity we work here with $M = \R \times \Sphere^{n-1}$.  The
functions $\phi_{s_0} \in \Cci{-2,2}$ are as above.

\begin{proposition}
  \label{prp:d.3}
  Let $M = \R \times \Sphere^{n-1}$, let $r \colon \R \to \R_+$ be a
  $\Contspace [2]$-function satisfying the estimate~\eqref{eq:d.3} for
  all $s_0 \in \R$.  Also define $\kappa_0 = \kappa_0(s)$ as in
  eqn.~\eqref{eq:d.5}, but now for all $s \in \R$. Let $g_0$ denote
  the metric $\dd s^2 + r(s)^2 g_{\Sphere^{n-1}}$.

  Let $g$ denote another metric on $M$ that satisfies the following
  conditions:
  \begin{enumerate}
  \item
    \label{prp.d.3.i}
    There exists a constant $m \ge 1$ such that, for all $s_0 \in
    \R$,
    \begin{equation}
      \label{eq:*}
      \frac{1}{m} g_0(s_0)
      \le g(s)
      \le m g_0(s_0), \qquad s \in [s_0-2, s_0 + 2]. 
    \end{equation}

  \item
    \label{prp.d.3.ii}
    There exists a (continuous) function $\kappa$ on
    $\R$ such that the sectional curvature of
    \begin{equation*}
      \Mscr_{s_0} := (M, \phi_{s_0}(s) g(s) + (1 - \phi_{s_0}(s)) g_0(s_0))
    \end{equation*}
    is bounded (in absolute value) by $\kappa(s_0)$, for all $s_0 \in \R$.
  \end{enumerate}
  Then $\inj_{\Mscr}(s,\omega)$, the {\colour[A3] injectivity radius} of
  $\Mscr = (M,g)$ at $(s,\omega) \in M$, obeys the lower bound
  \begin{equation}
    \label{eq:**}
    \inj_\Mscr(s,\omega)
    \ge C\, \min \Bigl\{ \frac{1}{\sqrt{\kappa_0(s)}}, 
                        \frac{1}{\sqrt{\kappa(s)}}, r(s) \Bigr\},
    \qquad s \in \R,
  \end{equation}
  for some positive constant $C$. 
\end{proposition}

\begin{proof}
  For $s_0 \ge 2$ given, let us consider the (complete) manifold
  \begin{equation*}
    \Mscr_{0,s_0} := (M, \phi_{s_0}(s) g_0(s) + (1 - \phi_{s_0}(s))
    g_0(s_0)).
  \end{equation*}
  Note that the manifolds $\Mscr_{0,s_0}$ and $\Mscr_{s_0}$ are
  identical outside the interval $[s_0-2,s_0+2]$. Furthermore, both
  have a metric that is independent of $s$ outside of $[s_0-2,s_0+2]$.
  It is immediate from assumption~\itemref{prp.d.3.i} that
  $\Mscr_{0,s_0}$ and $\Mscr_{s_0}$ are quasi-isometric with a constant 
  $\eta = 1/m$.
  As in the proof of \Lem{d.2}, one shows that the sectional curvature
  of $\Mscr_{0,s_0}$ is bounded by $C_1 \kappa_0(s_0)$, with a
  constant $C_1$ that is independent of $s_0 \in \R$.
  As for $\Mscr_{s_0}$, assumption~\itemref{prp.d.3.ii} says that the
  sectional curvature of $\Mscr_{s_0}$ is bounded by $\kappa(s_0)$.
  By \Lem{d.2}, the injectivity radius of $\Mscr_{0,s_0}(s)$ is
  bounded below by $C_0 \min \{(\kappa_0(s))^{-1}, r(s) \}$ and thus
  \Prp{d.1} yields the estimate 
  \begin{equation*}
    \inj_{\Mscr_{s_0}}(s_0,\omega) \ge \min \Bigl\{ \frac{\eta^2
        \pi}{\sqrt{\kappa(s_0)}},
     \frac{\eta C_0}{\sqrt{\kappa_0(s_0)}}, \eta C_0 r(s_0) \Bigr\}
  \end{equation*}
  and the desired estimate follows.
\end{proof}

\begin{remark}
  \label{rem:d.4}
  In concrete applications the required bound $\kappa$ on the
  sectional curvature can be obtained by direct calculation in terms
  of the metric (cf., e.g.,~\cite[p.~204~ff.]{oneill:83}).
\end{remark}

\begin{remark} 
  \label{rem:d.5}
  It is clear from the assumptions of \Lem{d.2} and \Prp{d.3} (in
  particular, eqns.~\eqref{eq:d.3} and \eqref{eq:*}) that the lower
  bounds~\eqref{eq:d.4} and~\eqref{eq:**} for the radius of
  injectivity also hold for the \emph{homogenized} radius of
  injectivity $\iota_\Mscr$ as defined in
  equation~\eqref{eq:unif.inj.rad}, possibly with a smaller positive
  constant. In both cases our assumptions imply that the Ricci
  curvature satisfies the lower bound required in \Prp{lower.bd.rh} in
  the form $\Ric^-_\Mscr(s,\omega) \ge - (n-1) \kappa(s_0)$ for $|s -
  s_0| < 2$ since the Ricci curvature at a point $x$ is the sum of the
  sectional curvatures of any $n-1$ orthogonal non-degenerate planes
  through $x$ (\cite[p.~88]{oneill:83}).  Therefore, \Prp{lower.bd.rh}
  yields that the homogenized harmonic radius $\iota_\Mscr$ obeys
  lower bounds analogous to the lower bounds for $\inj_\Mscr$; in
  other words, the function $r_0(x)$ of \Prp{lower.bd.rh} can be read
  off from the right hand side of eqns.~\eqref{eq:d.4} or
  \eqref{eq:**}.
\end{remark}   


\subsection*{Acknowledgments}

R.\ Weder and R.\ Hempel thank Volker En{\ss}, Aachen, for the
suggestion to study the openness of channels in a perturbational
setting.  R.~Hempel is most grateful to Brian Davies, London, for
fruitful discussions and suggestions which led to substantial
improvements. He would also like to thank the Isaac Newton Institute,
Cambridge, where a part of this work was done, for its hospitality.
The visit of R.~Hempel to the Isaac Newton Institute was supported by
the Program ``Spect'' of the European Science Foundation, Strasbourg.
O.\ Post kindly acknowledges the financial support given by the
SFB~647 ``Space---Time---Matter'' at the Humboldt University Berlin.
O.~Post would also like to thank Qinglan Xia for pointing his
attention to the chapter on quasi-metrics in the book of
Heinonen~\cite{heinonen:01}.  R.~Weder thanks Patrick Joly for his
kind hospitality at the project POEMS of the Institut National de
Recherche en Informatique et en Automatique (Inria),
Paris-Rocquencourt, where a part of this work was done. This research
was partially supported by CONACYT under Project CB-2008-01-99100.
R.~Weder is a Fellow of the Sistema Nacional de Investigadores.
{\colour Last but not least we would like to thank the unknown
  referees for various suggestions which we have been happy to
  incorporate into the paper.}

\newcommand{\etalchar}[1]{$^{#1}$}
\providecommand{\bysame}{\leavevmode\hbox to3em{\hrulefill}\thinspace}

\end{document}